\documentclass[11pt]{article}

\usepackage[margin=1in]{geometry}
\usepackage{cellspace}
\usepackage{hyperref}
\usepackage{times}  
\usepackage{mathpazo}
\usepackage{amssymb,amsmath,amsthm}
\usepackage{kbordermatrix}
\usepackage{epsfig}
\usepackage{tikz}
\usetikzlibrary{fit}


\renewcommand{\kbldelim}{\left(}
\renewcommand{\kbrdelim}{\right)}

\usetikzlibrary{positioning}
\usetikzlibrary{matrix}
\usepackage{forest}
\usepackage{xcolor}

\newtheorem{theorem}{Theorem}

\newtheorem{definition}[theorem]{Definition}

\newtheorem{claim}[theorem]{Claim}
\newtheorem{lemma}[theorem]{Lemma}

\newtheorem{corollary}[theorem]{Corollary}

\newtheorem{open}{Open problem}
\newcommand{\qedsymb}{\hfill{\rule{2mm}{2mm}}}
\renewenvironment{proof}[1][]{\begin{trivlist}
\item[\hspace{\labelsep}{\bf\noindent Proof#1:\/}] }{\qedsymb\end{trivlist}}

\def\R{\mathbb{R}}

\def\Q{\mathbb{Q}}

\newcommand{\rank}{\mathop{\mathrm{rank}}}
\newcommand{\Rbool}{{\rank}_\mathbb{B}}
\newcommand{\Rbin}{{\rank}_{\mathrm{bin}}}
\newcommand{\Rreal}{{\rank}_\mathbb{R}}
\newcommand{\Rnon}{{\rank}_{+}}

\renewcommand{\kbldelim}{\left(}
\renewcommand{\kbrdelim}{\right)}

\newcommand{\cellnode}[2][]{%
  \tikz[remember picture,overlay,baseline=(#2.base)]%
    \node[inner sep=0pt,outer sep=0pt,minimum size=0pt,#1] (#2) {};%
}

\providecommand{\keywords}[1]{\textbf{\textit{Key Words---}} #1}

\begin{document}

\title{Mathematical and computational perspectives on the Boolean and binary rank and their relation to the real rank }
\author{Michal Parnas\thanks{michalp@mta.ac.il.} \\
School of Computer Science \\
   The Academic College of Tel-Aviv-Yaffo \\
}
\date{}

\maketitle

\begin{abstract}
This survey provides a comprehensive overview of the study of the binary and Boolean rank from both a mathematical and a computational perspective, with particular
emphasis on their relationship to the real rank.
We review the basic definitions of these rank functions and present the main alternative formulations of the binary and Boolean rank,
together with their computational complexity and their deep connection to the field of communication complexity.
We summarize key techniques used to establish lower and upper bounds on the binary and Boolean rank,
including  methods from linear algebra, combinatorics and graph theory, isolation sets,
the probabilistic method, kernelization, communication protocols and the query to communication lifting technique.
Furthermore, we highlight the main mathematical properties of these ranks in comparison with those of the real rank,
and discuss several non-trivial bounds on the rank of specific families of matrices.
Finally, we present algorithmic approaches for computing and approximating these rank functions,
such as parameterized algorithms, approximation algorithms, property testing and approximate Boolean matrix factorization (BMF).
Together, the results presented outline the current theoretical knowledge in this area and suggest directions for further research.


\end{abstract}

\newpage

\tableofcontents
\newpage

\section{Introduction}

The {\em rank} of a matrix over the real numbers is a fundamental concept in mathematics which originated from the area of linear algebra, and in particular from the study of linear equations.
The solution of linear equations is a practical problem which was studied already by Egyptian mathematicians as early as 1650 B.C.,
when they considered problems which can be described as a linear equation with one unknown.
There is also evidence from about 300 B.C. of the ability to solve equations with two unknowns.
But it was only in the 19'th century that the theory of matrices evolved into its modern form.
Works by Cauchy, Cayley, Sylvester, Frobenius, Hamilton and other mathematicians from that period, formalized the notion of a matrix, the operations which can be performed on it,
its properties, the concepts of vectors and determinants, as well as defined the rank of a matrix (see a survey by Corry~\cite{Corry}).

Since then the rank of a matrix found numerous applications, which are not related to linear equations, for measuring algebraic and combinatorial properties in mathematics and in computer science.
For example, let $G$ be a directed graph with $n$ vertices and $c$ connected components whose incidence matrix is $M$. Then it can be shown that the rank of $M$ is equal to $n-c$.
See a book by Bapat~\cite{bapat2010graphs} for a proof of  this claim and other interesting connections between the rank and graph theory.
The rank is also a crucial measure in fields such as communication complexity, information theory, learning algorithms,  error correcting codes, circuit complexity and more.

Furthermore, there are variations of the rank concept with interesting connections between them,
such as the term rank studied by Ryser~\cite{ryser1958term}, the min-rank defined by Haemers~\cite{haemers1979some},
the sign rank (see~\cite{razborov2010sign} and references within), the approximate $\epsilon$-rank studied by Alon, Lee, Shraibman and Vempala~\cite{alon2013approximate},
the rank of a spanoid introduced by  Dvir, Gopi, Gu and Wigderson~\cite{ dvir2020spanoids},
and recently the $\pm$-rank of $0,\pm 1$ matrices defined by Brualdi and Dahl~\cite{BRUALDI2026318}.

The rank of a matrix $M$ over the real numbers, $\R$, is usually defined as the maximal number of linearly independent rows or columns of $M$.
To introduce the topic of this survey we consider the following equivalent definition of the real rank, which will be denoted here by $\Rreal(M)$.
Given an $n\times m$ matrix $M$ over $\R$, its rank, $\Rreal(M)$, is the minimal integer $d$ for which there exist real matrices $A$ and $B$ of size $n \times d$ and $d \times m$,
respectively, such that $M = A \cdot B$, where the operations are over $\R$.

In a similar way, consider the following rank functions defined over a semi-ring (see e.g. Gregory and Pullman~\cite{GregoryPullman}):
\begin{itemize}
\item
The {\em non-negative rank} of a non-negative matrix $M$  of size $n\times m$, denoted by $\Rnon(M)$,
is the minimal $d$ for which there exist non-negative matrices $A$ and $B$ of size $n \times d$ and $d \times m$ respectively,
such that $M = A \cdot B$, where the operations are over the reals.
\item The {\em binary rank} of a $0,1$ matrix $M$ of size $n\times m$, denoted by $\Rbin(M)$, is
the minimal $d$ for which there exist $0,1$ matrices $A$ and $B$ of size $n \times d$ and $d \times m$ respectively, such that $M = A \cdot B$, where the operations are over the integers.
\item The {\em Boolean rank} of a $0,1$ matrix  $M$, denoted by $\Rbool(M)$, is defined similarly to the binary rank, but here the operations are under Boolean arithmetic
  (namely, $0+x=x+0=x$, $1+1=1 \cdot 1 = 1$ and $x \cdot 0 = 0 \cdot x = 0$).
\end{itemize}
Such decompositions $M=A \cdot B$, for $A$ and $B$ of size $n \times d$ and $d \times m$, where $d$ is equal to the specific rank studied,
are called {\em optimal} for the given rank.
In this survey we focus on the last two rank functions: the Boolean and the binary rank and their connection to the real rank, but the techniques
and results presented may be of general interest and importance to the study of other rank functions.

As to the relationship between these three rank functions on a given $0,1$ matrix $M$: by definition, it always holds that
$$\Rreal(M) \leq \Rbin(M) \;\; \text{and} \;\; \Rbool(M) \leq \Rbin(M),$$
whereas $\Rbool(M)$ can be smaller or larger than $\Rreal(M)$.
The identity matrix is a trivial example where all three ranks are equal,
and the binary and Boolean rank are always equal to the real rank for real rank $1$ or $2$.
However, there exist families of matrices for which the Boolean rank is exponentially smaller than the binary and real rank,
whereas, for other families, the Boolean and binary rank are quasi polynomially larger than the real rank.
Tight bounds on the possible gaps between the three rank functions are yet to be found.

As we show in this survey, the Boolean and binary rank have several equivalent formulations,
each of which highlights a different way of interpreting the rank of the matrix.
These formulations include covering or partitioning the ones of a $0,1$ matrix by monochromatic rectangles,
biclique partition or cover of the edges of a bipartite graph, the intersection number of a matrix, the set basis problem and communication complexity.
These different viewpoints reveal rich connections to classic fields of research in mathematics, such as graph theory, algebra and combinatorics,
as well as to many areas of research in computer science, both theoretical and practical, such as
communication complexity, clustering and tiling databases, data mining, machine learning, bioinformatics and more.

These connections show that questions about rank functions often reappear in seemingly unrelated areas of research.
Thus, understanding these equivalences provides a common language for comparing techniques and transferring insights across fields of mathematics and computer science.
We will show that the methods developed by the different communities complement one another, and when used together can resolve problems which were open for many years.
See also a survey of Miettinen and Neuman from 2021 on the various applications of the Boolean and binary rank in computer science~\cite{miettinen2021recent},
as well as a survey of Monson, Pullman and Rees~\cite{Monson} from 1995 on the mathematical applications and variations of these rank functions.

The real rank of a matrix, or the rank over any field for that matter, is well understood.
Many powerful techniques, mostly from linear algebra, were devised over the years and led to combinatorial and computational applications involving the real rank.
In particular, it is easy to compute the real rank of a matrix in polynomial time using Gaussian elimination.
As to the Boolean and binary rank, although they were studied extensively, both by mathematicians and by computer scientists,
their behaviour is far less understood, mainly because they are not defined over a field.
In particular, the Boolean and binary rank are $NP$-hard to compute and the binary and Boolean rank of only a few concrete families of $0,1$ matrices has been computed precisely.

This survey presents mathematical properties of these rank functions and their connection to the real rank,
along with relevant computational complexity and algorithmic results.
As the scope of this survey is limited, it cannot include all known results.
Some results are reformulated for clarity or emphasis, while others appear here for the first time.
We include proofs, as well as lower and upper bound techniques, to demonstrate the wide variety of methods used.
The proofs were selected for their elegance or significance, and some proofs are presented, although simple,
since the corresponding original proof or paper cannot be found online (for example~\cite{ Monson} and~\cite{Caen2}).
Another goal of this survey is to highlight the diverse results and methods which emerged independently across different fields of  mathematics and computer science.
Some of the results were independently rediscovered due to their publication in separate research communities.
A unified survey can, therefore, assist those interested in this topic to gain a comprehensive understanding of the progress made to date.

\paragraph{Overview:}
In Section~\ref{Sec-notation} we present the needed notation.
Section~\ref{Sec:equivalent} describes the main alternative formulations of the binary and Boolean rank, thus, demonstrating their importance.
In Section~\ref{Sec-complexity} we discuss the computational complexity of the Boolean and binary rank, as well as show their fundamental connection to the field of communication complexity.
Section~\ref{Sec-lower-and-upper} summarizes key lower and upper bound techniques used to give bounds on the binary and Boolean rank.
These techniques include linear algebra, combinatorics and graph theory, isolation sets, the probabilistic method, kernelization,
communication protocols and the query to communication lifting technique.
In Section~\ref{Sec-properties} we consider mathematical properties of the binary and Boolean rank in comparison with those of the real rank,
as well as show non-trivial bounds on the binary and Boolean rank of specific families of matrices.
In Section~\ref{Sec-algorithms} we survey algorithmic directions taken to compute or approximate these rank functions,
including parameterized algorithms, approximation algorithms, property testing and approximate Boolean matrix factorization (BMF).
Finally, in Section~\ref{sec-discussion} we summarize some of the open problems suggested throughout the survey.

\section{Notation}
\label{Sec-notation}
Most matrices in this survey are $0,1$ matrices, unless stated otherwise.
The element on the $i$'th row and $j$'th column of a matrix $M$ is denoted by $M_{i,j}$.
The transpose of a matrix $M$ is denoted by $M^t$, and similarly for a vector.
Given a $0,1$ matrix $M$, its {\em complement} $\overline{M}$ is the matrix resulting by replacing each $0$ of $M$ with a $1$ and each $1$ with a $0$.
A matrix is called {\em circulant} if each row is a cyclic shift of the row that precedes it by one element to the right.

Most graphs in this survey are bipartite graphs $G = (V_1,V_2,E)$, where $V_1$ is a set of $n$ vertices on one side and $V_2$ is a set of $m$ vertices on the other side,
and $E$ is the set of edges between $V_1$ and $V_2$.
The {\em reduced  adjacency matrix} of $G$ is a $0,1$ matrix $M$ of size $n \times m$,
such that there is a row for each vertex in $V_1$ and a column for each vertex in $V_2$, and $M_{i,j} = 1$ if and only if $(i,j) \in E$.
In this case $M$ is also called a {\em bi-adjacency matrix}.

We consider in some cases general graphs $G = (V,E)$, where $V$ is a set of $n$ vertices and $E$ the set of edges.
The  {\em adjacency matrix} of  $G$  is a $0,1$ matrix $M$ of size $n \times n$, such that
$M_{i,j} = 1$ if and only if $(i,j) \in E$.

The {\em chromatic number} $\chi(G)$ of a graph $G = (V,E)$ is the minimum number of colors required to color all vertices in $V$, such that if $(i,j) \in E$  then $i$ and $j$ have different colors.

The notations $\widetilde{O}$ and $\widetilde{\Omega}$ are used to hide poly-logarithmic factors in the arguments of $O$ and $\Omega$.

\section{Alternative Formulations}
\label{Sec:equivalent}

The broad interest in the Boolean and binary rank and their applicability stems also from the following alternative formulations,
as well as their deep connection to the field of communication complexity which is described in Section~\ref{subsec-communication}.
See Figure~\ref{fig:Alternative} for an illustration of the  formulations described next.


\vspace{-0.5 cm}
\paragraph{Covering and partitioning with monochromatic rectangles:}  A {\em monochromatic combinatorial rectangle} in a $0,1$ matrix $M$ is a sub-matrix of $M$, all of whose entries have the same value.
The rows and columns of a rectangle do not have to be consecutive in $M$.
The {\em partition number} of $M$, denoted by $Partition_1(M)$, is the minimum number of monochromatic rectangles required to partition the ones of $M$,
and the {\em cover number} of $M$, denoted by $Cover_1(M)$, is the minimum number of monochromatic rectangles required to cover the ones of $M$, where rectangles may overlap.

\begin{lemma}[\cite{Gregory}]
Let $M$ be a $0,1$ matrix. Then $\Rbin(M) = Partition_1(M)$ and $\Rbool(M) = Cover_1(M)$.
\end{lemma}
\begin{proof}
Assume that $M$ is a matrix of size $n \times m$.
Each monochromatic rectangle in $M$ can be defined by a $0,1$ matrix of rank $1$ which is a product of two $0,1$ vectors: $a^t \cdot b$,
where $a^t$ is a column vector of length $n$ and $b$ is a row vector of length $m$.
The lemma follows by noticing that if the partition or cover of the ones of $M$ includes $d$ rectangles then:
$$M = \sum_{ i= 1}^d a^t_i \cdot b_i = A \cdot B,$$
where the vectors $a^t_1,...,a^t_d$ are the columns of a matrix $A$ of size $n \times d$, and $b_1,...,b_d$ are the rows of a matrix $B$
of size $d \times m$,
and the operations are either the regular addition in case of disjoint rectangles, or the Boolean operations, that is $1+1 = 1$, in case of rectangles which may overlap.
The minimality of a partition or cover by rectangles corresponds to the optimality of a decomposition $A \cdot B$.
\end{proof}

\vspace{-0.5 cm}
\paragraph{Biclique edge cover and partition of bipartite graphs:}
The Boolean or binary rank of a $0,1$ matrix $M$ is equal to the minimal number of bicliques needed to cover or partition, respectively, all the
edges of a bipartite graph $G$ whose {\em reduced adjacency matrix} is $M$.
Specifically, if $G$ is a bipartite graph with $n$ and $m$ vertices on each side, then $M$ has $n$ rows and $m$ columns,
and $M_{i,j} = 1$  if and only if $(i,j)$ is an edge of the graph $G$.
A monochromatic combinatorial rectangle of ones in $M$ corresponds to a biclique in $G$ (see Gregory, Pullman, Jones and Lundgren~\cite{Gregory}).

\vspace{-0.5 cm}
\paragraph{Subset intersection:}
The {\em intersection number} of a $0,1$-matrix $M$ is the smallest integer $d$ such that it is possible to assign
subsets  of $ \{1,...,d\}$ to the rows and columns of  $M$,
where if $X_i$ and $Y_j$ are the subsets assigned to row $i$ and column $j$, respectively, then $X_i \cap Y_j \neq \emptyset$ if and only if $M_{i,j} = 1$.
Notice that for each $1 \leq t \leq d$, the set of all entries $M_{i,j}$ for which $t \in X_i \cap Y_j$, defines a monochromatic rectangle of ones.
Since the rectangles defined by the subsets cover all ones of $M$, we get a cover of size $d$, and thus, the intersection number of $M$ coincides with its Boolean rank.
An analogue of the intersection number, where $|X_i \cap Y_j| = 1$ if and only if $M_{i,j} = 1$, and the intersection is empty otherwise,
coincides with the binary rank.

\vspace{-0.5 cm}
\paragraph{Set basis problem:}  Let $\mathcal{S} = \{S_1,S_2,...,S_m\}$ be a family of sets over some universe $U$.
A {\em basis} for $S$ is a family $\mathcal{B} = \{B_1,...,B_d\}$ of sets, such that for each $S_i \in \mathcal{S}$ there exists a subset of sets from $\mathcal{B}$ whose union equals $S_i$.
The goal is to find a basis with minimal cardinality for $\mathcal{S}$. A variation of this problem is the {\em normal set basis problem}.
Here a basis $\mathcal{B}$ for $\mathcal{S}$ is called {\em normal} if for each $S_i \in \mathcal{S}$ there exists a subset of pairwise disjoint sets from $\mathcal{B}$ whose union equals $S_i$,
and the goal is again to find a minimal normal basis for $\mathcal{S}$. See Stockmeyer~\cite{stockmeyer1975set} and Jiang and Ravikumar~\cite{jiang1993minimal}.

Given a $0,1$ matrix $M$ of size $n \times m$ we can define a set basis problem as follows. The universe from which the set elements are chosen is $U = \{1,2,...,n\}$.
Let $\mathcal{S} = \{S_1,...,S_m\}$, where $S_j = \{i_1,...,i_t\}$ if column $j$ of $M$ has ones in rows $i_1,...,i_t$ and zeros elsewhere.
A rectangle cover of size $d$ of the ones of $M$ corresponds to a set basis of size $d$ for $\mathcal{S}$, and a rectangle partition corresponds to a normal set basis for $\mathcal{S}$.
To verify this, let $M = X \times Y$ be a Boolean decomposition of $M$, where $X$ is of size $n \times d$ and $Y$ of size $d \times m$.
The columns of $X$ can be considered as the basis for $\mathcal{S}$, and the rows of $Y$ indicate which subsets of $\mathcal{S}$ are covered by each subset in the basis. That is,
 $Y_{i,j} = 1$ if and only if column $X_i$ is included in the union of the subsets covering $S_j$.
Similarly, if $X\cdot Y$ is a binary decomposition of $M$ then  both the rectangles and the basis subsets are disjoint.

\renewcommand{\kbldelim}{(}
\renewcommand{\kbrdelim}{)}

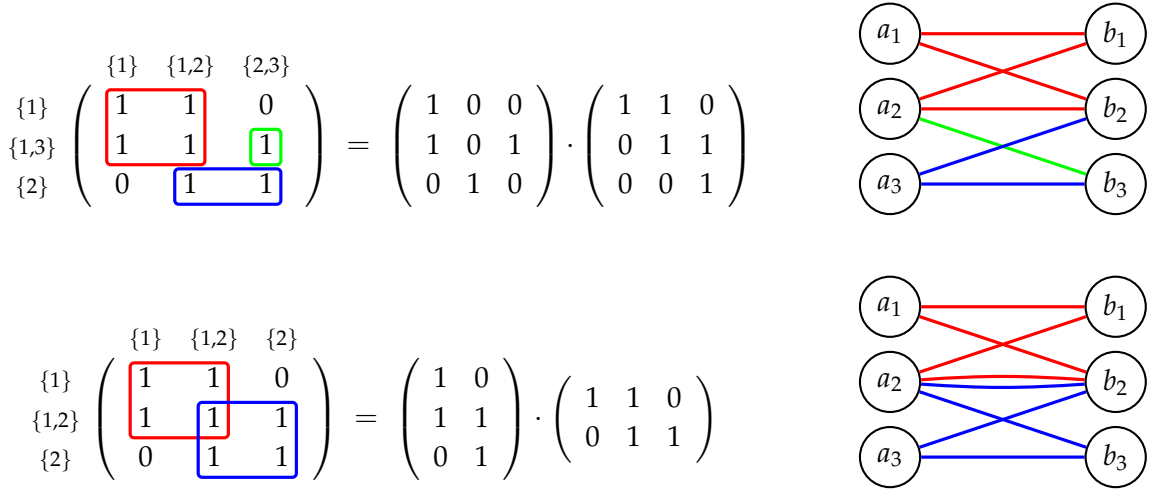
\begin{figure}[htb!]
\centering

\begin{minipage}[c]{0.6\textwidth}
\centering
\[
\kbordermatrix{
    & \{1\} & \{1,2\} & \{2,3\} \\
  \{1\}   & \cellnode{c11}1 & \cellnode{c12}1 & 0 \\
  \{1,3\} & \cellnode{c21}1 & \cellnode{c22}1 & \cellnode{gr11}\cellnode{gr12}\cellnode{gr21}\cellnode{gr22}1 \\
  \{2\}   & 0 & \cellnode{g11}\cellnode{g12}1 & \cellnode{g21}\cellnode{g22}1
}
\begin{tikzpicture}[overlay,remember picture]
  \node[fit=(c11)(c12)(c21)(c22), inner xsep=3pt, inner ysep=3pt] (RedFit) {};
  \node[fit=(g11)(g12)(g21)(g22), inner xsep=3pt, inner ysep=3pt] (BlueFit)  {};
   \node[fit=(gr11)(gr12)(gr21)(gr22), inner xsep=3pt, inner ysep=3pt] (GreenFit)  {};

  \draw[red,very thick,rounded corners=1.5pt]
    ([yshift=+6pt]RedFit.north west)
      rectangle
    ([xshift=+5pt,yshift=-1pt]RedFit.south east);

  \draw[blue,very thick,rounded corners=1.5pt]
    ([yshift=+6pt]BlueFit.north west)
      rectangle
    ([xshift=+5pt,yshift=-1pt]BlueFit.south east);

      \draw[green,very thick,rounded corners=1.5pt]
    ([yshift=+6pt]GreenFit.north west)
      rectangle
    ([xshift=+5pt,yshift=-1pt]GreenFit.south east);

\end{tikzpicture}
\;=\;
\left(
  \begin{array}{ccc}
    1 & 0 & 0 \\
    1 & 0 & 1 \\
    0 & 1 & 0
  \end{array}
\right)
\cdot
\left(
  \begin{array}{ccc}
    1 & 1 & 0 \\
    0 & 1 & 1 \\
    0 & 0 & 1
  \end{array}
\right)
\]

\end{minipage}\hfill
\begin{minipage}[c]{0.38\textwidth}
\centering
\begin{tikzpicture}[thick, every node/.style={circle, draw, minimum size=0.8cm, inner sep=0pt}]
  \node (a1) at (0,2) {$a_1$};
  \node (a2) at (0,1) {$a_2$};
  \node (a3) at (0,0) {$a_3$};

  \node (b1) at (3,2) {$b_1$};
  \node (b2) at (3,1) {$b_2$};
  \node (b3) at (3,0) {$b_3$};

  \draw[red, very thick]   (a1) -- (b1);
  \draw[red, very thick]   (a1) -- (b2);
  \draw[red, very thick]   (a2) -- (b1);
  \draw[red, very thick]   (a2) -- (b2);
  \draw[green, very thick] (a2) -- (b3);
  \draw[blue, very thick]  (a3) -- (b2);
  \draw[blue, very thick]  (a3) -- (b3);
\end{tikzpicture}
\end{minipage}

\vspace{2em}

\begin{minipage}[c]{0.6\textwidth}
\centering
\[
\kbordermatrix{
    & \{1\} & \{1,2\} & \{2\} \\
  \{1\}   & \cellnode{b11}1 & \cellnode{b12}1 & 0 \\
  \{1,2\} & \cellnode{b21}1 & \cellnode{b22}\cellnode{r11}1 & \cellnode{r12}1 \\
  \{2\}   & 0 & \cellnode{r21}1 & \cellnode{r22}1
}
\begin{tikzpicture}[overlay,remember picture]
  \node[fit=(b11)(b12)(b21)(b22), inner xsep=3pt, inner ysep=3pt] (RedFit) {};
  \node[fit=(r11)(r12)(r21)(r22), inner xsep=3pt, inner ysep=3pt] (BlueFit)  {};

  \draw[red,very thick,rounded corners=1.5pt]
    ([yshift=+6pt]RedFit.north west)
      rectangle
    ([xshift=+5pt,yshift=-1pt]RedFit.south east);

  \draw[blue,very thick,rounded corners=1.5pt]
    ([yshift=+6pt]BlueFit.north west)
      rectangle
    ([xshift=+5pt,yshift=-1pt]BlueFit.south east);
\end{tikzpicture}
\;=\;
\left(
  \begin{array}{cc}
    1 & 0 \\
    1 & 1 \\
    0 & 1 \\
  \end{array}
\right)
\cdot
\left(
  \begin{array}{ccc}
    1 & 1 & 0 \\
    0 & 1 & 1 \\
  \end{array}
\right)
\]
\end{minipage}\hfill
\begin{minipage}[c]{0.38\textwidth}
\centering
\begin{tikzpicture}[thick, every node/.style={circle, draw, minimum size=0.8cm, inner sep=0pt}]
  \node (a1) at (0,2) {$a_1$};
  \node (a2) at (0,1) {$a_2$};
  \node (a3) at (0,0) {$a_3$};
  \node (b1) at (3,2) {$b_1$};
  \node (b2) at (3,1) {$b_2$};
  \node (b3) at (3,0) {$b_3$};
  \draw[red, very thick]               (a1) -- (b1);
  \draw[red, very thick]               (a1) -- (b2);
  \draw[red, very thick]               (a2) -- (b1);
  \draw[blue, very thick]              (a2) -- (b3);
  \draw[blue, very thick]              (a3) -- (b2);
  \draw[blue, very thick]              (a3) -- (b3);
  \draw[red, very thick, bend left=4]  (a2) to (b2);
  \draw[blue,  very thick, bend right=4] (a2) to (b2);
\end{tikzpicture}
\end{minipage}
\caption{An example of a $3 \times 3$ matrix which has  binary rank $3$ and Boolean rank $2$.
The top part of the figure presents an optimal binary decomposition of the matrix, and a partition of the ones into $3$ rectangles.
Next to it is the corresponding bipartite graph with a biclique edge partition of size $3$.
Each biclique has a different color, as that of the corresponding rectangle in the partition.
\newline
The lower part of the figure shows an optimal Boolean decomposition of the matrix, and a cover of the ones by $2$ rectangles.
Next to it is the corresponding bipartite graph with a biclique edge cover of size $2$.
The edge between $a_2$ and $b_2$ in this bipartite graph belongs to both the red and the blue bicliques, and, therefore, is represented twice in the figure.
\newline
Above and to the left of each of the two matrices is a labeling of the rows and columns which corresponds to the subset intersection problem.
For the first matrix each label is a subset of $\{1,2,3\}$, and for the second matrix a subset of $\{1,2\}$.
\newline
Finally, if we let $\mathcal{S} = \{\{1,2\},\{1,2,3\},\{2,3\}\}$, then  $\{\{1,2\},\{2\},\{3\}\}$ is a minimal normal set basis for $\mathcal{S}$
and $\{\{1,2\},\{2,3\}\}$ is a minimal set basis for $\mathcal{S}$.}
\label{fig:Alternative}	
\end{figure}

\section{Complexity Aspects of the Rank Functions}

Beyond the differences in the algebraic definition of the real, Boolean and binary rank functions, and the various combinatorial interpretations of the Boolean and binary rank,
these rank functions exhibit also significantly different computational behavior.
While the real rank is efficiently computable using  Gaussian elimination, the Boolean and binary rank are in general hard to determine.
At the same time, these three rank functions are tightly intertwined with the models of communication complexity,
where they serve as key measures for understanding the amount of information which must be exchanged when computing a function of two variables.
This section examines and illustrates the computational complexity of these rank functions,
as well as their interpretation within communication complexity.

\label{Sec-complexity}
\subsection{Computational Complexity}

The real rank of a matrix is  easy to compute in polynomial time using Gaussian elimination,
but both the Boolean and the binary rank were shown to be $NP$-complete.
Orlin~\cite{orlin1977contentment} showed in 1977 that computing the Boolean rank is $NP$-complete by showing a polynomial reduction
from the {\em clique cover problem} of a graph to the biclique edge cover problem of a bipartite graph.
The clique cover problem requires to find the minimum number of cliques which cover all vertices of a given graph.
Since the first problem was shown to be $NP$-complete~\cite{Karp1972}, it follows that the latter is also $NP$-complete.

\begin{lemma}[\cite{orlin1977contentment}]
There is a polynomial reduction from the clique cover problem of a graph to the biclique edge cover problem of a bipartite graph.
\end{lemma}
\begin{proof}
Let $G = (V,E)$ be an undirected graph, where $V = \{1,2,...,n\}$ is the set of vertices.
Define a bipartite graph $H_G $ as follows. There are $n$ vertices $\{x_1,...,x_n\}$ on the left side and $n$ vertices  $\{y_1,...,y_n\}$ on the right side of $H_G$.
As to the edges of $H_G$,
there is an edge $(x_i,y_i)$ in $H_G$ for each $1 \leq i \leq n$,
and $(x_i,y_j)$ is an edge of $H_G$ if $(i,j)\in E$.
It is easy to verify that a clique cover of the vertices of $G$ induces a biclique edge cover of the same size for $H_G$ and vice versa.
\end{proof}

Actually, the Boolean rank problem can also be shown to be $NP$-complete  using a result of Stockmeyer~\cite{stockmeyer1975set} from 1975,
who proved that the set basis problem described in Section~\ref{Sec:equivalent} is $NP$-complete, since  both problems are equivalent.

There is also a nice polynomial reduction from the Boolean rank of a  $0,1$ matrix $M$  to the chromatic number, $\chi(G_M)$, of a graph $G_M$, as proved by the following lemma.
In Section~\ref{SubSec-approx} we show how this reduction can be used to give an approximation algorithm for the Boolean rank.
See Figure~\ref{fig:reduction} for an example of the reduction.

\begin{lemma}
\label{lem-boolean-color}
Let $M$ be a $0,1$ matrix of size $n \times m$. There exists a graph $G_M$ with $O(n\cdot m)$ vertices, such that  $\chi(G_M) = \Rbool(M)$.
\end{lemma}
\begin{proof}
The reduction from $M$ to $G_M$ is as follows. The graph $G_M$ contains a vertex $x_{i,j}$ for each $i,j$ for which $M_{i,j} = 1$.
As to the edges of $G_M$, there is an edge between $x_{i,j}$ and $x_{k,\ell}$ if entries $(i,j)$ and $(k,\ell)$ cannot
belong to the same monochromatic rectangle of ones in $M$, that is, either $M_{i,\ell} = 0$ or $M_{k,j} = 0$ or both.
This means that vertices $x_{i,j}$ and $x_{k,\ell}$ cannot be colored with the same color since they are connected with an edge.
The graph $G_M$ is also called the {\em conflict graph} of $M$.

We now prove that $\chi(G_M) = \Rbool(M)$.
Consider a monochromatic rectangle cover $R_1,...,R_d$ of the ones of $M$.
Color vertex $x_{i,j}$ with color $t$ if entry $(i,j)$ is covered by rectangle $R_t$.
If an entry belongs to more than one  rectangle in the cover, color it arbitrarily with the color of one of these rectangles.
For every two entries $(i,j),(k,\ell)$ colored with the same color $t$ it holds that $(i,j),(k,\ell)\in R_t$.
Since $R_t$ is a monochromatic rectangle, then $M_{i,j} = M_{k,\ell} = M_{i,\ell} = M_{k,j} = 1$.
Thus, there is no edge between the corresponding vertices in $G_M$, and so, this is a coloring whose size is equal to the number of rectangles in the cover.

Similarly, a coloring of size $d$ in $G_M$  induces a cover of the ones of $M$ by $d$ monochromatic rectangles as follows.
First, place in $R_t$ all entries $(i,j)$ for which vertex $x_{i,j}$ has color $t$.
Next add to $R_t$ all entries in the sub-matrix induced by the entries added to $R_t$ in the first stage. Again an entry can belong to more than one rectangle.
Finally, the entries added to $R_t$ in the second stage are also one entries of the matrix,
since if $(i,j), (k,\ell)$ were both added to $R_t$ in the first stage,
then by definition of $G_M$ there is no edge between vertices $x_{i,j}$ and $x_{k,\ell}$,
and this means that also $M_{i,\ell} = M_{k,j} = 1$. Thus, $R_t$ is a monochromatic rectangle as claimed.
\end{proof}

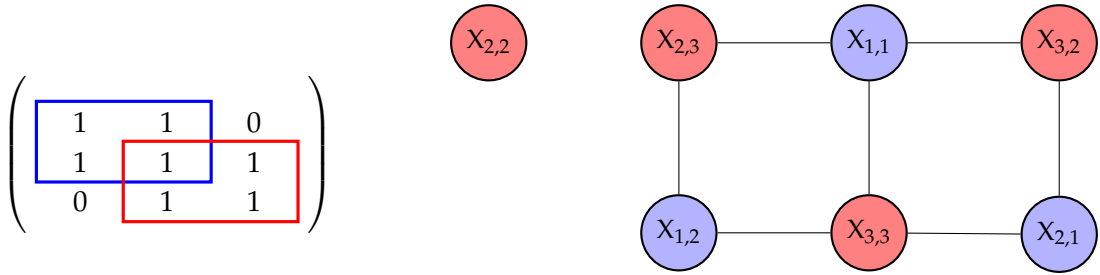
\begin{figure}[htb!]
\centering
\begin{minipage}[c]{0.25\textwidth}
\centering
\[
\left(
\begin{tikzpicture}[scale=1.1, baseline={(current bounding box.center)}]
  \matrix (m)[
    matrix of math nodes,
    nodes in empty cells,
    minimum width=width("998888"),
  ] {
    1  & 1  & 0   \\
    1  & 1  & 1   \\
    0   & 1 & 1   \\
    } ;
  \draw[blue, very thick] (m-2-1.south west) rectangle (m-1-2.north east);
  \draw[red, very thick] (m-3-2.south west) rectangle (m-2-3.north east);
\end{tikzpicture}
\right)
\]
\end{minipage}
\hspace{1cm}
\begin{minipage}[c]{0.6\textwidth}
\centering
\begin{tikzpicture}[
    node distance=1.5cm,
    every node/.style={circle, draw=black, thick, minimum size=1cm, font=\small, inner sep=0pt},
    rednode/.style={fill=red!50},
    bluenode/.style={fill=blue!30}
]
\node[rednode] (x22) at (0,0) {X$_{2,2}$};
\node[rednode, right=1.5cm of x22] (x23) {X$_{2,3}$};
\node[bluenode, below=1.5cm of x23] (x12) {X$_{1,2}$};
\node[bluenode, right=1.5cm of x23] (x11) {X$_{1,1}$};
\node[rednode, right=1.5cm of x11] (x32) {X$_{3,2}$};
\node[bluenode, below=1.52cm of x32] (x21) {X$_{2,1}$};
\node[rednode, below=1.5cm of x11] (x33) {X$_{3,3}$};
\draw (x23) -- (x12);
\draw (x23) -- (x11);
\draw (x11) -- (x32);
\draw (x11) -- (x33);
\draw (x33) -- (x21);
\draw (x32) -- (x21);
\draw (x33) -- (x12);

\end{tikzpicture}
\end{minipage}
\caption{The reduction described in Lemma~\ref{lem-boolean-color} from the Boolean rank problem of the matrix on the left to the vertex coloring problem of the graph on the right.
The ones of the matrix can be covered by the blue and the red monochromatic rectangles as shown. The vertices of the graph are, thus, also colored with two colors.
Vertices $x_{2,2},x_{2,3},x_{3,2},x_{3,3}$ are colored in red and the remaining vertices are colored in blue.
Note that $x_{2,2}$ can be colored either in red or in blue, and the corresponding one in the matrix belongs to both the red and the blue rectangles.}
\label{fig:reduction}	
\end{figure}

Jiang and Ravikumar~\cite{jiang1993minimal} proved in 1993 that the binary rank is $NP$-complete by showing a polynomial reduction from the vertex cover problem
to the normal set basis problem which is equivalent to the binary rank problem.
This reduction was used in 2025 by Hirahara,  Ilango  and Loff~\cite{hirahara2025communication} to prove
that determining  the deterministic communication complexity  of a $0,1$ matrix is also $NP$-hard (see Subsection~\ref{subsec-communication}).
In their proof they use the formulation of the binary rank of a matrix instead of the normal set basis problem used in~\cite{jiang1993minimal}.

As computing the binary and Boolean rank exactly is $NP$-hard, approximation algorithms are required.
However, it was shown by Chalermsook,  Heydrich,  Holm, and  Karrenbauer~\cite{chalermsook2014nearly} that it is hard to approximate the
Boolean rank\footnote{A similar hardness result was claimed in~\cite{chalermsook2014nearly} for the binary rank,
but the authors later stated that it contains a flaw, and therefore, their result holds only for the Boolean rank.} of a $0,1$ matrix $M$ of size $n \times n$  to within a
factor of $n^{1 - \epsilon}$ for any given $\epsilon > 0$. This is not surprising due to the close relationship of the Boolean rank and the chromatic number
as shown above. Indeed, Zuckerman~\cite{zuckerman2007linear} proved in 2007 that it is $NP$-hard to approximate the chromatic number of a graph with $n$ vertices
to within $n^{1 - \epsilon}$, for any given $\epsilon > 0$.
In Section~\ref{SubSec-approx}, we present simple approximation algorithms  which achieve an approximation ratio of $n /\log n$ for the Boolean rank and a ratio of $O(n /\log^2 n)$ for the binary rank.

Other directions of relaxing these problems include parameterized algorithms and property testing, as discussed in Section~\ref{Sec-algorithms}.
In the context of parameterized algorithms, it was shown that the binary and Boolean rank are {\em fixed-parameter tractable} with respect to the rank parameter $d$.

Furthermore, due to the hardness of these problems, there are results which compute the binary and Boolean rank of specific families of matrices or give lower and upper bounds.
For example,  Gregory~\cite{gregory1989biclique} and Haviv and Parnas~\cite{haviv2023binary} considered the binary rank of circulant matrices.
Regular matrices were studied by Pullman~\cite{Pullman88}, Gregory, Pullman, Jones and  Lundgren~\cite{Gregory}, Hefner,  Henson,  Lundgren and  Maybee~\cite{HefnerHLM90} and
by Haviv and Parnas~\cite{haviv2022mfcs, haviv2023binary}.
Boyer~\cite{boyer1998rank} considered the biclique edge partition of the complement of a path in a complete bipartite graph.
Guo, Huynh and Macchia~\cite{guo2018biclique} computed the biclique edge cover size of grids,
and Amilhastre,  Janssen and Vilarem~\cite{AMILHASTRE1998125} gave a polynomial time algorithm for finding the minimum biclique edge cover of bipartite domino-free graphs.
These results use a variety of tools and techniques from various fields of mathematics and computer science, as  demonstrated in subsequent sections of the survey.

\subsection{Communication Complexity}
\label{subsec-communication}
The Boolean and binary rank play a fundamental role in communication complexity, a field introduced in 1979 in a seminal paper by Yao~\cite{Yao79}.
The communication complexity problem in its basic form can be described as a game between two players, Alice and Bob,
who want to compute together some function $f(i,j)$.
Both players have access to a common matrix $M$, where $M$ is usually a $0,1$ matrix representing the function $f$ they want to compute.
Alice is given a row index $i$ and Bob receives a column index $j$, and their shared goal is to determine the value of $f(i,j) = M_{i,j}$,
while exchanging as little information as possible.
The information is shared between them by transmitting $0,1$ bits in rounds, according to a {\em communication protocol} they agreed upon in advance, until both players can determine $M_{i,j}$.
The {\em cost} of a communication protocol is the number of bits transmitted in the worst case between Alice and Bob while using this protocol,
over all input pairs $i,j$ of row and column indices.

The following communication complexity measures are used throughout the survey:
\begin{itemize}
\item
The {\em deterministic communication complexity} of $M$, denoted by $D(M)$,
is the minimum cost of a deterministic communication protocol for $M$.
\item
The {\em non-deterministic} communication complexity of $M$, denoted by $N_1(M)$, is defined similarly with the exception that the communication is non-deterministic.
That is, the players are given some "proof" or {\em certificate},
and  there exists a certificate and a transmission of bits between them where they both {\em accept} if and only if $M_{i,j}=1$.
The non-deterministic communication complexity is defined as the number of bits of the certificate given to Alice and Bob, in addition to the number of bits transmitted between them until they reach a decision.
\item
The {\em unambiguous} non-deterministic communication complexity of $M$, denoted by $U(M)$, further requires
that each input $i,j$ has at most one accepting computation.
\item The {\em co-non-deterministic} communication complexity of $M$, denoted by $N_0(M)$, is defined similarly to the non-deterministic communication complexity,
but here there exists a certificate and a transmission of bits between Alice and Bob where they both {\em accept} if and only if $M_{i,j}=0$.
Thus, the co-non-deterministic communication complexity is equal to the non-deterministic communication complexity of the complement matrix $\overline{M}$, that is, $N_0(M) = N_1(\overline{M})$.
\end{itemize}
The communication complexity of a $0,1$ matrix $M$ is strongly related to partitions and covers of $M$ by monochromatic combinatorial rectangles.
Recall that $Partition_1(M)$ and $Cover_1(M)$ denote the minimum number of monochromatic rectangles required to partition or cover the ones of $M$, respectively,
and let $Cover_0(M)$ denote the minimum number of monochromatic rectangles required to cover the zeros of $M$.

Results of Yao~\cite{Yao79}, Aho, Ullman and Yannakakis~\cite{aho1983notions} and  Yannakakis~\cite{Yannakakis91}
relate the above communication complexity measures to the partition and cover numbers as  summarized by the following theorem.
For additional background on the various communication complexity measures, see the books by Kushilevitz and Nisan~\cite{KN97} and Rao and Yehudayoff~\cite{rao2020communication}
on communication complexity, as well as Chapter 4 in a book by Jukna~\cite{jukna2012boolean}.

\begin{theorem}
\label{theo-communication-partition}
Let $M$ be a $0,1$ matrix. Then:
\begin{itemize}
\item
$\log_2 (Partition_1(M)) \leq D(M)  \leq O(\log^2 Partition_1(M))$.
\item
$N_1(M) = \lceil \log_2 (Cover_1(M))  \rceil $.
\item
$U(M) = \lceil \log_2 (Partition_1(M)) \rceil$.
\item
$N_0(M) = \lceil \log_2 (Cover_0(M))\rceil = \lceil \log_2 (Cover_1(\overline{M}))  \rceil $.
\end{itemize}
\end{theorem}

By the discussion in Section~\ref{Sec:equivalent}, $\Rbool(M) = Cover_1(M)$ and $\Rbin(M) = Partition_1(M)$. Thus,
the Boolean rank determines the non-deterministic complexity, and the binary rank determines the unambiguous non-deterministic complexity,
and gives an approximation up to a polynomial of the deterministic communication complexity.

G\"{o}\"{o}s, Pittasi and Watson~\cite{Goos} gave in 2015 an almost matching lower bound to
the upper bound on the deterministic communication complexity in terms of the partition number and proved the following theorem.
See also the discussion in Section~\ref{subsec-protocols}.

\begin{theorem}[\cite{Goos}]
\label{theo-deterministic-partition}
There exists a family of matrices $M$, such that $D(M) \geq \widetilde{\Omega}(\log^2 Partition_1(M))$.
\end{theorem}

The non-deterministic communication complexity is  $NP$-hard since it is determined by the Boolean rank which was proved to be $NP$-hard.
Yao~\cite{Yao79} asked already in 1979 if the deterministic communication complexity is  $NP$-hard,
and this was an open question for many years (see Kushilevitz and Weinreb~\cite{kushilevitz2009complexity} and Hirahara, Ilangoand and Loff~\cite{hirahara2021hardness}).
Finally, Hirahara,  Ilango  and Loff~\cite{hirahara2025communication} proved in 2025 that determining  the deterministic communication complexity of a $0,1$ matrix is  $NP$-hard.
Their proof uses the reduction  of Jiang and Ravikumar~\cite{jiang1993minimal} from the normal set basis problem to the binary rank problem.
Independently, Gaspers, He and Mackenzie~\cite{gaspers2025np} proved also in 2025 that the deterministic communication complexity is $NP$-hard by showing a reduction
from the $NP$-hard Vector Bin Packing problem.

It should by noted that the famous {\em log-rank conjecture} of Lov{\'a}sz and Saks~\cite{LovaszS88} suggests
that the deterministic communication complexity of a matrix $M$ is polynomially related to $\log_2 \Rreal(M)$.
Since $\log \Rbin(M)$ gives an approximation up to a polynomial of the deterministic communication complexity of $M$,
this conjecture in fact suggests that for every $0,1$ matrix $M$, $\log_2 \Rbin(M)$ and $\log_2 \Rreal(M)$ are polynomially related.
As to the current status of the log-rank conjecture, it is known that:
$$\log_2 \Rreal(M) \leq D(M) \leq O(\sqrt{\Rreal(M)}).$$
The lower bound on $D(M)$ in terms of the real rank was proved by Mehlhorn and Schmidt~\cite{mehlhorn1982vegas} in 1982.
The upper bound was proved in 2023 by Sudakov and Tomon~\cite{sudakov2023matrix} who improved by a factor of $\log_2 \Rreal(M)$ the previous upper bound given by Lovett~\cite{Shachar}.
Thus, there is still an exponential gap between the lower and upper bound as a function of the real rank.
See also a survey by Lovett~\cite{LovettA14} and by Lee and Shraibman~\cite{lee2023around},
as well as a recent paper of Hambardzumyan,  Lovett and Shirley~\cite{hambardzumyan2025log} about new formulations for the log-rank conjecture.

Due to these tight mutual connections between communication complexity and the Boolean and binary rank,
understanding the Boolean and binary rank can contribute to a better understanding of the different communication complexity measures.
On the other hand, techniques and results from communication complexity were already found useful for computing or giving bounds on
the Boolean and binary rank of concrete matrices (see Sections~\ref{subsec-lifting} and~\ref{Subsec-complemement}),
as well as used by the approximation algorithm for the binary rank described in Section~\ref{SubSec-approx}.

\section{Lower and Upper Bound techniques}
\label{Sec-lower-and-upper}

In this section we present some central techniques for establishing upper and lower bounds on the Boolean and binary rank of matrices.
As we demonstrate, these techniques draw from various fields of mathematics, such as algebra, combinatorics and graph theory,
alongside tools developed within the computer science community.
This variety of approaches has helped resolve problems which were open for many years, although many challenges remain.

Methods presented in this section will reappear in Section~\ref{Sec-properties}, where we explore properties of the Boolean and binary rank and present examples of
bounds for specific families of $0,1$ matrices, as well as in Section~\ref{Sec-algorithms}
which focuses on algorithms for computing or approximating these rank functions.
The applicability of the methods presented is demonstrated on the following two families of matrices, the family $C_n$ and the family $D_{n,k}$, described next.
See Figure~\ref{fig:D84} for an illustration of both families, and Table~\ref{tab:gaps} which summarizes the results presented in this Section
regarding these families.

\vspace{-0.5 cm}
\paragraph{The family $C_n$:} Let $C_n$ be the $0,1$  matrix of size $n \times n$, for $n \geq 2$, defined by an all-zero main diagonal and ones elsewhere.
The matrix $C_n$ is, thus, the complement of the identity matrix, and is also the reduced adjacency matrix of the {\em crown graph},
i.e. a complete bipartite graph from which a perfect matching was removed.

The matrix $C_n$ represents also the two-party communication problem of {\em non-equality}.
Here the two players, Alice and Bob, receive a row index $i$ and a column index $j$, respectively,
and their goal is to communicate as few bits as possible and compute the Non-Equality function:
$$ \text{Non-Equality}(i,j) = \left\{
                         \begin{array}{cc}
                           1 & \text{if} \; i \neq j \\
                           0 & \text{otherwise} \\
                         \end{array}
                         \right.
$$
Thus, the matrix $C_n$ captures exactly the decision that should be made by Alice and Bob on input $i,j$, since entry $(i,j)$ of $C_n$ is $1$ if and only if $i \neq j$.
In a similar way, the identity matrix captures the two-party communication problem of {\em equality}. In this case, Alice and Bob should accept if and only if $i = j$.

\vspace{-0.5 cm}
\paragraph{The family $D_{n,k}$:} Let $D_{n,k}$ be the $0,1$ circulant matrix of size $n \times n$, defined by a first row which starts with $n-k$ consecutive ones followed by $k$ zeros.
Note that $C_n = D_{n,1}$ up to a permutation of the columns.

\begin{figure}[htb!]
$$
D_{6,3} = \left(
  \begin{array}{cccccc}
    1 & 1 & 1 & 0 & 0 & 0  \\
    0 & 1 & 1 & 1 & 0 & 0  \\
    0 & 0 & 1 & 1 & 1 & 0 \\
    0 & 0 & 0 & 1 & 1 & 1  \\
    1 & 0 & 0 & 0 & 1 & 1 \\
    1 & 1 & 0 & 0 & 0 & 1 \\
   \end{array}
\right)
\;\;\;\;\;\;\;\;\;\;\;\;
C_{6} = \left(
  \begin{array}{cccccc}
    0 & 1 & 1 & 1 & 1 & 1  \\
    1 & 0 & 1 & 1 & 1 & 1  \\
    1 & 1 & 0 & 1 & 1 & 1  \\
    1 & 1 & 1 & 0 & 1 & 1  \\
    1 & 1 & 1 & 1 & 0 & 1  \\
    1 & 1 & 1 & 1 & 1 & 0  \\
   \end{array}
\right)
$$
\caption{On the left is the matrix $D_{6,3}$ and on the right the matrix $C_6$.
The ones on the main diagonal of $D_{6,3}$ form an isolation set of size $6$, and thus, $\Rbin(D_{6,3}) = 6$.
The maximum isolation set in $C_6$ is of size $3$, but $\Rreal(C_6) = 6$ and thus, $\Rbin(C_6) = 6$ as well.
See Subsection~\ref{subsec-isolation} for the definition of isolation sets,
and Subsection~\ref{subsec-linear} for the real rank of $C_n$.}
\label{fig:D84}	
\end{figure}

\begin{table}[ht]
    \centering
    \begin{tabular}{|Sc|Sc|Sc|Sc|Sc|}
    \hline
        Matrix  & $\Rreal$ & $\Rbool $ & $\Rbin $  & Section \\ [6pt] \hline \hline
        $C_n$ & $n$ & $\Theta(\log n)$ & $n $ &  \ref{subsec-linear},~\ref{Subsec-combinatorics}\\ [6pt] \hline
        \shortstack{$D_{n,k}$ \\ $ k < n/2$}    & $ n - \gcd(n,k) + 1$ & $O(k \log n)$ &  \shortstack{at least \\ $ \min\{n, n - \gcd(n,k) + 2\}$} & \ref{subsec-linear} \\ [6pt] \hline
        \shortstack{$D_{n,k}$ \\ $ k \geq n/2$} & $ n - \gcd(n,k) + 1$ & $n$ & $n$ &  \ref{subsec-isolation} \\ [6pt] \hline
    \end{tabular}
    \caption{The gaps between the rank functions for the matrices $C_n$ and $D_{n,k}$.}
    \label{tab:gaps}
\end{table}

Before we describe the various techniques, recall that by the definition of the rank functions presented in the introduction, it always holds that
$$\Rreal(M) \leq \Rbin(M) \;\; \text{and} \;\; \Rbool(M) \leq \Rbin(M),$$
whereas $\Rbool(M)$ can be smaller or larger than $\Rreal(M)$.
The binary and Boolean rank are always equal to the real rank for real rank $1$ or $2$,
but when $\Rreal(M) \geq 3$ there are examples of matrices for which the rank functions differ.

As we show, the matrix $C_n = D_{n,1}$ is an example where the Boolean rank is exponentially smaller than the real and the binary rank,
and this is the maximal gap possible between these two rank functions since $\Rbin(M) \leq 2^{\Rbool(M)}$ for any $0,1$ matrix $M$.
Also, there are values of $n,k$ for which  the real, binary and Boolean rank all differ on $D_{n,k}$.
In Section~\ref{Subsec-complemement} we describe a result which shows that the Boolean and binary rank can be quasi-polynomially larger than the real rank.

\subsection{Linear algebra}
\label{subsec-linear}
As noted above, the real rank of a $0,1$ matrix always serves as a lower bound for its binary rank.
This simple observation has been used to obtain tight or nearly tight lower bounds on the binary rank of certain families of matrices.
In other cases, as we show next, more advanced linear algebraic techniques are required.

The prominent Graham Pollak theorem~\cite{GrahamP71} from 1971 gives a lower bound on the number of bicliques required to partition the edges of the complete graph $K_n$ on $n$ vertices.
Their original proof uses the following lemma of Witsenhausen.
Denote by  $\delta_{-}(M)$ and by $\delta_{+}(M)$ the number of negative and positive eigenvalues, respectively, of a matrix $M$.

\begin{lemma}
Let $G$ be a simple graph and let $M$ be its adjacency matrix.
The number of bicliques required to partition the edges of $G$ is at least $\max\{\delta_{-}(M),\delta_{+}(M)\}$.
\end{lemma}

The Graham Pollak theorem follows immediately since $\max\{\delta_{-}(M),\delta_{+}(M)\} = n-1$, where $M$ is the adjacency matrix of $K_n$.
Several alternative proofs were found over the years for this theorem, such as a proof by
Tverberg~\cite{Tverberg1982OnTD} from 1982 and a proof by Peck~\cite{peck1984new} from 1984.
We now describe the elegant proof of Tverberg which uses elementary tools from linear algebra.

\begin{theorem}[\cite{GrahamP71},~\cite{Tverberg1982OnTD}]
Any biclique partition of the edges of $K_n$  includes at least $n-1$ bicliques.
\end{theorem}
\begin{proof}
 Consider any biclique partition of the edges of $K_n$ into $m $ bicliques, $(A_i,B_i)$,  $1 \leq i \leq m$,
where $A_i$ is the set of vertices on one side of the biclique and $B_i$ is the set of vertices on the other side.
Assume by contradiction that $m < n-1 $, let the vertices of $K_n$ be labeled by $\{1,...,n\}$,
and define a variable $x_i$ for each vertex $i$ of $K_n$. Since the bicliques partition all edges of $K_n$ the following equation holds:
$$
\sum_{i < j} x_i x_j = \sum_{k = 1}^m \left(\sum_{i \in A_k} x_i \right) \left( \sum_{j \in B_k} x_j \right)
$$
where the sum on the left goes once over all edges of the graph $K_n$, and the sum on the right goes over all bicliques in the partition, and for each biclique $(A_k,B_k)$
goes over each vertex  $i \in A_k$ and each vertex  $j \in B_k$ and multiplies the corresponding variables, thus, again going once over all edges of the graph.
Now, consider the following system of linear equations:
$$
x_1 + ... + x_n = 0
$$
$$
\sum_{i \in A_k} x_i = 0\;\; \text{for} \;\; 1 \leq k \leq m
$$
There are $m + 1 \leq n-1$ equations by our assumption on $m$ and there are $n$ variables $x_i$. Thus, there is a non trivial solution to this system of equations
for which at least one of the variables $x_i $ is non-zero.
Finally consider the following:
$$
0 = (x_1 + ... + x_n)^2 = \sum_{i = 1}^n x_i^2 + 2 \sum_{i < j} x_i x_j = \sum_{i = 1}^n x_i^2 + 2 \sum_{k = 1}^m \left(\sum_{i \in A_k} x_i \right) \left( \sum_{j \in B_k} x_j \right)
 = \sum_{i = 1}^n x_i^2 > 0
$$
where the inequality holds since at least one of $x_i$'s is non-zero.
This is of course a contradiction, and hence, $m \geq n-1$ as claimed.
\end{proof}

Inspired by the Graham Pollak theorem, Alon, Saks and Seymour conjectured in 1991 that any graph $G$ whose edges can be partitioned into $k$ bicliques, has a chromatic number of at most $k+1$.
This conjecture holds for the complete graph $K_n$ whose chromatic number is $n$. It took around 20 years until this conjecture was refuted by Huang and Sudakov~\cite{HuangS12}.
See also Section~\ref{Subsec-complemement}.

Linear algebra was also used in the results of Gregory~\cite{gregory1989biclique} and Haviv and Parnas~\cite{haviv2023binary} who considered the binary rank of $D_{n,k}$.
The real rank of $D_{n,k}$ can be computed as follows (see Ingleton~\cite{ingleton1956rank} and Grady and Newman~\cite{GRADY199711}):

\begin{theorem}
\label{theo-realDnk}
For all integers $n > k \geq 0$, $\Rreal(D_{n,k}) = n-\gcd(n,k)+1.$
\end{theorem}

This immediately implies that if $n$ and $k$ are relatively prime then $D_{n,k}$ has full rank over the reals, and thus, also binary rank $n$. In particular, $\Rbin(C_n) = n$,
since $C_n = D_{n,1}$.
However, when $\gcd(n,k) > 1$,  $\Rreal(D_{n,k}) < n$ and the question was raised whether  $\Rbin(D_{n,k}) > \Rreal(D_{n,k})$.
Gregory~\cite{gregory1989biclique}  proved in 1989 that $\Rbin(D_{n,2})=n$ for all $n \geq 3$,
and Haviv and Parnas~\cite{haviv2023binary}  generalized his result in 2023 and proved:

\begin{theorem}[\cite{haviv2023binary}]
\label{theo-binDnk}
For all integers $n > k >0$, $\Rbin(D_{n,k}) \geq \min\{\Rreal(D_{n,k})+1,n\}.$
\end{theorem}

Both results of~\cite{gregory1989biclique} and~\cite{haviv2023binary} use the following similar idea.
Assume that $\Rbin(D_{n,k}) = d$, and let $D_{n,k} = X \cdot Y$ be any binary decomposition of $D_{n,k}$,
where $X$ is a $0,1$ matrix of size $n \times d$ and $Y$ is a $0,1$ matrix of size $d \times n$.
If $\Rreal(D_{n,k}) < n$, it is possible to show that there exist a column $x$ of $X$ and a row $y$ of $Y$,
such that $\Rreal(D_{n,k} - x\cdot y) = \Rreal(D_{n,k})$.
This implies that $X' \cdot Y'$ is a binary decomposition of size $d-1 \geq \Rreal(D_{n,k})$ for the matrix $D_{n,k} - x \cdot y$,
where $X'$ is the matrix $X$ without column $x$ and $Y'$ is the matrix $Y$ without row $y$. Thus, $d \geq \Rreal(D_{n,k})+ 1$.

Using the formulation of the binary rank as a partition of the ones of the matrix into monochromatic rectangles of ones, this idea can be reformulated as follows.
Let $P$ be some partition of the ones of $D_{n,k}$ into monochromatic rectangles.
If $\Rreal(D_{n,k}) < n$ it is possible to show that there exists a rectangle $R$ in $P$,
so that if we remove $R$ from $D_{n,k}$, that is, exchange all the ones in this rectangle into zeros, then the real rank of the resulting matrix is $\Rreal(D_{n,k})$.
Since the partition $P \setminus R$ covers all ones of the resulting matrix, then $P \setminus R$ has at least $\Rreal(D_{n,k})$ rectangles, and
thus, with $R$ the partition $P$ has at least $\Rreal(D_{n,k})+1$ rectangles as claimed.

The exact binary rank of $D_{n,k}$ for the various values of $n,k$ is still unknown.
Hefner,  Henson,  Lundgren, and  Maybee~\cite{HefnerHLM90} conjectured that $\Rbin(D_{n,3})=n$, for all $n \geq 6$.

\subsection{Isolation (fooling) Sets}
\label{subsec-isolation}
A useful tool for proving lower bounds on the binary and Boolean rank of $M$ is finding a large {\em isolation set} in $M$
or a {\em fooling set} as it is called in communication complexity.

\begin{definition}
An {\em isolation set} in a $0,1$ matrix $M$,  is  a subset of $1$ entries of $M$, such that no two of these ones belong to an all-one $2\times 2$ sub-matrix of $M$,
and no two of these ones are in the same row or column of $M$. The size of the maximum isolation set of $M$ is denoted by $i(M)$.
\end{definition}

Thus, no two ones of an isolation set can belong to the same monochromatic rectangle,
and therefore, $i(M)$ provides a lower bound on $\Rbin(M)$ and $\Rbool(M)$ (see, e.g, Beasley~\cite{BEASLEY20123469}).

It was shown by Shitov~\cite{SHITOV20132500} that the maximum isolation set problem is $NP$-hard, by showing a reduction from the maximum independent set problem of a graph,
which is known to be $NP$-hard.
However, for some matrices it is relatively easy to show that there is a large isolation set and this gives an immediate lower bound on the binary and Boolean rank.
For example, the matrix $D_{n,k}$ has an isolation set of size $n$ for $k \geq \lfloor n/2 \rfloor$, and thus:

\begin{lemma}
If $k \geq \lfloor n/2 \rfloor $ then $\Rbool(D_{n,k}) = \Rbin(D_{n,k}) = n$.
\end{lemma}
\begin{proof}
It is easy to verify that when $k \geq \lfloor n/2 \rfloor $, the $n$ ones on the main diagonal of $D_{n,k}$ form an isolation set of size $n$.
The result follows, since the ones of every row of $D_{n,k}$ can each be covered  by a single monochromatic  rectangle.
\end{proof}
Note, that in this case using isolation sets gives a better lower bound on the binary rank than using the real rank as a lower bound,
since $\Rreal(D_{n,k}) = n-\gcd(n,k)+1$  can be smaller than $n$ for certain values of $n,k$.
However, there are other examples where there is a large gap between $\Rbool(M), \Rbin(M)$ and $i(M)$.
For example, $i(C_n) = 3$ for $n \geq 3$,
whereas de Caen, Gregory, and Pullman~\cite{Caen2} showed that $\Rbool(C_n) = \Theta(\log_2 n)$ and since $C_n = D_{n,1}$ then $\Rreal(C_n) = \Rbin(C_n)= n$.
See Figure~\ref{fig:D84} for an illustration.

The idea of an isolation set can be generalized to isolation sets of order $k$, or as they are sometimes called, {\em extended} isolation sets (see~\cite{dietzfelbinger1996comparison},\cite{orlitsky1988communication}).

\begin{definition}
A subset $F$ of one entries of a $0,1$ matrix $M$ is called an {\em isolation set of order $k$} for $M$, if any $k+1$ ones in $F$ cannot be contained in the same monochromatic rectangle of ones in $M$.
\end{definition}

The isolation sets discussed above are, thus, isolation sets of order $1$, and since at most $k$ ones of an isolation set of order $k$ can be contained in the same monochromatic rectangle, we get immediately:

\begin{claim}
Let $F$ be an isolation set of order $k$ in a $0,1$ matrix $M$. Then $\Rbool(M), \Rbin(M) \geq |F| /k$.
\end{claim}

Another direction of research which attracted attention is finding the maximal gap possible between $\Rreal(M)$ and  $i(M)$.
A result of Dietzfelbinger, Hromkovi{\v{c}} and Schnitger~\cite{dietzfelbinger1996comparison} proves that the gap can be at most quadratic:

\begin{theorem}[\cite{dietzfelbinger1996comparison}]
Let $M$ be a $0,1$ matrix. Then  $i(M) \leq \Rreal(M)^2$.
\end{theorem}

A combination of results of de Caen, Gregory, Henson, Lundgren, and Maybee~\cite{realnonnegative} and Parnas and Shraibman~\cite{parnas2025study}
gives a simple family of $0,1$ matrices, as described in the following theorem,
with a multiplicative gap of almost $2$ between the real rank and the binary and Boolean rank and maximum isolation set, as the size of the matrix grows.
This gap can be amplified using Kronecker products (see Section~\ref{subsec-kronecker}).

\begin{theorem}[\cite{realnonnegative,parnas2025study}]
\label{theodoublegap}
For all $n \geq 4$, $n \neq 5$, there exists a $0,1$ matrix $M$ of size $n \times n$, where all rows and columns of $M$ are distinct,
such that $\Rreal(M) = \lfloor n/2\rfloor + 1$ and $\Rbin(M) = \Rbool(M) = i(M) = 2\lfloor n/2\rfloor$.
\end{theorem}
\begin{proof}
The even sized matrices are the matrices $D_{n,n/2}$ of size $n \times n$.
The odd sized matrices are constructed by taking an even sized matrix $D_{n,n/2}$ and adding to it a carefully selected row and column
which are a linear combination of the rows and columns of $D_{n,n/2}$ and are different from all rows and columns of $D_{n,n/2}$.
\end{proof}

Friesen, Hamed,  Lee and Theis~\cite{friesen2015fooling} prove that for any $d \geq 1$ there exists a matrix $M_d$ over $\Q$,
such that $\Rreal(M_d) = d$ and $M_d$ has an isolation set of size $\binom{d+1}{2}$.
In this case the definition of an isolation set is such that $M_{i,j} \neq 1$ for any entry $(i,j)$ in the isolation set,
and the other requirements from an isolation set can be defined in a similar way to those of $0,1$ matrices.
Shigeta and  Amano~\cite{shigeta2015ordered} give a construction which almost matches the quadratic upper bound of~\cite{dietzfelbinger1996comparison}
between the real rank and the size of a maximum isolation set.

\begin{theorem}[\cite{shigeta2015ordered}]
There is a family of  $0,1$ matrices of size $n \times n$ with real rank $n^{1/2 + o(1)}$ and an isolation set of size $n$.
\end{theorem}

This theorem gives immediately an almost quadratic  gap between the real and the Boolean and binary  rank.
In Section~\ref{Subsec-complemement} we show a family of matrices with a quasi polynomial gap between the real and the Boolean and binary rank.
Additional examples where isolations sets were used can be found in~\cite{beasley2012isolation, de1988boolean, haviv2022upper, HefnerHLM90, watson2016nonnegative}.


\subsection{Combinatorics and Graph Theory}
\label{Subsec-combinatorics}
As described in Section~\ref{Sec:equivalent}, the Boolean and binary rank have alternative formulations which are closely related to problems in combinatorics and graph theory.
Thus, it is not surprising that these fields play an important role when studying the Boolean and binary rank.
One classic example is a theorem proved by Sperner~\cite{sperner1928satz} in 1928,
which was used many years later, in 1981, by de Caen, Gregory and Pullman~\cite{Caen2} to give a tight lower bound on the Boolean rank of the matrix $C_n$.

A {\em Sperner family} or an {\em anti-chain of sets}, is a family of sets for which none of the sets is a subset of another.
Using this definition Sperner proved the following well known theorem, which is named after him.

\begin{theorem}[Sperner's Theorem]
\label{thm:sperner}
Let $S$ be a Sperner family of subsets of $\{1,2,...,n\}$. Then, $|S| \leq {n \choose \lfloor n/2\rfloor}$.
\end{theorem}

It is possible to identify each row of a $0,1$ matrix of size $n \times m$ with a subset of the elements of $\{1,2,...,m\}$,
where if a row has a one in column $j$ then the corresponding subset contains $j$ as an element.
Thus, the rows of a $0,1$ matrix are an {\em anti-chain} if no row is  contained in another row as a subset. Let
$$\sigma(n) = \min \left\{d\;|\; n \leq {d \choose \lfloor d/2\rfloor}\right\}.$$
Note that $\sigma(n) = (1+o(1)) \cdot \log_2 n$.
Using Sperner's theorem, de Caen, Gregory and Pullman~\cite{Caen2} proved the following general lower bound on the Boolean rank of matrices,
and used this theorem to give a tight bound on the Boolean rank of the matrix $C_n$.

\begin{theorem}[\cite{Caen2}]
\label{thm:caen}
Let $M$ be a $0,1$ matrix of size $n \times m$. If the rows of  $M$ form an anti-chain, then $\Rbool(M) \geq \sigma(n)$.
\end{theorem}
\begin{proof}
Let $\Rbool(M) = d$. Thus, there exist $0,1$ matrices $X$ and $Y$, where $X$ is of size $n \times d$ and $Y$ is of size $d \times m$, such that $M = X\cdot Y$
and the operations are Boolean.
Since the rows of $M$ are an anti-chain, then the rows of $X$ are also an anti-chain. If we identify each $0,1$ row of length $d$ of $X$ with a subset of $\{1,...,d\}$,
then by Sperner's theorem, $n \leq {d \choose \lfloor d/2\rfloor}$. Hence, $d \geq \sigma(n)$ as required.
\end{proof}

\begin{corollary}[\cite{Caen2}]
\label{coro-Cn}
$\Rbool(C_n) = \sigma(n) = \Theta(\log n)$.
\end{corollary}
\begin{proof}
By Theorem~\ref{thm:caen},  $\Rbool(C_n) \geq \sigma(n)$. To prove an upper bound on $\Rbool(C_n)$: Let $X$ be a $0,1$ matrix of size $n \times \sigma(n)$,
whose rows are all distinct and each have exactly $\lceil \sigma(n)/2 \rceil$ ones. Let $Y  = \overline{X}^t$ be the transpose of the complement of $X$.
Therefore, $Y$ is of size $\sigma(n) \times n$ and has $\lfloor \sigma(n)/2 \rfloor$ ones in each column.
By construction, $C_n = X \times Y$. Therefore, $\Rbool(C_n) \leq \sigma(n)$, and we are done.
\end{proof}

Another general lower bound technique on the Boolean rank was given by Jukna and Kulikov~\cite{jukna2009covering} in 2009. They considered the graph formulation of the problem
and proved the following theorem:

\begin{theorem}[\cite{jukna2009covering}]
\label{theo-jukna2009}
Let $G = (V,E)$ be a graph with a matching of size $m$. Then at least $m^2/|E|$ bicliques are required to cover its edges.
\end{theorem}
\begin{proof}
Assume that the edges of $G$ can be covered by $t$ bicliques $B_1,...,B_t$. Let $M_i$ be the edges of the matching covered by $B_i$,
where if an edge of the matching is covered by more than one biclique, we assign it arbitrarily to one of the bicliques in the cover.
Let $F_i$ be the edges of the biclique spanned by the vertices of $M_i$,  that is, for each edge in $F_i$, both endpoints of this edge are part of edges of $M_i$.

For each $1 \leq i \leq t$,  the edges of $F_i$ form a complete bipartite graph with $|M_i|$ vertices on each side, that is,  $|F_i| = |M_i| \cdot |M_i|$.
Furthermore, $F = F_1 \cup F_2 \cup ... \cup F_t$ is a union of vertex disjoint bicliques which contains all $m$ edges of the matching.
Using the Cauchy-Schwarz inequality we get:
$$
m^2 = (|M_1| + ... + |M_t|)^2 \leq t \cdot (|M_1|^2 + |M_2|^2 + ... + |M_t|^2) = t \cdot |F|.
$$
Hence, $t \geq m^2 /|F| \geq m^2 /|E|$ as claimed.
\end{proof}

Theis~\cite{theis2007some} showed that the bound of~\cite{jukna2009covering} is bounded above by the size of a maximum isolation set in the corresponding adjacency matrix of the graph.
However, as stated above, the maximum isolation set problem is $NP$-hard, whereas finding a maximum matching can be done in polynomial time.
Thus, the theorem of~\cite{jukna2009covering} can be useful, especially for sparse graphs.
For example, a $d$-regular bipartite graph with $n$ vertices on each side always has a perfect matching of size $n$.
Hence, by Theorem~\ref{theo-jukna2009}, the biclique edge cover of such a graph is at least $n^2/(d\cdot n) = n/d$.

\subsection{The Probabilistic method}
\label{Sec-prob}

The probabilistic method is a powerful method which gives a non-constructive proof of the existence of an object with a certain property.
The idea is to prove that an object chosen randomly from a set of objects has the required property with a non-zero probability, and therefore, there exists such an object.
See a book by Alon and Spencer~\cite{alon2016probabilistic} for a comprehensive description of the beautiful results which can be proved using the probabilistic method.

Alon~\cite{alon1986covering} used this method to prove an upper bound on the number of cliques needed to cover the edges of a graph
which is a complement of a graph in which the degree of each vertex is bounded by $d$.
We slightly modify his argument and give a general upper bound on the Boolean rank of a matrix which has at most $d$ zeros in each column
(a similar claim holds for the rows).
See also Jukna~\cite{jukna2009set} who proved a similar result  while using the equivalent terminology of biclique edge covers of bipartite graphs.

\begin{theorem}
\label{theo:Alon}
Let $M$ be a $0,1$ matrix of size $n \times n$, $n \geq 2$, with at most $d$ zeros in each column. Then $\Rbool(M) \leq c(d) \log_2 n$ where $c(d) =  3e(d+1)/\log_2 e $.
\end{theorem}
\begin{proof}
If $d = 0$ the theorem clearly holds as the Boolean rank is then $1$. Thus, assume that $d \geq 1$.
Let $R$ be a subset of row indices of $M$, where for each $i$, row $i$ is chosen independently and randomly with probability $1/(d+1)$.
Let $C$ be all column indices $j$ for which $M_{i,j} = 1$ for each $i \in R$.
Let $S  = R \times C$  be the resulting subset of pairs $(i,j)$.
Note that $S$ is a monochromatic rectangle of ones in $M$, since all columns in $C$ have only ones in the rows in $R$.

What is the probability that entry $(i,j)$, for which $M_{i,j} = 1$, is included in the rectangle $S$?
The probability that row $i$ is selected initially into $R$  is $1/(d+1)$.
The pair $(i,j)$ will be in $S$ if for all $i'\in R$ it holds $M_{i',j} = 1$, and thus, $j$ will be included in $C$ (that is, all rows in $R$ have a $1$ in column $j$).
But since there are at most $d$ zeros in each column of $M$, then $(i,j)$ will be in $S$ if we didn't select any of at most $d$ rows which have zeros in column $j$,
and therefore, $(i,j)$ will be in $S$ with probability at least:
$$
\frac{1}{(d+1)}\left( 1 - \frac{1}{(d+1)}\right)^{d} \geq \frac{1}{e (d+1)}
$$
Now if we choose $t = c(d)\log n $ rectangles $S_1,...,S_{t}$ independently at random as described above,
the probability that $(i,j)$ is not included in any of these $t$ rectangles is at most:
$$
\left(1- \frac{1}{e(d+1)} \right)^t \leq e^{-t/(e(d+1))} = e^{-3\log_2 n/\log_2 e } = n^{-3}.
$$
Since there are at most $n^2$ ones in $M$, the probability that one of them is not covered by these $t$ rectangles is at most:
$$n^2 \cdot n^{-3} < 1.$$
Hence, there exists at least one choice of rectangles $S_1,...,S_t$ for which all ones of $M$ are covered, and thus, $\Rbool(M) \leq t = c(d)\log_2 n$ as claimed.
\end{proof}

As an immediate corollary we get an upper bound on the Boolean rank of $D_{n,k}$, since it has $k$ zeros in each column (and row).

\begin{corollary}
\label{coro-boolDnk}
Let $n > k \geq 1$ be integers. Then $\Rbool(D_{n,k}) \leq O(k\log n)$.
\end{corollary}

In Section~\ref{subsec-kronecker} we demonstrate how to generalize the argument of Alon~\cite{alon1986covering} to give an upper bound on the Boolean rank of
the $k$'th Kronecker power of a matrix with at most $d$ zeros in each column (row).
Another nice example is given by Rao and Yehudayoff~\cite{rao2020communication} who use the probabilistic method to prove an upper bound on the Boolean rank
of the $k$-disjointness problem studied in the context of communication complexity.

\subsection{Kernelization}
\label{Sec-kernel}
The kernelization method is a central technique in parameterized algorithms,
where the goal is to extract a small {\em kernel} of the original problem whose size depends only on a chosen parameter of the problem, such as the rank $d$ of a matrix,
so that solving the problem on this reduced instance yields a solution to the full problem.
Similar ideas appear in other fields of computer science and mathematics, where a large input is replaced by a compact core.
Examples include core-sets in computational geometry and clustering, where  a small subset of points approximates the full input set of points.
Similarly, in property testing algorithms a randomly chosen small "skeleton" serves as an approximate representation of the large object tested.
In these settings the reduced core or skeleton does not have to be fully equivalent to the original input, rather, it approximates it according to the criteria relevant to each application.
As we describe in Section~\ref{Sec-algorithms}, this approach was used to provide parameterized algorithms for the binary and Boolean rank,
as well as to design approximation and property testing algorithms for these ranks.

\begin{figure}[htb!]
\centering
\begin{tikzpicture}[>=latex, every node/.style={font=\large}]
\matrix (A) [matrix of math nodes,
             nodes in empty cells,
             left delimiter={(},
             right delimiter={)},
             row sep=3pt, column sep=3pt]
{
    |[draw=blue, thick]| 1 & 1 & 0 & 0 & |[draw=blue, thick]| 1 & 1 & 0 \\
    |[draw=blue, thick]| 1 & 1 & 0 & 0 & |[draw=blue, thick]| 1 & 1 & 0 \\
    0 & 1 & 1 & 1 & 0 & 0 & 0 \\
    0 & 0 & 0 & 1 & 0 & 1 & 0 \\
    |[draw=blue, thick]| 1 & 0 & 1 & 1 & |[draw=blue, thick]| 1 & 0 & 1 \\
    0 & 1 & 0 & 0 & 0 & 1 & 0 \\
};
\end{tikzpicture}
\;\;\;\;\;\;\;\;
\begin{tikzpicture}[>=latex, every node/.style={font=\large}]
\matrix (B) [matrix of math nodes,
             nodes in empty cells,
             left delimiter={(},
             right delimiter={)},
             row sep=3pt, column sep=3pt]{
    |[draw=blue, thick]| 1 & 1 & 0 & 0 & 1 & 0  \\
    0 & 1 & 1 & 1 & 0 & 0  \\
    0 & 0 & 0 & 1 & 1 & 0 \\
    |[draw=blue, thick]| 1 & 0 & 1 & 1 & 0 & 1 \\
    0 & 1 & 0 & 0 & 1 & 0 \\
};
\end{tikzpicture}
\caption{On the left a $0,1$ matrix $M$ with identical rows and columns, and on the right its kernel $M'$.
The blue rectangle covering some of the ones of $M'$ can be extended in an obvious way to cover the ones in the duplicate rows and columns removed from $M$.}
\label{fig:kernel}
\end{figure}

A simple common observation used by these results is that the Boolean and binary rank of a matrix $M$ does not change if we apply the following operations,
sometimes called  {\em kernelization rules}, to the rows and columns of $M$:
\begin{enumerate}
\item Remove from $M$ all-zero rows and columns.
\item Remove from $M$ duplicate rows and columns.
\end{enumerate}
Of course, the real rank also does not change by this modification of $M$. The resulting smaller sub-matrix is called the {\em kernel} of $M$.
See Figure~\ref{fig:kernel} for an illustration.

When considering the equivalent problem of finding a minimum biclique edge cover or partition for a bipartite graph $G$ whose reduced adjacency matrix is $M$,
this means to perform the following operations on the graph:
\begin{enumerate}
\item Remove all vertices with no neighbors (which corresponds to removing all-zero rows and columns in $M$).
\item Remove all vertices but one from any set of vertices with the same neighbor set (which corresponds to removing duplicate rows and columns in $M$).
A pair of vertices which have the same neighbor set are sometimes called {\em twin} vertices.
\end{enumerate}

The resulting sub-graph is called the {\em kernel} of $G$.
Thus, the number of rows and columns of the kernel of a matrix $M$ is the same as the number of vertices on each side of the kernel of a bipartite graph $G$.

The following simple lemma bounds the size of the kernel as a function of the Boolean or binary rank of $M$, and shows that the number of rows/columns of any $0,1$ matrix $M$
is at most $2^{\Rbin(M)}$ and at most $2^{\Rbool(M)}$.

\begin{lemma}
\label{lem:binary_size of matrix}
Let $M$ be a $0,1$-matrix of binary or Boolean rank  $d$.
Then $M$ has at most $2^d$ distinct rows and at most $2^d$ distinct columns.
\end{lemma}
\begin{proof}
If $\Rbool(M) = d$  then the  one entries of $M$ can be covered by $d$ monochromatic rectangles.
Any two rows that share a monochromatic rectangle must have one entries in the columns that belong to this rectangle.
Therefore, there are at most $2^d$ distinct rows in $M$ according to the monochromatic rectangles to which each row belongs.
A similar argument holds for the columns and for the binary rank.
\end{proof}

Nor et al.~\cite{nor2012mod} define a more elaborate kernelization for the biclique edge cover problem,
but the number of vertices on each side of the kernel sub-graph remains at most $2^d$.

Yet another bound can be achieved on the product of the number of distinct rows and distinct columns of a $0,1$ matrix $M$ as a function of the rank of $M$,
as proved by Bshouty~\cite{bshouty2023property}, using the following combinatorial lemma of Sgall~\cite{sgall1999bounds}.

\begin{lemma}[\cite{sgall1999bounds}]
Let $\cal A,\cal B$ be two families of subsets of $\{1,2,...,d\}$, such that for every $A \in \cal A$
and $B \in \cal B$ it holds that $|A \cap B| \leq s$. Then
$$|{\cal A}| \cdot |{\cal B}| \leq {d \choose \leq s }\cdot 2^d$$
where ${d \choose \leq s} = \sum_{i = 0}^d {d \choose i}$.
\end{lemma}

Bhsouty defined the $s$-binary rank, $br_s(M)$, of a $0,1$ matrix $M$, to be the minimal number of monochromatic rectangles
required to cover the ones of $M$, such that each one is covered by at most $s$ monochromatic rectangles.
Thus, $br_1(M)  = \Rbin(M)$. Using Sgall's lemma Bhsouty proved the following lemma, which is used in Section~\ref{SubSec-parameter}.

\begin{lemma}[\cite{bshouty2023property}]
\label{lem-product-rows-cols}
Let $M$ be a $0,1$ matrix with $br_s(M) \leq d$, and let $R(M)$ and $C(M)$ be the number of distinct rows and columns of $M$, respectively. Then:
$$
R(M) \cdot C(M) \leq {d \choose \leq s }\cdot 2^d.
$$
In particular, if $\Rbin(M) \leq d$, then $R(M)\cdot C(M) \leq (d+1)\cdot 2^d$.
\end{lemma}

The following lemma shows that it is possible to bound the size of the  kernel  also as a function of the real rank of the matrix.

\begin{lemma}
\label{clm:real_size of matrix}
Let $M$ be a $0,1$-matrix of real rank $d$. Then $M$ has at most $2^d$ distinct rows and at most $2^d$ distinct columns.
\end{lemma}
\begin{proof}
If $\Rreal(M) = d$ then $M$ has $d$ rows $R_1,...,R_d$, such that each other row of $M$ is a linear combination of $R_1,...,R_d$.
But then each column of $M$ is determined by a fixed subset of $d$ coordinates in $R_1,...,R_d$.
Since $M$ is $0,1$ matrix, then $M$ has at most $2^d$ options for these $d$ coordinates, and therefore, at most $2^d$ distinct columns.
A similar proof holds for the rows.
\end{proof}

The Boolean and binary rank are obviously bounded by the number of different rows or columns the matrix has, or in other words, by the size of the kernel of the matrix.
Therefore, Lemma~\ref{lem:binary_size of matrix} and Lemma~\ref{clm:real_size of matrix} provide the following immediate simple upper bounds.

\begin{corollary}
\label{coro-trivial-upper}
Let $M$ be a $0,1$ matrix. Then:
\begin{itemize}
\item $\Rbin(M) \leq 2^{\Rbool(M)}$.
\item $\Rbool(M), \Rbin(M) \leq 2^{\Rreal(M)}$.
\end{itemize}
\end{corollary}

For matrices with small constant real rank, a theorem of Parnas and Shraibman~\cite{parnas2025study} shows that the gap between the real and the binary and Boolean
rank is considerably smaller, and that the size of the kernel of such matrices is in fact smaller than $2^{\Rreal(M)}$. It follows from this theorem that
that the bound given in Theorem~\ref{theodoublegap} is in fact tight for matrices of size $n = 4,6,7$.

\begin{theorem}[\cite{parnas2025study}]
\label{theoremLowerBound}
Let $M$ be a $0,1$ matrix with $\Rreal(M) = d$.
\begin{itemize}
\item For $d = 3,4$ it holds that $d \leq \Rbin(M) \leq 2d-2$ and $d-1 \leq \Rbool(M), i(M) \leq 2d-2$.
\item
If $d = 3$ the kernel of $M$ has at most $4$ rows or columns, and if $d = 4$ the kernel has at most $7$ rows or columns.
\end{itemize}
\end{theorem}

\subsection{Communication complexity protocols}
\label{subsec-protocols}

The close connection between communication complexity and the Boolean and binary rank described in Section~\ref{subsec-communication},
makes it possible to find the Boolean and the binary rank of certain families of matrices using known results from communication complexity.
Theorem~\ref{theo-communication-partition} describes connections between the various  communication complexity measures and the Boolean and binary rank.
For example,  the non-deterministic communication complexity, $N_1(M)$, of a matrix $M$ is equal to $\lceil\log_2 \Rbool(M)\rceil$,
and therefore, the cost of any non-deterministic communication complexity protocol provides immediately an upper bound on the Boolean rank of the corresponding matrix.
Furthermore, the binary rank can be determined exactly by the unambiguous communication complexity,
or determined approximately by the deterministic communication complexity, $D(M)$, of $M$, since $\log_2 (\Rbin(M)) \leq D(M)  \leq O(\log^2 \Rbin(M))$.
We now describe two classic problems in communication complexity and show how they determine the Boolean and binary rank of the relevant matrices.

The first is the {\em non-equality} communication problem mentioned in the beginning of Section~\ref{Sec-lower-and-upper}.
Alice and Bob receive a row index $i$ and a column index $j$, respectively,
and their goal is to decide if $i \neq j$, while communicating as few bits as possible.
The matrix $C_n$ discussed previously, represents the non-equality function, and we already saw that $\Rbin(C_n) = n$ and $\Rbool(C_n) = \Theta(\log n)$.

Here is an alternative proof for an upper bound on the Boolean rank of $C_n$, which is based on a non-deterministic protocol between Alice and Bob.
Recall that in this case the two players are given some certificate,
and there exists a certificate and a transmission of bits between them where they both {\em accept} if and only if entry $(i,j)$ is $1$.
A possible certificate for the non-equality problem is the location of the bit on which $i$ and $j$ differ.
Given such a certificate, Alice and Bob can exchange the bit they each hold in this position.
If indeed $i \neq j$ and the certificate is a correct location of a bit on which $i$ and $j$ differ, they both accept.
In all other cases they would reject, as required.

As to the cost of this simple protocol, if the matrix $C_n$ is of size $n \times n$ then $i$ and $j$ can each be represented by $\log n$ bits,
and therefore, the location of a bit on which they differ can be represented by $\log \log n$ bits.
Thus, the length of the certificate is $\log \log n$, and then Alice and Bob exchange two more bits.
Hence, the non-deterministic communication complexity is $2+ \log \log n$ and thus, $\Rbool(C_n) = O(\log n)$.

The second communication problem is the {\em clique vs. independent set} problem defined by Yannakakis~\cite{Yannakakis91} in 1988.
Given some graph $G = (V,E)$, where $|V| = n$, Alice is given a subset of vertices $C \subseteq V$ which forms a clique in $G$,
and Bob is given a subset of vertices $I \subseteq V$ which is an independent set in $G$.
Their goal is to decide if $C \cap I \neq \emptyset$, in which case they accept, and otherwise, they reject.
Notice, that the intersection of a clique and an independent set contains at most one vertex, and thus, $|C \cap I| \in \{0, 1\}$.

The clique vs. independent set problem can be described by a matrix $M_G$, where there is a row for each clique of $G$ and a column for each independent set  in $G$,
and there is a $1$ in entry $i,j$
if and only if the clique and the independent set associated with row $i$ and column $j$, respectively, intersect. See Figure~\ref{fig:clique-independent} for an illustration.

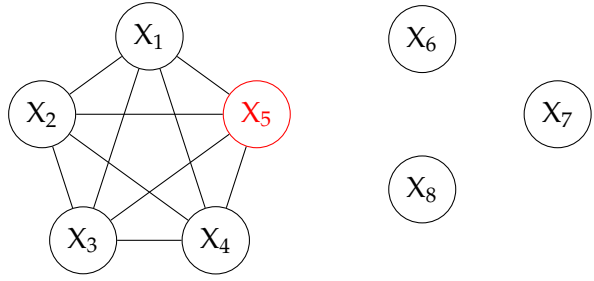
\begin{figure}[htb!]
\centering
\begin{tikzpicture}[
    vertex/.style={circle, draw, minimum size=6mm},
 shared/.style={circle, draw, color=red, linewidth=1.2pt, minimum size=6mm}
    every edge/.style={draw, thick}
]

\def\r{1.5}

\node[vertex] (c1) at (90:\r) {X$_{1}$};
\node[vertex] (c2) at (162:\r) {X$_{2}$};
\node[vertex] (c3) at (234:\r) {X$_{3}$};
\node[vertex] (c4) at (306:\r) {X$_{4}$};
\node[circle, draw, color=red] (v)  at (18:\r) {X$_{5}$}; 

\foreach \x/\y in {
    c1/c2,c1/c3,c1/c4,c1/v,
    c2/c3,c2/c4,c2/v,
    c3/c4,c3/v,
    c4/v}
{
    \draw (\x) -- (\y);
}
\node[vertex] (i1) at ($(v)+(2.2,1)$) {X$_{6}$};
\node[vertex] (i2) at ($(v)+(4.0,0)$) {X$_{7}$};
\node[vertex] (i3) at ($(v)+(2.2,-1)$) {X$_{8}$};
\end{tikzpicture}
\caption{On the left is a clique on $5$ vertices, $X_1, X_2, X_3, X_4,X_5$, and on the right an independent set which contains the vertices $X_5,X_6,X_7,X_8$.
If Alice gets the clique and Bob the independent set, they should output $1$, since the clique and the independent set share the vertex $X_5$, which is colored red.}
\label{fig:clique-independent}
\end{figure}

For this problem it is not even clear how the matrix $M_G$ looks in general, since it depends on the structure of the graph $G$,
and thus, solving this problem by presenting a communication protocol is a convenient option for determining the Boolean and binary rank of $M_G$.

A non-deterministic protocol for this problem can proceed as follows. Alice and Bob receive as a certificate the name of the vertex $v$ on which $C$ and $I$ intersect,
and then they can each check if $v$ is indeed contained in the clique and the independent set they hold, respectively, and send to each other one more bit with their decision.

As to the cost of this protocol, if the number of vertices is $n$, then a vertex can be described by $\log n$ bits, and hence, the non-deterministic communication complexity
is $2 + \log n$ bits. Hence, $\Rbool(M_G) = O(2^{2 + \log n}) = O(n)$.
Furthermore, this non-deterministic protocol is unambiguous, since $C \cap I$ contains at most one vertex,
and thus, there is at most one certificate that will convince Alice and Bob that $C \cap I \neq \emptyset$. Hence, $\Rbin(M_G) = O(n)$ as well.

As to a lower bound on the Boolean and binary rank of $M_G$, the matrix $M_G$ contains as a sub-matrix the identity matrix of size $n\times n$,
since there are $n$ rows and columns in $M_G$ which correspond to cliques and independent sets of size $1$, one row and column for each vertex $v \in V$.
Therefore, $\Rbool(M_G), \Rbin(M_G) \geq n$.

It is also possible to describe a partition of the ones of $M_G$ into $n$ rectangles, which follows directly from the non-deterministic protocol presented above.
Define a rectangle $R_v$ for each $v \in V$, where $R_v$ contains all the one entries of $M_G$ for which the corresponding clique $C$ and independent set $I$ intersect on vertex $v$.
The rectangles are disjoint since each pair $C,I$ intersects in at most one vertex, and it is not hard to verify that they are indeed monochromatic rectangles.

We now describe a deterministic communication protocol for the clique vs. independent set problem given by Yannakakis~\cite{Yannakakis91},
whose deterministic communication complexity is $O(\log ^2 n)$.
The deterministic protocol has $\log n$ rounds, where in each round $O(\log n)$ bits are exchanged between Alice and Bob, for a total of $O(\log ^2 n)$ bits.
At the beginning of the protocol all vertices in $V$ are {\em candidates} which can belong to $C \cap I$, and
in each round Alice and Bob disqualify at least half of the remaining candidates.
Thus, after at most $\log n$ rounds there is at most one candidate and they can easily decide if $C \cap I \neq \emptyset$ by exchanging two more bits.

The complete protocol is described below and from it follows immediately the following upper bound on the deterministic communication complexity of the clique vs. independent set problem.

\begin{theorem}[\cite{Yannakakis91}]
The deterministic communication complexity of the clique vs. independent set problem of a graph on $n$ vertices is at most $O(log^2 n)$.
\end{theorem}

\bigskip
\begin{center}
\fbox{
\begin{minipage}{6in}
\leftline{\bf Protocol for Clique vs. Independent set $(C,I, G = (V,E))$}
\begin{enumerate}
\item
Let  $ V$ be the initial set of candidate vertices that can belong to $C \cap I$.
\item
While $|V| > 1$:
\begin{enumerate}
\item
Alice checks if there exists a vertex $v \in C$ with at most $|V|/2$ neighbors in $V$.
\begin{itemize}
\item
If so she sends $v$ to Bob, and if $v \in I$ then Bob sends $1$ to Alice and they can both stop and accept.
\item
Otherwise, remove from $V$ all vertices which are not neighbors of $v$ (the vertices removed cannot be in the clique $C$).
\end{itemize}
\item
If Alice failed, Bob checks if there exists a vertex $u \in I$ with least $|V|/2$ neighbors in $V$.
\begin{itemize}
\item
If so he sends $u$ to Alice, and if $u \in C$ then Alice sends $1$ to Bob and they can both stop and accept.
\item
Otherwise, remove from $V$  all neighbors of $u$ (the vertices removed cannot be in the independent set $I$).
\end{itemize}
\item
Otherwise, if both Alice and Bob failed, then all vertices $v \in C$ have more than $|V|/2$ neighbors in $V$,
and all vertices $u \in I$ have less than $|V|/2$ neighbors in $V$. Hence, $C \cap I = \emptyset$ and Alice and Bob can stop and reject.
\end{enumerate}
\item If Alice and Bob didn't accept or reject so far, then $V = \{v\}$, and they can check if $v \in C, v \in I$ and accept or reject accordingly.
\end{enumerate}
\end{minipage}
}
\end{center}
\bigskip

Yannakakis~\cite{Yannakakis91} also provided a reduction from any given communication problem to the clique vs. independent set problem.
This reduction and the above protocol yield an upper bound on the deterministic communication complexity, $D(M)$, of any $0,1$ matrix $M$ in terms of its partition number,
as described by the next theorem. See Figure~\ref{fig-rec-reduction} for an illustration of the reduction.

\begin{theorem}[\cite{Yannakakis91}]
\label{theo-rec-reduction}
For any  $0,1$ matrix $M$, $D(M) \leq O(\log ^2 Partition_1(M))$.
\end{theorem}
\begin{proof}
Let $m = Partition_1(M)$ be the minimal number of monochromatic rectangles required to partition the ones of $M$. 
Define the following graph $G_M$ on $m$ vertices. There is a vertex in $G_M$ for each rectangle in the partition of the ones of $M$.
Two vertices  will be connected by an edge in $G_M$ if and only if the corresponding rectangles share a row in $M$.
Note that any set of rectangles in $M$ which share a row are mapped to a clique in $G_M$, and any set of rectangles which share a column are mapped to an independent set of $G_M$,
since in a partition rectangles which share a column  cannot share a row.

Assume that Alice gets a row index $i$ and Bob gets a column index $j$ of $M$.
Alice can map her index to the clique $C$ in $G_M$ represented by the rectangles in row $i$,
and Bob can map his column index to the independent set $I$ represented by the rectangles in column $j$.
Now they can use the deterministic protocol described above for the clique vs. independent set to decide if $C$ and $I$ intersect.
If so, this means that there exists a rectangle in the partition of the ones of $M$, which includes entry $M_{i,j}$, and therefore, $M_{i,j} = 1$, and they can accept.
Otherwise, they reject.
The deterministic communication complexity of this protocol is $O(\log^2 m)$ as claimed.
\end{proof}

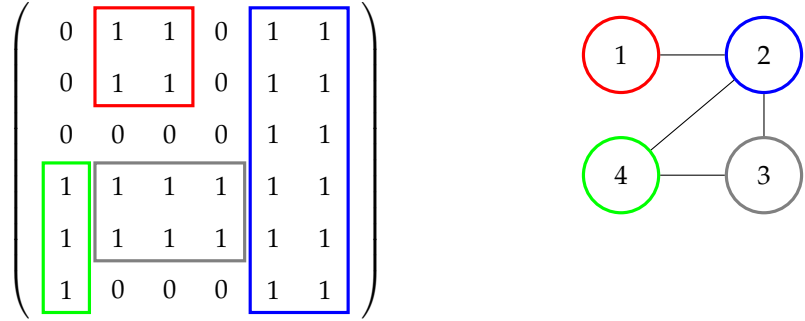
\begin{figure}[htb!]
\centering
\begin{tikzpicture}[every node/.style={font=\small}]

\matrix (A) [matrix of nodes,
             left delimiter={(},
             right delimiter={)},
             nodes={minimum width=6mm, minimum height=6mm, anchor=center},
             column sep=2.5pt,
             row sep=2.5pt] {
  0 & 1 & 1 & 0 & 1 & 1 \\
  0 & 1 & 1 & 0 & 1 & 1 \\
  0 & 0 & 0 & 0 & 1 & 1 \\
  1 & 1 & 1 & 1 & 1 & 1 \\
  1 & 1 & 1 & 1 & 1 & 1 \\
  1 & 0 & 0 & 0 & 1 & 1 \\
};

\draw[green, very thick]
    (A-4-1.north west) rectangle (A-6-1.south east);
\draw[blue, very thick]
    (A-1-5.north west) rectangle (A-6-6.south east);
\draw[red, very thick]
    (A-1-2.north west) rectangle (A-2-3.south east);
\draw[gray, very thick]
    (A-4-2.north west) rectangle (A-5-4.south east);

\coordinate (G0) at ([xshift=3.5cm]A.east);

\node[circle, draw=red,   very thick, minimum size=10mm] (v1) at ($(G0)+(0, 1.4)$) {1};
\node[circle, draw=blue,  very thick, minimum size=10mm] (v2) at ($(G0)+(1.9,1.4)$) {2};
\node[circle, draw=gray,  very thick, minimum size=10mm] (v3) at ($(G0)+(1.9,-0.2)$) {3};
\node[circle, draw=green, very thick, minimum size=10mm] (v4) at ($(G0)+(0,-0.2)$) {4};

\draw (v1) -- (v2);
\draw (v2) -- (v3);
\draw (v4) -- (v3);
\draw (v4) -- (v2);
\end{tikzpicture}
\caption{On the left is a matrix with a partition of the ones into $4$ rectangles, each represented with a different color.
On the right is the corresponding clique vs. independent set graph according to the reduction described in Theorem~\ref{theo-rec-reduction}.}
\label{fig-rec-reduction}
\end{figure}

For many years it was not clear if $O(\log^2 n)$ is a tight upper bound for the deterministic communication complexity of the clique vs. independent set of a graph with $n$ vertices.
By the reduction described in Theorem~\ref{theo-rec-reduction}, any lower bound on the deterministic communication complexity of a given matrix $M$ in terms of the partition number of $M$,
provides immediately a lower bound for the clique vs. independent set problem.
Finally, Goos,  Pitassi and Watson~\cite{Goos} used the lifting technique described in the next subsection and
proved Theorem~\ref{theo-deterministic-partition}, thus, showing that
there can be an almost quadratic gap between $\log Partition_1(M) $ and the deterministic communication complexity of $M$.
This theorem implies a lower bound of $\widetilde{\Omega}(\log^2 n)$ on the deterministic communication  complexity of
the clique vs. independent set problem for a graph with $n$ vertices, thus, providing an almost matching lower bound to the upper bound given by Yannakakis~\cite{Yannakakis91}.

Another problem which was open for many years is the co-non-deterministic communication complexity of
the clique vs. independent set problem. While it is easy to provide a small certificate for a non-deterministic protocol for this problem, as we saw above,
it was not clear if a small certificate is possible for a co-non-deterministic protocol for this problem.
See Section~\ref{Subsec-complemement} for a discussion on the lower bounds given over the years for the co-non-deterministic communication complexity of
the clique vs. independent set problem. These results are closely related to the maximal gap possible between the binary rank of a matrix and its complement,
as well as to the Alon, Saks, Seymour conjecture.

\subsection{The Lifting Technique}
\label{subsec-lifting}
The query to communication lifting technique is a powerful method in computer science which allows one to {\em lift} lower bounds from the easier model of query complexity
to the more complex model of communication complexity.
In the {\em query model} the algorithm has to compute $f(z)$ for some function $f:\{0,1\}^n \rightarrow \{0,1\}$, while querying as few bits as possible of the input $z$.
The queries of the algorithm can be described by a decision tree as follows.
The root and each internal node of the tree contain some bit $z_i$ of $z$, according to the query the algorithm makes in that round,
and the value of $z_i$ determines if the algorithm continues to the left or right sub-tree of that node.
The leaves contain the output of the algorithm, that is, the value of $f(z)$ corresponding to the answers the algorithm received for its queries.
The complexity of the algorithm is the length of the longest path from the root to a leaf in its query tree, and the query complexity $Query(f)$ of a function $f$
is the complexity of the algorithm with the smallest query complexity for $f$.

A communication complexity protocol can also be formulated as a tree.
Assume that our players are Alice and Bob, where Alice has some input $x$ and Bob has some input $y$,
and their goal is to compute $F(x,y)$ for some function $F:\{0,1\}^n \times \{0,1\}^n \rightarrow \{0,1\}$.
If  Alice is the first to send a bit to Bob, then the root of the tree contains the first bit sent by Alice, and according to the value of this bit
there are two internal nodes which correspond to the first bit sent by Bob, and so on.
The leaves contain the value of $F(x,y)$ according to the bits sent throughout the communication between Alice and Bob.
Thus, the deterministic communication complexity $D(F)$ of $F$, is bounded above by the height of this tree. See Figure~\ref{fig-tree}.

\forestset{
  decision/.style={
    circle,
    draw,
    inner sep=2pt,
    minimum size=10pt,
    align=center,
  },
  leaf/.style={
    rectangle,
    draw,
    inner sep=6pt,
    minimum width=8mm,
    align=center,
  },
  mytree/.style={
    for tree={
      grow'=south,
      parent anchor=south,
      child anchor=north,
      edge={->,line width = 1pt},
      l sep=12mm,
      s sep=7mm,
      anchor=center,
      font=\small
    }
  }
}

\begin{figure}[htb!]
\begin{center}
\begin{forest}
  mytree
  [{$z_2$}, decision
    [{$z_4$}, decision, edge label={node[midway,left]{0}}
      [ {1}, leaf, fill=green!40, edge label={node[midway,left]{0}} ]
      [ {$z_3$}, decision, edge label={node[midway,right]{1}}
        [ {0}, leaf, fill=red!40, edge label={node[midway,left]{0}} ]
        [ {1}, leaf, fill=green!40, edge label={node[midway,right]{1}} ]
      ]
    ]
    [{$z_5$}, decision, edge label={node[midway,right]{1}}
      [ {0}, leaf, fill=red!40, edge label={node[midway,left]{0}} ]
      [ {$z_7$}, decision, edge label={node[midway,right]{1}}
        [ {1}, leaf, fill=green!40, edge label={node[midway,left]{0}} ]
        [ {0}, leaf, fill=red!40, edge label={node[midway,right]{1}} ]
      ]
    ]
  ]
\end{forest}
\;\;\;\;\;\;\;\;
\begin{forest}
  mytree
  [{$Alice$}, decision
    [{$Bob$}, decision, edge label={node[midway,left]{0}}
      [ {1}, leaf, fill=green!40, edge label={node[midway,left]{0}} ]
      [ {$Alice$}, decision, edge label={node[midway,right]{1}}
        [ {0}, leaf, fill=red!40, edge label={node[midway,left]{0}} ]
        [ {1}, leaf, fill=green!40, edge label={node[midway,right]{1}} ]
      ]
    ]
    [{$Bob$}, decision, edge label={node[midway,right]{1}}
      [ {0}, leaf, fill=red!40, edge label={node[midway,left]{0}} ]
      [ {$Alice$}, decision, edge label={node[midway,right]{1}}
        [ {1}, leaf, fill=green!40, edge label={node[midway,left]{0}} ]
        [ {0}, leaf, fill=red!40, edge label={node[midway,right]{1}} ]
      ]
    ]
  ]
\end{forest}
\end{center}
\caption{On the left is a decision tree in the query model and on the right a communication complexity protocol between Alice and Bob described by a tree.}
\label{fig-tree}
\end{figure}
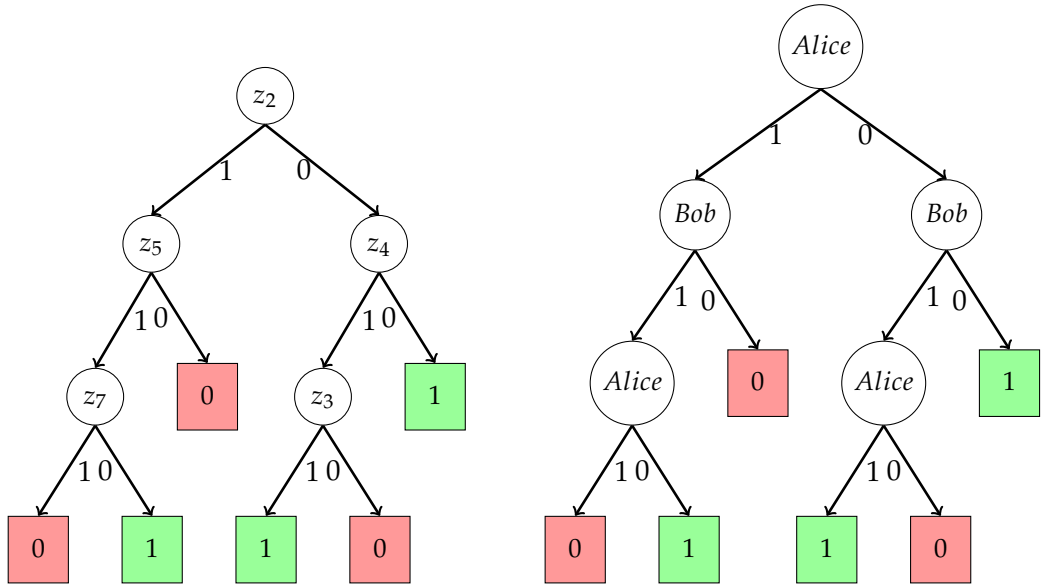

There are functions for which there is a considerable difference in the complexity of the two models.
For example, consider the parity function of a $0,1$ vector $z = (z_1,...,z_n)$:
$$Parity(z) =  z_1 \oplus z_2 \oplus ... \oplus z_n.$$
The query complexity of $Parity(z)$ is $n$,
since as long as the algorithm doesn't query all bits of $z$, the last bit queried can change the parity of the bits queried so far.
However, consider the scenario where Alice and Bob partition between them the bits of $z$, such that Alice gets the first half  of the bits, $x$, and Bob gets the second half, $y$.
Then they communicate to compute together the function $PARITY(x,y)$ defined as:
$$PARITY(x,y) =  (x_1 \oplus x_2 \oplus ... \oplus x_n) \oplus (y_1 \oplus y_2 \oplus ... \oplus y_n).$$
A possible communication protocol for the $PARITY$ function can proceed as follows.
Alice computes $Parity(x)$ and sends this one bit to Bob, and Bob computes $Parity(y)$ and sends this one bit to Alice.
But $PARITY(x,y) = Parity(x) \oplus Parity(y)$, and therefore, each of them knows $PARITY(x,y)$
after communicating only $2$ bits.

The idea of the {\em lifting technique} is to define some complex enough gadget function $g$ and some function $F$ which depends on both $g$ and on the
function $f$ computed by the query algorithm, such that when Alice and Bob want to compute  $F(x,y)$ they must in fact follow the steps of the query algorithm which computes $f$,
and in each such step they must communicate many bits in order to find the value of the gadget function $g$.

Specifically, the gadget is defined as a function $g:\{0,1\}^a \times \{0,1\}^b \rightarrow \{0,1\}$,
and  the function $F:\{0,1\}^{na} \times \{0,1\}^{nb} \rightarrow \{0,1\}$ is defined as:
$$F(x,y) = f(g(x_1,y_1),...,g(x_n,y_n)),$$
where the inputs $x_i,y_i$ transferred to the gadget $g$ are not necessarily bits, but can be blocks of bits.
The function $F$ is sometimes denoted as $F = f\circ g^n$.

A general lifting theorem will show that if Alice and Bob want to compute such a function $F(x,y)$, where
computing $g(x_i,y_i)$ requires them to communicate $C$ bits and the query complexity of $f$ is $Query(f)$,
then the best communication protocol for $F$ requires  them to communicate at least $C \cdot Query(f)$ bits.
Thus, the main ingredient in a proof of such a lifting theorem is to show that the gadget $g$ is complex enough in the sense that if Alice and Bob want to know the $i$'th bit, $z_i = g(x_i,y_i)$,
of the input of $f$, they must send to each other many bits in order to compute $g(x_i,y_i)$.

Note that for the parity function $f = Parity(z)$, a gadget such as $g(x_i,y_i) = x_i \oplus y_i = z_i$, which depends on one bit of Alice and one bit of Bob, will not work,
since as we saw:
$$F(x,y) = Parity(g(x_1,y_1),...,g(x_n,y_n)) = Parity(Parity(x) \oplus Parity(y)).$$
Thus, in this case Alice and Bob do not have to actually compute the value of  $g(x_i,y_i)$ for many $i$'s,
but can compress their inputs in advance and communicate only $2$ bits as explained above.

The heart of a lifting theorem is, therefore, to define a {\em hard} gadget $g$, and show that computing $F(x,y) = f(g(x_1,y_1),...,g(x_n,y_n))$
requires all communication complexity protocols to follow, in fact, the query process of the best query complexity algorithm for $f$,
such that whenever the query algorithm wants to know the $i$'th bit of $f$,
Alice and Bob must compute $g(x_i,y_i)$, and this requires them to communicate many bits.

Raz and McKenziee~\cite{RazM99} were the first to use this idea in 1997, which was then called a simulation argument.
Goos,  Pitassi and Watson~\cite{Goos} generalized their argument and proved the following general lifting theorem:

\begin{theorem}[\cite{goos2018deterministic}]
\label{theo-goos}
There is a gadget $g$ whose deterministic communication complexity is $\Theta(\log n)$, such that for all $f:\{0,1\}^n \rightarrow \{0,1\}$,
the deterministic communication complexity, $D(F)$, of $F = f\circ g^n$ is:
$$D(F) = Query(f) \cdot \Theta(\log n).$$
\end{theorem}
Proving that $D(f\circ g^n) \leq Query(f) \cdot \Theta(\log n)$ is easy, since Alice and Bob can simply simulate the best query algorithm for $f$,
where every time the query algorithm wants to query the $i'th$ bit of the input to $f$, Alice and Bob communicate $\Theta(\log n)$ bits to decide what $g(x_i,y_i)$ is.
The lower bound requires to show that if there was a better communication complexity protocol for $F = f\circ g^n$ then
there would be a query algorithm which decides $f$ with a lower query complexity.

Recall that the binary rank is closely related to the deterministic communication complexity and the Boolean rank determines the non-deterministic communication complexity.
G\"{o}\"{o}s, Pittasi and Watson~\cite{Goos} proved that there exists a matrix $M$ whose deterministic communication complexity is at least $\widetilde{\Omega}(\log^2 \Rbin(M))$
(see Theorem~\ref{theo-deterministic-partition}, Section~\ref{subsec-communication}). Their proof  uses the lifting technique and Theorem~\ref{theo-goos} mentioned above.
In Section~\ref{Subsec-complemement} we describe how the lifting technique and the connection to communication complexity were used in a series of works
which resulted in a quasi polynomial gap between the binary rank of a matrix and the Boolean rank of its complement.

\section{Mathematical properties of the rank functions}
\label{Sec-properties}

The real rank of a matrix is a well understood concept, supported by a rich collection of  algebraic tools and properties, which makes it
highly useful in many applications in mathematics and computer science.
For example, it is known that $\Rreal(M)$ equals the minimum size of a spanning set among the rows (columns) of $M$ and also equals the maximum size of an independent set among the rows (columns) of $M$.
Furthermore, the real rank satisfies fundamental properties, including:
\begin{itemize}
\item {\bf Sub-additivity:} $\Rreal(N + M) \leq \Rreal(N) + \Rreal(M)$.

\item {\bf Multiplicity under Kronecker product:} $\Rreal(N \otimes M) = \Rreal(N)\cdot \Rreal(M)$.

\item {\bf Base spanning property:} Every base for the rows (columns) of $M$ spans all other bases of the rows (columns) of $M$.

\item {\bf Full rank sub-matrix:} Every matrix $M$ of real rank $d$ contains a $d \times d$ sub-matrix  with real rank $d$.
\end{itemize}

In contrast, the behavior of the Boolean and binary rank remains less understood, although impressive progress was made over the years.
Given their importance in diverse applications,
it is essential to systematically investigate the properties of these rank functions and develop effective algorithms to approximate them.
Accordingly, prior research has examined if the basic properties of the real rank extend to
the Boolean and binary rank and under what conditions. These results help to clarify the limits of the linear algebraic tools,
and also sharpen the differences between these rank functions and the real rank.
In this section we survey some of these results.

\subsection{Spanning sets and Bases}
\label{Sec-base}

The concepts of a spanning set and a base are fundamental concepts in algebra, in particular, when discussing the rank of a matrix over the reals.
It is possible to define these concepts over any semiring equipped with two operations, addition and multiplication.
In our context this means that when discussing a $0,1$ matrix $M$,  a {\em base} is, as for the real rank, a minimal subset of vectors that spans the columns
of $M$, but here the vectors are $0,1$, the coefficients in a linear combination are restricted to $0,1$,
 and the addition and multiplication are the binary or Boolean operations, depending on the rank.

Cohen and Rothblum showed~\cite{Cohen} that the nonnegative rank of a matrix $M$ is $d$ if and only if there exists a base of size $d$ which spans the columns of $M$.
A similar claim holds for the binary and Boolean rank.

\begin{lemma}
The Boolean or binary rank of a $0,1$ matrix $M$ is $d$ if and only if there exists a $0,1$ base of size $d$ which spans the columns of $M$ under the binary or Boolean operations, respectively.
\end{lemma}
\begin{proof}
It is easy to verify that if $M = X\cdot Y$ is an optimal binary or Boolean decomposition of $M$,
then the columns of $X$ are a binary or Boolean base, respectively, for the columns of $M$, and the columns of $Y$ are the coefficients in the corresponding linear combination.
The optimality of the decomposition guarantees that this is a base.
\end{proof}

Parnas and Shraibman~\cite{parnas2018augmentation} defined the concept of the {\em Base Graph} of a matrix.
The vertices of this graph are the different bases of the matrix $M$ for a given rank function,
and there is a directed edge from base $U$ to base $V$ if $U$ spans $V$.
Whereas this graph has a trivial structure for the real rank, since any base of  $M$ spans all other bases of $M$,
it has interesting properties for the Boolean and binary rank. Specifically it is always acyclic.

\begin{lemma}[\cite{parnas2018augmentation}]
The base graph of a $0,1$ matrix $M$ is always acyclic, for both the binary and the Boolean rank.
\end{lemma}

The base graph is used in~\cite{parnas2018augmentation} to study the {\em augmentation property}.
Given a matrix $M$ and a column vector $x$, denote by $(M|x)$
the matrix resulting by augmenting $M$ with $x$ as the last column.
Denote by $(M|x_1,...,x_t)$ the matrix resulting  by augmenting $M$  with the column vectors $x_1,...,x_t$.

\begin{definition}
A matrix $M$ has the {\em Augmentation property} for a given rank $R()$ if
for any set of column vectors $x_1,...,x_t$ for which $R(M|x_i) = R(M)$, for
$1\leq i \leq t$, it holds that $R(M|x_1,...,x_t)= R(M)$.
\end{definition}

The following simple lemma shows that the augmentation property holds for the real rank for any matrix $M$. 

\begin{lemma}
\label{sumofrowaugmented}
Let $M$ be a real matrix and $x$ a column vector.
If $\Rreal(M) = \Rreal(M|x)$ then any dependence between the rows of $M$ is preserved in $(M|x)$.
A similar claim holds when the role of the columns and rows is reversed.
\end{lemma}

However, the augmentation property does not hold in general for the Boolean or binary  rank. The following example is from~\cite{parnas2018augmentation}.
$$
(M|x,y)=
\left(
     \begin{array}{ccc|cc}
 1 & 1 & 1 &0&0\\
 1 & 1 & 1 &1&0\\
 0 & 1 & 1 &1&1\\
 0 & 0& 1  &0&0\\
\end{array}
\right)
$$
The boolean rank of $M$ remains $3$, after augmenting it with $x$ or $y$ separately.
However, when $M$ is augmented with both $x$ and $y$, it is easy to verify that the Boolean rank increases to $4$.
The following theorem proves that a necessary and sufficient condition for the Augmentation property to hold
for the Boolean or binary rank, is having a (unique) base that spans all other bases of the matrix,
or in other words, the Base graph of the matrix should have exactly one source.

\begin{theorem}[\cite{parnas2018augmentation}]
A $0,1$ matrix $M$ has the Augmentation property for the Boolean (binary)
rank if and only if  $M$ has a (unique) base that spans all other bases of $M$ under the Boolean (binary) operations.
\end{theorem}

\subsection{The full rank sub-matrix}
\label{Subsec-fullrank}
It is well known that if $\Rreal(M) = d$, then $M$ contains a $d \times d$ sub-matrix whose real rank is exactly $d$.
It was shown already in 1981 by de Caen, Gregory and Pullman~\cite{Caen2} that this property does not hold for the Boolean rank.
They showed that the matrix $C_7$ has Boolean rank $5$, whereas any $5 \times 5$ sub-matrix of $C_7$ has Boolean rank at most $4$.
The following lemma generalizes this argument for the family of matrices $C_n$.
Recall that $\sigma(n) = \min \left\{d\;|\; n \leq {d \choose \lfloor d/2\rfloor}\right\}$,
and that by Corollary~\ref{coro-Cn}, $\Rbool(C_{n}) = \sigma(n)$. See Section~\ref{Subsec-combinatorics}.

\begin{lemma}
\label{lem:crown}
Let $ n =  {k \choose \lfloor k/2\rfloor} + 1$.
Then $\Rbool(C_n) = k+1$ and every proper sub-matrix of $C_n$ has Boolean rank at most $k$.
\end{lemma}
\begin{proof}
By Corollary~\ref{coro-Cn}, it holds that $\Rbool(C_n) = k+1$.
Consider any proper sub-matrix of $C_n$.
If this is a principal sub-matrix of $C_n$ then by Corollary~\ref{coro-Cn} its Boolean rank is at most $k$,
since it is equal to $C_m$ for some  $m \leq {k \choose \lfloor k/2\rfloor}$.

Otherwise, consider a sub-matrix $M$ that is achieved from $C_n$ by removing one of the rows of $C_n$.
The sub-matrix $M$ is the concatenation of the matrix $C_{n-1}$ and the all-one column vector (maybe after rearranging rows).
But any cover of the ones of $C_{n-1}$ by monochromatic rectangles can be easily extended to cover
the all-one column of $M$ without increasing the number of rectangles.
To see this note that each row of $C_{n-1}$ contains at least one $1$, and thus,
the cover contains a rectangle which covers ones in this row, and this rectangle can be extended to cover the corresponding ones in the all-one column which share rows with this rectangle.
Hence, $\Rbool(M) = \Rbool(C_{n-1}) = \sigma(n-1) \leq k$.
A similar argument holds if we remove one of the columns of $C_n$.
\end{proof}

In fact, since $\Rbool(C_n) = \sigma(n) = \Theta(\log n)$, we can actually get a stronger result.

\begin{lemma}
\label{lem:submatrixCn}
Every sub-matrix of $C_n$ of size $ \sigma(n) \times  \sigma(n)$ has Boolean rank at most $O(\log \log n)$.
\end{lemma}
\begin{proof}
There is only one zero in each row and column of $C_n$. Therefore, any sub-matrix of $C_n$ contains at most one zero in each row and column.
Thus, applying Theorem~\ref{theo:Alon} to any sub-matrix of $C_n$ of size $\sigma(n) \times \sigma(n) = O(\log n \times \log n)$,
we can conclude that its Boolean rank is at most $O(c\times \log \log n)$, for some constant $c$.
\end{proof}

Hrube{\v{s}~\cite{hrubevs2025hard}  showed in 2025 that the full rank sub-matrix property doesn't hold for the binary rank.
He gave  an example of a $5 \times 6$ matrix whose binary rank is $5$, but every sub-matrix of size $5 \times 5$ of this matrix has binary rank at most $4$.
Theorem~\ref{lem-binary-fullrank} generalizes this construction and presents a family of matrices of size $(n+2)\times (n+3)$, $n \geq 3$,
such that each proper sub-matrix has a strictly smaller binary rank than that of the full matrix.
See Figure~\ref{fig:hrubevs} for an illustration.

\begin{figure}[htb!]
\begin{center}
$
\left(
\begin{array}{cccccc|ccc}
\textcolor{blue}{\bf 1} & 0 & 0 & \cdots & 0 & 0 & 0 & 1 & \textcolor{blue}{\bf 1} \\
0 & \textcolor{green}{\bf 1} & 0 & \cdots & 0 & 0 & 1 & 0 & 1  \\
0 & 0 & \textcolor{brown}{\bf 1} & 0 & \cdots & 0 &  1 &  1 & 0 \\
\cline{7-9}
\vdots & \vdots &  & \ddots & \vdots & \vdots &  1 &  1 & 1 \\
\vdots & \vdots &  &  & \ddots & \vdots & \vdots & \vdots & \vdots \\
0 & 0 & 0 & \cdots & 0 & 1 &  1&  1 & 1 \\
\hline
0 & 0 & 0 & \cdots & 0 & 0 & \textcolor{red}{\bf 1} & \textcolor{red}{\bf 1} & \textcolor{red}{\bf 1} \\
\hline
\textcolor{blue}{\bf 1} & \textcolor{green}{\bf 1} & \textcolor{brown}{\bf 1} & \cdots & 1 & 1 & 1 & 1 & \textcolor{blue}{\bf 1}
\end{array}
\right)
$
\;\;\;\;\;\;
$
\left(
\begin{array}{cccccc|ccc}
\textcolor{blue}{\bf 1} & 0 & 0 & \cdots & 0 & 0 & 0 & \textcolor{blue}{\bf 1} & \textcolor{cyan}{\bf 1} \\
0 & \textcolor{green}{\bf 1} & 0 & \cdots & 0 & 0 & \textcolor{green}{\bf 1} & 0 & \textcolor{cyan}{\bf 1}  \\
0 & 0 & \textcolor{brown}{\bf 1} & 0 & \cdots & 0 &  \textcolor{red}{\bf 1} &  \textcolor{red}{\bf 1} & 0 \\
\cline{7-9}
\vdots & \vdots &  & \ddots & \vdots & \vdots &  \textcolor{red}{\bf 1} &  \textcolor{red}{\bf 1} & \textcolor{cyan}{\bf 1} \\
\vdots & \vdots &  &  & \ddots & \vdots & \vdots & \vdots & \vdots \\
0 & 0 & 0 & \cdots & 0 & 1 &  \textcolor{red}{\bf 1}&  \textcolor{red}{\bf 1} & \textcolor{cyan}{\bf 1} \\
\hline
0 & 0 & 0 & \cdots & 0 & 0 & \textcolor{red}{\bf 1} & \textcolor{red}{\bf 1} & \textcolor{cyan}{\bf 1} \\
\hline
\textcolor{blue}{\bf 1} & \textcolor{green}{\bf 1} & \textcolor{brown}{\bf 1} & \cdots & 1 & 1 & \textcolor{green}{\bf 1} & \textcolor{blue}{\bf 1} & \textcolor{cyan}{\bf 1}
\end{array}
\right)
$
\end{center}
\caption{A generalization of the construction of~\cite{hrubevs2024hard}. On the left is an attempt to partition the ones of the matrix $M$ of size $(n+2) \times (n+3)$
into only $n+1$ rectangles.
Since the first $n+1$ ones on the main diagonal of $M$ are an isolation set and must belong to different rectangles, such a partition forces the partial partition described, where
each rectangle is described with a different color. The rightmost one on the second row of $M$ cannot belong to any of these $n+1$ rectangles.
Thus, $\Rbin(M) = n+2$, since each row can be covered by a different rectangle.
\newline
On the right we describe an alternative partition of the ones of $M$ into $n+2$ rectangles, which is used in the proof of Theorem~\ref{lem-binary-fullrank}.}
\label{fig:hrubevs}
\end{figure}

\begin{theorem}
\label{lem-binary-fullrank}
Let $M$ be the $(n+2)\times (n+3)$ matrix in Figure~\ref{fig:hrubevs}, where $n \geq 3$.
Then $\Rbin(M) = n+2$ whereas every $(n+2)\times (n+2)$ sub-matrix of $M$ has binary rank at most $n+1$.
\end{theorem}
\begin{proof}
First note that the first $n+1$ ones on the main diagonal of $M$ form an isolation set of size $n+1$.
Thus, each of these ones should be in a separate monochromatic rectangle of ones in any partition of the ones of $M$.
Denote these rectangles by $R_1,...,R_{n+1}$, where $R_i$ contains the $i'th$ one on the main diagonal of $M$.

Assume by contradiction that $\Rbin(M) = n+1$, and that $R_1,...,R_{n+1}$ are the only rectangles in the partition.
Note that by the structure of $M$ we must have:
$M_{1,n+3}, M_{n+2,1} \in R_1$ and so also $M_{n+2,n+3} \in R_1$.
Furthermore, $M_{n+2,i} \in R_i$ for $2 \leq i \leq n$, and $M_{n+1,n+2}, M_{n+1,n+3}\in R_{n+1}$.
But then the one in position $M_{2,n+3}$ cannot belong to any of these $n+1$ rectangles and we get a contradiction.
Thus, $\Rbin(M) > n+1$, and so $\Rbin(M) = n+2$, since the ones in each row can be covered by a single rectangle.

An alternative partition into $n+2$ rectangles $R_1,...,R_{n+2}$ is shown in Figure~\ref{fig:hrubevs} and it will be needed shortly:
$$R_1 = \{M_{1,1},M_{1,n+2}, M_{n+2,1}, M_{n+2,n+2}\}$$
$$R_2 = \{M_{2,2},M_{2,n+1}, M_{n+2,2}, M_{n+2,n+1}\}$$
$$R_i = \{M_{i,i},M_{n+2,i}\}, \;\; for \;\; 3 \leq i \leq n$$
and $R_{n+1}$ covers all ones in the last column of $M$ and $R_{n+2}$ covers all remaining ones in columns $n+1$ and $n+2$ of $M $ not covered by the previous rectangles.

Now we show that by removing any column of $M$ we get a sub-matrix $M'$ with binary rank at most $n+1$.
First, remove the last column of $M$.
The ones in the resulting sub-matrix can be partitioned into $n+1$ rectangles, since the last column of $M$ was covered by a separate rectangle $R_{n+1}$ in
the alternative partition of $M$. Therefore, the first $n+2$ columns of $M'$ can be partitioned into $n+1$ rectangles as in $M$.

Next remove the $n$'th column of $M$.
The ones in the resulting sub-matrix $M''$ can partitioned into $n+1$ rectangles by taking the alternative partition of $M$ into $n+2$ rectangles,
and just removing the rectangle covering the ones in column $n$ of $M$. The same argument holds if we remove any column $3 \leq i \leq n$ of $M$.

Finally, note that, up to a permutation of rows and columns, removing columns $1$ or $2$ of $M$ results in $M'$, and removing columns $n+1$ or $n+2$
of $M$ results in $M''$.
\end{proof}

It is interesting to note that the example given by Hrube\v{s}~\cite{hrubevs2025hard} is minimal.

\begin{claim}
There is no $0,1$ matrix $M$ with binary rank $4$, such that each sub-matrix of $M$ has a strictly smaller binary rank.
\end{claim}
\begin{proof}
We first claim that if $\Rbin(M) = 4$ then $\Rreal(M) = 3$.
Otherwise, if $\Rreal(M) = 4$ then $M$ contains a sub-matrix of size $4 \times 4$ with real rank $4$, and this sub-matrix has, of course, binary rank $4$.
If $\Rreal(M) \leq 2$ then $\Rbin(M) = \Rreal(M)$, since the binary rank is equal to the real rank for real rank $1,2$. Hence, $\Rreal(M) = 3$.

But if $\Rbin(M) = 4$ and $\Rreal(M) = 3$, then  using a result of Parnas and Shraibman~\cite{parnas2025study}, the matrix $M$ must contain $C_4$ as a sub-matrix,
and, thus, $M$ has a sub-matrix of size $4 \times 4$ with binary rank $4$, since $\Rbin(C_4) = 4$.
\end{proof}

We note that the non-negative rank also does not have the full-rank sub-matrix property as the following theorem of Moitra~\cite{moitra2016almost} shows.
Hrube{\v{s} asked if a larger gap can be achieved between the binary rank of a matrix and the binary rank of any of its proper sub-matrices,
as is the case for the Boolean and non-negative rank.

\begin{theorem}[\cite{moitra2016almost}]
\label{theo-moitra}
For any $r$, there is a non-negative matrix $M$ of size $3rn \times 3rn$ with non-negative rank  at least $4r$,
such that every sub-matrix of $M$ with at most $n-1$ rows has non-negative rank at most $3r$.
\end{theorem}

\subsection{Kronecker products}
\label{subsec-kronecker}

The Kronecker product $A \otimes B$ of $0,1$-matrices $A$ and $B$, is defined as the block matrix whose $i,j$'th block is formed by multiplying  $A_{i,j}$
with the entire matrix $B$. That is, if $A$ is a matrix of size $n \times m$ and $B$ is of size $p \times q$, then $ A \otimes B$ is of size $np \times mq$ and is given by:
$$
A \otimes B =
\left(
  \begin{array}{ccc}
    A_{1,1}\cdot B & \cdots & A_{1,m} \cdot B \\
    \vdots & \ddots & \vdots \\
    A_{n,1}\cdot B & \cdots & A_{n,m} \cdot B \\
  \end{array}
\right)
$$
The Kronecker product has important applications in mathematics and computer science.
As we show in this section, it can be used to amplify a given gap between rank functions, and has strong connections to direct-sum questions in communication complexity.
See also Sch\"{a}cke~\cite{schacke2004kronecker} for other applications of Kronecker products.

It is well known that the real rank is multiplicative under the Kronecker product, that is,
$$\Rreal(A \otimes B) = \Rreal(A)\cdot \Rreal(B).$$
It is also easy to  verify that the Boolean and binary rank are sub-multiplicative with respect to this product, that is,
$\Rbool(A \otimes B) \leq  \Rbool(A) \cdot \Rbool(B)$ and similarly for the binary rank.
Recall that $i(A)$ denotes the size of a maximum isolation set in $A$. The following claim was observed by de Caen, Gregory and Pullman~\cite{de1988boolean}:
\begin{claim}
\label{amplify-gap}
Let $A$ and $B$ be two $0,1$ matrices. Then:
$$\Rbool(A \otimes B) \geq \max\{i(A) \cdot \Rbool(B),\Rbool(A) \cdot i(B)\},$$
and
$$
i(A)\cdot i(B) \leq i(A \otimes B) \leq \min \{i(A)\cdot \Rbool(B), \Rbool(A) \cdot i(B)\}.
$$
A similar result holds for the binary rank.
\end{claim}

This claim can be used in some cases to amplify the gap between two rank functions as described next.
Let $M^{\otimes k}$ be the $k$'th Kronecker power of $M$, which can be defined recursively as follows:
$$M^{\otimes 1} = M, \;\;\; M^{\otimes k} = M \otimes M^{\otimes (k-1)} \;\; \textrm{for} \;\; k \geq 2.$$
Recall that Theorem~\ref{theodoublegap} guarantees a family of matrices $M$ for which $\Rreal(M) =  \lfloor n/2\rfloor + 1$, and $i(M) = \Rbin(M) = 2\lfloor n/2\rfloor$.
Now, using Claim~\ref{amplify-gap}, we can amplify this gap as follows.
Dietzfelbinger,  Hromkovi{\v{c}} and Schnitger~\cite{dietzfelbinger1996comparison} used the Kronecker product to prove a similar result for  $n = 4$.

\begin{corollary}
There exists a family of matrices $M$ of size $n \times n$, for all $n \geq 4$, $n \neq 5$, such that for all $k \geq 1$,
$\Rreal(M^{\otimes k}) = (\lfloor n/2\rfloor + 1)^k$ and $\Rbin(M^{\otimes k}) = (2\lfloor n/2\rfloor)^{k}$.
\end{corollary}

In 1988,  de Caen, Gregory, and Pullman~\cite{de1988boolean} asked if, as for the real rank, the Boolean rank is multiplicative under the Kronecker product,
or if $\Rbool(A \otimes B) < \Rbool(A)\cdot \Rbool(B)$ for some matrices $A,B$.
They suggested the matrix $C_n$ as a candidate for inequality, and further proposed to estimate $\Rbool(C_n \otimes C_n)$ as a function of $n$.

The question of~\cite{de1988boolean} is closely related to {\em direct-sum} questions in communication complexity.
Direct-sum questions ask if solving two instances of a communication problem requires twice the communication needed to solve one instance of the problem.
Specifically, when considering the non-deterministic setting, Alice gets two row indices $i_1,i_2$ of a $0,1$ matrix $A$ and Bob gets two column indices $j_1,j_2$ of $A$.
Their goal is to decide if both instances $A_{i_1,j_1}$ and $A_{i_2,j_2}$ are $1$, or in other words, to determine the value of $A_{i_1,j_1} \cdot A_{i_2,j_2}$.
The non-deterministic communication complexity of this problem is exactly the non-deterministic communication complexity associated with the Kronecker product $A \otimes A$.
Hence, the question of~\cite{de1988boolean} is, in fact, equivalent to a direct sum question for the non-deterministic communication complexity, $N_1(C_n)$, of the matrix $C_n$,
since $\lceil\log \Rbool(C_n)\rceil = N_1(C_n)$ (see Section~\ref{subsec-communication}).
See also~\cite{FederKNN95,karchmer1995fractional,KN97} for more information on direct-sum questions in communication complexity.

Claim~\ref{amplify-gap} provides a simple lower bound of $\Rbool(C_n \otimes C_n) \geq 3 \cdot \Rbool(C_n)$, since $i(C_n)=3$ for every $n \geq 3$.
Recall also that $\Rbool(C_n) = \Theta(\log n)$.
As to an upper bound on $\Rbool(C_n \otimes C_n)$, Watts~\cite{watts2001boolean} proved in 2001 that $\Rbool(C_4 \otimes C_4) < \Rbool(C_4) \times \Rbool(C_4)$,
thus, answering the question of~~\cite{de1988boolean}. She also showed that this example is minimal in terms of the size of the matrix.

It should be noted that the question of~\cite{de1988boolean} was in fact already answered in 1995 by Karchmer, Kushilevitz, and Nisan~\cite{karchmer1995fractional},
who studied direct-sum questions in communication complexity. Their analysis shows that
the non-deterministic communication complexity of two instances of the non-equality problem associated with the matrix $C_n$,
is larger only by some additive constant than the non-deterministic complexity of a single instance.
This additive constant translates into a multiplicative constant in the Boolean rank of $C_n \otimes C_n$ compared to the Boolean rank of $C_n$.
Thus, their result implies that $\Rbool(C_n \otimes C_n) < \Rbool(C_n)^2$, for a sufficiently large $n$,
and also shows that $\Rbool(C_n \otimes C_n) = O(\log n)$. See~\cite{KN97}, Chapter 4, for a nice description of the argument of~\cite{karchmer1995fractional}.

Using the probabilistic method described in Section~\ref{Sec-prob},
Jukna~\cite{Jukna2022} proved in 2022 an upper bound on $\Rbool(A \otimes B)$, for matrices $A$ and $B$ that have a bounded number of zeros in each column (row).

\begin{theorem}[\cite{Jukna2022}]
Let $A$ and $B$ be two $0,1$ matrices of size $n \times n$, such that $A$ has at most $c_1$ zeros in each column and $B$ has at most  $c_2$ zeros in each column.
Then  $\Rbool(A \otimes B) = O(c_1 \cdot c_2 \cdot \log n)$.
\end{theorem}

Thus, we get yet another proof that $\Rbool(C_n \otimes C_n) < \Rbool(C_n)^2$ for a sufficiently large $n$,
since $C_n$ has one zero in each column, and therefore, $\Rbool(C_n \otimes C_n)  = O(\log n)$.
As we show next, a similar analysis to that of~\cite{Jukna2022} can be used to give an upper bound on  $\Rbool(M^{\otimes k})$
for any matrix $M$ which has at most $d$ zeros in each column (row).

\begin{figure}[htb!]
\centering
$$
M =
\left(
\begin{array}{cc}
1 & 1 \\
0 & 1 \\
\end{array}
\right)
\;\;\;\;\;\;\;\;
M^{\otimes 3} =
\left(
\begin{array}{@{}c|c@{}}
\begin{array}{cc|cc}
    1 & 1 & 1 & 1 \\
    0 & 1 & 0 & 1  \\\hline
    0 & 0 & 1 & 1  \\
    0 & 0 & 0 & 1  \\
\end{array}
&
\begin{array}{cc|cc}
     1 & 1 &  1 & 1\\
     0 & 1 &  0 & 1 \\\hline
      0 & 0 &  1 & 1 \\
     0 & 0 &  0 & 1 \\
\end{array}\\\hline
\begin{array}{cc|cc}
    0 & 0 & 0 & 0  \\
    0 & 0 & 0 & 0  \\\hline
    0 & 0 & 0 & 0  \\
    0 & 0 & 0 & 0  \\
\end{array}
&
\begin{array}{cc|cc}
  1 & \textcolor{red}{\bf 1} &  1 & 1 \\
  0 & 1 &  0 & 1 \\\hline
  0 & 0 &  1 & 1 \\
 0 & 0 &  0 & 1 \\
    \end{array}
\end{array}
\right)
$$
\vspace{2em}
\begingroup
   \renewcommand{\arraystretch}{1.6}
$$
M^{\otimes 3} = = \left(\begin{array}{c|c}
  B_{1,1}  &   B_{1,2}\\
\hline
 B_{2,1}  &  {\bf B_{2,2}}
\end{array}\right)
\;\;\;\;\;\;\;\;\;
B_{2,2} = \left(
\begin{array}{c|c}
  {\bf  B'_{1,1}} & B'_{1,2} \\\hline
  B'_{2,1} & B'_{2,2}
  \end{array}\right)
\;\;\;\;\;\;\;\;\;
B'_{1,1} = \left(
\begin{array}{c|c}
   1 & \textcolor{red}{\bf 1} \\\hline
  0 & 1
  \end{array}\right)
$$
\endgroup
\caption{The Kronecker product $M^{\otimes 3}$ of the matrix $M$.
The red one in entry $(5,6)$, that is in row $5$ and column $6$ of $M^{\otimes 3}$,
can be uniquely represented by two vectors $r = (2,1,1)$ and $c = (2,1,2)$, which represent the blocks to which row $5$ and column $6$ belong to, respectively.
For example, row $5$ belongs to block $B_{2,2}$ of $M^{\otimes 3}$, and in this
block to block $B'_{1,1}$ and finally it is in the first row of block $B'_{1,1}$.
Similarly the vector $c = (2,1,2)$ represents the blocks to which column $6$ belongs.}
\label{fig-kronekerblocks}
\end{figure}

\begin{theorem}
\label{theo-block}
Let $M$ be a $0,1$ matrix of size $n \times n$ with at most $d$ zeros in each column. Then for any $k \geq 1$,
$\Rbool(M^{\otimes k}) \leq c(d,k) \cdot \log n$, where $c(d,k)  = 3k\cdot e^k \cdot (d+1)^k \cdot \log_2 e$.
\end{theorem}
\begin{proof}
If $d = 0$ the theorem clearly holds as the Boolean rank of $M^{\otimes k}$ is  $1$ in this case, and if $k = 1$  the theorem follows from Theorem~\ref{theo:Alon}.
Thus, assume that $d \geq 1$ and $k \geq 2$.
Recall that $M^{\otimes k}$ is a matrix of blocks, where each block contains smaller blocks of a certain structure and so on.
Specifically, since $M^{\otimes k} = M \otimes M^{\otimes k-1}$, then at the topmost level $M^{\otimes k}$ is composed of blocks of $M^{\otimes k-1}$ or all zero blocks of the corresponding size.
Inside each non-zero block of the form $M^{\otimes k-1} = M \otimes M^{\otimes k-2}$ there are blocks of the form $M^{\otimes k-2}$ or smaller all-zero blocks, and so on,
until the innermost blocks are either all-zero or equal to $M$, and inside such a block each element can be described as a block of size $1 \times 1$ which is either a one or a zero.

Therefore, any entry $(i,j)$ of $M^{\otimes k}$ can be represented by two vectors $r = (r_1,r_2,...,r_k)$ and $c = (c_1,c_2,...,c_k)$,
which describe the sequence of blocks to which $(i,j)$ belongs, where each block is described by two indices $r_i,c_i$,
such that $r_i$ and $c_i$ are indices of a row and a column of blocks in that level, respectively.
Thus, $(i,j)$ belongs to block $r_1,c_1$ at the topmost level, and then to block  $r_2,c_2$ inside block $r_1,c_1$ and so on,
until the pair $r_k,c_k$ represents a row and a column index inside the innermost block of $M^{\otimes k}$. See Figure~\ref{fig-kronekerblocks} for an illustration.

Let $R$ be a subset of row indices of $M^{\otimes k}$, where each row $i \in R$ is chosen independently and randomly as follows,
according to the sequence of blocks it is contained in. If row $i$ is represented by the row sequence $(r_1,r_2,...,r_k)$,
then  $r_1$ is chosen with probability $1/(d+1)$, and then  $r_2$ is chosen with probability $1/(d+1)$ and so on.
Thus, the row sequence $(r_1,r_2,...,r_k)$ which uniquely represents row $i$, is chosen with probability $1/(d+1)^k$.

Let $C$ be a subset of column indices $j$ of $M^{\otimes k}$, such that for each $i \in R$ it holds that $M^{\otimes k}_{i,j} = 1$.
Let $S  = R \times C$ be the resulting set of pairs $(i,j)$.
The set $S$ is a monochromatic rectangle of ones in $M^{\otimes k}$, since all column indices in $C$ include only ones in the rows in $R$.

What is the probability that a pair $(i,j)$ for which $M^{\otimes k}_{i,j} = 1$ is included in the rectangle defined by $S$?
The probability that row $i$ is selected initially into $R$  is $1/(d+1)^k$.
The pair $(i,j)$ will be in $S$ if for all $i'\in R$ it holds $M^{\otimes k}_{i',j} = 1$, and then $j$ will be included in $C$ (that is, all rows in $R$ have a $1$ in column $j$).
Thus, $(i,j)$ will be in $S$ if we didn't select any row whose row sequence includes an all zero block which includes column $j$,
and since there are at most $d$ zeros in each column of $M$,  this happens with probability at least:
$$
\frac{1}{(d+1)^k}\left( 1 - \frac{1}{(d+1)}\right)^{k\cdot d} \geq \frac{1}{e^k (d+1)^k}
$$
Now if we choose $t =  c(d,k)\cdot \log_2 n  $ rectangles $S_1,...,S_{t}$ independently at random as described above, where $c(d,k)$ is as stated in the theorem,
the probability that $(i,j)$ is not covered by one of these rectangles is at most:
$$
\left(1- \frac{1}{e^k(d+1)^k} \right)^t \leq e^{-t/(e^k(d+1)^k)} = e^{-3k\log_2 n/\log_2 e } = n^{-3k}.
$$
Since there are at most $n^{2k}$ ones in $M^{\otimes k}$, the probability that one of them is not covered by these $t$ rectangles is at most:
$$n^{2k} \cdot n^{-3k} < 1$$
Hence, there exists at least one choice of rectangles $S_1,...,S_t$ for which all ones of $M^{\otimes k}$ are covered,
and thus, $\Rbool(M) \leq t = c(d,k) \cdot \log_2 n $ as required.
\end{proof}

Since $C_n$ has exactly one zero in each column we get the following corollary.

\begin{corollary}
There exists a constant $c> 1$ such that $\Rbool(C_n^{\otimes k}) = O(c^k \cdot \Rbool(C_n))$.
\end{corollary}

Another approach for upper bounding the Boolean rank of Kronecker products was developed by Haviv and Parnas~\cite{haviv2022upper} who
generalized the method of Watts~\cite{watts2001boolean} mentioned above, which was used to give an upper bound on $\Rbool(C_4 \otimes C_4)$.
Theorem~\ref{theorem:RankTensorIntro} presented next, which was proved in~\cite{haviv2022upper}, reduces the challenge of proving upper bounds on  $\Rbool(A \otimes B)$ to
finding two sequences $\mathcal{M}$ and $\mathcal{N}$ of low-rank sub-matrices for each of the matrices $A$ and $B$,
such that every $1$-entry of $A$ and $B$ is covered by many of them,
where a sequence of $0,1$-matrices is said to {\em cover} a matrix if their sum is equal to that matrix under the Boolean operations.
See Figure~\ref{fig:Watts} for an illustration.

\begin{figure}[htb!]
\begin{center}
\begin{tabular}{ccc}
$ \mathcal{M} \; = \;\;$
$
\left(
  \begin{array}{cccc}
    0 & 0 & 1 & 1 \\
    0 & 0 & 1 & 1 \\
    1 & 1 & 0 & 0 \\
    1 & 1 & 0 & 0 \\
  \end{array}
\right)$
&
$
\left(
  \begin{array}{cccc}
    0 & 1 & 0 & 1 \\
    1 & 0 & 1 & 0 \\
    0 & 1 & 0 & 1 \\
    1 & 0 & 1 & 0 \\
  \end{array}
\right)
$
&
$
\left(
  \begin{array}{cccc}
    0 & 1 & 1 & 0 \\
    1 & 0 & 0 & 1 \\
    1 & 0 & 0 & 1 \\
    0 & 1 & 1 & 0 \\
  \end{array}
\right)$
\end{tabular}
\end{center}
\caption{A sequence of matrices $\mathcal{M}$ used in~\cite{watts2001boolean}, as generalized in~\cite{haviv2022upper}.
Here $\mathcal{N} = \mathcal{M}$. The Boolean rank of each of these three matrices is $2$,
their Boolean sum is equal to $C_4$, and every one in $C_4$ is covered by two of them.
Thus, by Theorem~\ref{theorem:RankTensorIntro}, $\Rbool(C_4 \otimes C_4) \leq \sum_{t=1}^{3} \Rbool(M_{t}) \cdot \Rbool(N_{t}) = 3\cdot 2 \cdot 2 = 12 < 16 = \Rbool(C_4)\cdot \Rbool(C_4)$. }
\label{fig:Watts}	
\end{figure}

\begin{theorem}[\cite{haviv2022upper}]
\label{theorem:RankTensorIntro}
Let $A$ and $B$ be two $0,1$ matrices.
Suppose that there exist two sequences of $0,1$ matrices, $\mathcal{M} = M_{1},...,M_{s}$ and $\mathcal{N} = N_{1},...,N_{s}$, where repetitions of matrices are allowed,
such that:
\begin{enumerate}
  \item $\mathcal{M}$ is a cover of $A$, and
  \item for every $i,j$ such that $A_{i,j}=1$, the matrices $N_t$ with ${(M_t)}_{i,j} = 1$ form a cover of $B$.
\end{enumerate}
Then, $\Rbool(A \otimes B) \leq \sum_{t=1}^{s} \Rbool(M_{t}) \cdot \Rbool(N_{t}).$
\end{theorem}
\begin{proof}
We first show that
$$
A \otimes B =  \sum_{t = 1}^s (M_{t} \otimes N_{t}),
$$
where the operations are the Boolean operations. This is done by comparing the matrices on both sides of this equality, block by clock.

If $A_{i,j}=0$ then block $i,j$ of $A \otimes B$ is the all-zero matrix.
Since $\mathcal{M}$ is a cover of $A$, then $(M_t)_{i,j}=0$ for every $1 \leq t \leq s$.
Hence, block $i,j$ of $M_t \otimes N_t$ is also the all-zero matrix for all $1 \leq t \leq s$. Thus, the corresponding block of their sum is all-zero as well.

Now if $A_{i,j}=1$ then block $i,j$ of $A \otimes B$ is $B$.
As to block $i,j$ of the matrix $\sum_{t = 1}^s (M_{t} \otimes N_{t})$,
it is precisely the sum of the matrices $N_t$ for which $(M_t)_{i,j}=1$. By the second assumption of the theorem, these matrices form a cover of $B$. Hence, this block is $B$ as well.

Finally, using the sub-additivity and sub-multiplicativity of the Boolean rank:
\[\Rbool(A \otimes B) = \Rbool \bigg ( \sum_{t = 1}^s (M_{t} \otimes N_{t}) \bigg ) \leq \sum_{t=1}^{s}{\Rbool(M_{t} \otimes N_{t})}  \leq \sum_{t=1}^{s} \Rbool(M_{t}) \cdot \Rbool(N_{t}),\]
as required.
\end{proof}

Using the method introduced in Theorem~\ref{theorem:RankTensorIntro}, Haviv and Parnas~\cite{haviv2022upper} prove the
following theorem which settles the question of~\cite{de1988boolean} regarding the Boolean rank of $C_n \otimes C_n$, for all  $n \neq 5, 6$.

\begin{theorem}[\cite{haviv2022upper}]
Let $n,m \geq 7$. Then $\Rbool(C_n \otimes C_m) < \Rbool(C_n) \cdot \Rbool(C_m)$ .
\end{theorem}

As can be seen by the above discussion, there are several methods for proving upper bounds on the Boolean rank of Kronecker products.
A  long standing open problem is the intriguing question  concerning the multiplicativity of the binary rank under Kronecker products.
This question was also considered by Watson in~\cite{watson2016nonnegative} who gave an example the following matrix for which the binary and non-negative rank are not equal:
$$M = \left(
  \begin{array}{cccccc}
    1 & 1 & 1 & 0 & 0 & 0 \\
    1 & 0 & 0 & 1 & 1 & 0 \\
    0 & 1 & 0 & 1 & 0 & 1 \\
    0 & 0 & 1 & 0 & 1 & 1 \\
    1 & 1 & 1 & 1 & 1 & 1 \\
  \end{array}
\right)
$$
Watson showed that $\Rbin(M) = 5$ and $\Rnon(M) = 4$.
While trying to increase the gap between the binary and the non-negative rank, Watson suggested that for this matrix $M$ it holds that
$\Rbin(M^{\otimes k}) = \Rbin(M)^k$.
To the best of our knowledge, no progress was made so far regarding this problem and it is still not known if
there exist matrices $A,B$ for which $\Rbin(A \otimes B) < \Rbin(A)\Rbin(B)$.

\subsection{The rank of a matrix and its complement}
\label{Subsec-complemement}

Given a $0,1$ matrix $M$, its {\em complement} $\overline{M}$ is the matrix achieved by replacing each $0$ of $M$ with a $1$ and each $1$ with a $0$.
The problem of determining the relationship between the rank of $M$ and the rank of $\overline{M}$, for a given rank function, is another example of a question which was
raised by both the mathematical community and the computer science community, mostly in the context of communication complexity.
Therefore, in this sub-section we consider the known bounds on the rank of a matrix and its complement, for the real, Boolean and binary rank functions,
where some of the methods used are combinatorial or algebraic and others employ the powerful lifting technique described in Section~\ref{subsec-lifting}.

It is well known that there is a gap of at most one between $\Rreal(\overline{M})$ and $\Rreal(M)$, as the following simple claim proves.
As we show in this section, such a statement does not hold in general for the Boolean and binary rank.

\begin{claim}
Let $M$ be a $0,1$ matrix. Then $|\Rreal(M) - \Rreal(\overline{M})| \leq 1$.
\end{claim}
\begin{proof}
The claim follows directly by using the sub-additivity of the real rank and since $\overline{M} = M - J$, where $J$ is the all one matrix of rank $1$.
\end{proof}

As described in Section~\ref{subsec-communication}, the non-deterministic communication complexity, $N_1(M)$,  of $M$ is equal to $\lceil\log \Rbool(M)\rceil$, and thus,
also equal to $\log$ of the minimal number of monochromatic rectangles required to cover all ones of $M$.
The co-non-deterministic communication complexity, $N_0(M)$, of $M$ is equal to the non-deterministic communication complexity of $\overline{M}$,
and therefore, also equal to $\log$ of the minimal number of monochromatic rectangles  required to cover the zeros of $M$.
Thus, the question of determining the relationship between $\Rbool(M)$ and  $\Rbool(\overline{M})$ is equivalent to determining the relationship between $N_1(M)$ and $N_0(M)$.

Indeed, there can be an exponential gap between $\Rbool(\overline{M})$ and $\Rbool(M)$.
For example, we saw such a gap for the matrix $C_n$, since the complement of $C_n$ is the identity matrix $I_n$, and
$\Rbool(I_n) = n$, whereas $\Rbool(C_n)  = \Theta(\log n)$.
The matrix $C_n$ corresponds to the non-equality function considered in communication complexity,
and therefore, there is an exponential gap between the non-deterministic and co-nondeterministic communication complexity of non-equality.

What about the binary rank of a matrix and its complement?
A result of Yannakakis~\cite{Yannakakis91}  proves that the deterministic communication complexity, $D(M)$, of any $0,1$ matrix $M$,
is at most $O(\log^2 \Rbin(M))$ (see Theorem~\ref{theo-rec-reduction}, Section~\ref{subsec-protocols}).
But the deterministic communication complexity of a matrix $M$ is equal to the deterministic communication complexity of its complement $\overline{M}$,
since the two players can simply flip the output of the deterministic protocol for $M$. Therefore:
$$\log \Rbin(\overline{M}) \leq D(\overline{M}) \leq O(\log^2 \Rbin(M)).$$
This implies that if $\Rbin(M) = d$ then
$$
\Rbool( \overline{M}) \leq \Rbin( \overline{M}) \leq d^{O(\log d)}.
$$
Since $\lceil\log \Rbool(\overline{M})\rceil = N_0(M)$, then any lower bound on the co-non-deterministic communication complexity of some matrix $M$ or on the
Boolean rank of $\overline{M}$, will provide  a lower bound for $\Rbin(\overline{M})$.

For many years it was an open problem to find a tight lower bound on the co-non-deterministic communication complexity of the
clique vs. independent set problem discussed in Section~\ref{subsec-protocols}.
Furthermore, it was shown that this problem is also related to the Alon, Saks, and Seymour conjecture mentioned in Section~\ref{subsec-linear}.
This conjecture suggested that $\chi(G) \leq bp(G)+1$, where $bp(G)$  is the minimal number of bicliques  needed to partition the edges of a graph $G$ and $\chi(G)$ is its chromatic number.
See  Huang and Sudakov~\cite{HuangS12} and Bousquet, Lagoutte, and Thomass{\'{e}}~\cite{BousquetLT14} for a description of the connections and equivalence of these problems.

We now describe the progress made regarding these problems in terms of a lower bound on $\Rbool(\overline{M})$ for a matrix with $\Rbin(M) = d$.
However,  due to the above discussion, the results mentioned imply similar lower bounds on
the con-nondeterministic communication complexity of the clique vs. independent set problem,
as well as provide separations between $bp(G)$ and $\chi(G)$.

The first non-trivial lower bound was given by Huang and Sudakov~\cite{HuangS12} who disproved the Alon, Saks, and Seymour conjecture.
Building on a construction of Razborov~\cite{Razborov92}, the result of~\cite{HuangS12} provides a
family of graphs $G$,  with $bp(G)=d$ and $\chi(G) \geq \Omega(d^{6/5})$, for infinitely many integers $d$.
See~\cite{CioabaT11} for extended constructions.
The result of Huang and Sudakov can be translated into a family of  matrices $M$ for which $\Rbool( \overline{M}) \geq \Omega(d^{6/5})$.

The constant $6/5$ in the exponent was improved over the years, as can be seen in Table~\ref{tab:lowerbounds}, until finally,
 Balodis, Ben-David, G{\"{o}}{\"{o}}s, Jain, and Kothari~\cite{BBGJK21} proved the following theorem which gives a lower bound almost matching
 the upper bound given by Yannakakis.

 \begin{theorem}[\cite{BBGJK21}]
 \label{theo-complement}
 There exists  a family of matrices $M$ with $\Rbin(M) = d$ and $\Rbool( \overline{M}) \geq d^{\widetilde{\Omega}(\log d)}$.
 \end{theorem}

  \begin{table}[ht]
    \centering
    \begin{tabular}{|Sc|Sc|Sc|}
    \hline
        $\Rbool(\overline{M}) = $ & Authors & Year \\ [6pt]
        \hline \hline
        $\Omega(d^{6/5})$ & Huang and Sudakov~\cite{HuangS12} & 2012  \\ [6pt] \hline
        $\Omega(d^{3/2})$ & Amano~\cite{Amano14} & 2014  \\ [6pt] \hline
        $\Omega(d^{2})$ & Shigeta and Amano~\cite{ShigetaA15} & 2015 \\ [6pt] \hline
        $d^{\Omega(\log^{0.128} d)}$ & G{\"{o}}{\"{o}}s~\cite{Goos15} & 2015 \\ [6pt] \hline
        $d^{\widetilde{\Omega}(\log d)}$ & Balodis, Ben-David, G{\"{o}}{\"{o}}s, Jain, Kothari~\cite{BBGJK21} & 2021 \\ [6pt] \hline
    \end{tabular}
    \caption{Constructions of families $M$ for which $\Rbin(M) = d$, and $\Rbool(\overline{M})$ is bounded below as stated in the table.}
    \label{tab:lowerbounds}
\end{table}

It is interesting to also consider the gap between the rank of a matrix and the rank of its complement for specific families of matrices.
One such family is the family  of  {\em regular} $0,1$ matrices.

\begin{definition}
A $0,1$ matrix $M$ is called $k$-regular if all of its rows and columns have exactly $k$ ones.
\end{definition}

Brualdi, Manber, and Ross~\cite{BrualdiMR86} proved in 1986 the following for the real rank of regular matrices and their complements:

\begin{theorem}[\cite{BrualdiMR86}]
\label{theo-real-acomplement}
For every $k$-regular $0,1$ matrix $M$ of size $n \times n$, where $0 < k < n$,
it holds that $\Rreal(M) = \Rreal(\overline{M})$.
\end{theorem}

Following their work, Pullman~\cite{Pullman88} asked in 1988 whether every regular matrix $M$ satisfies $\Rbin(M) = \Rbin(\overline{M})$,
and Hefner,  Henson,  Lundgren, and  Maybee  conjectured in~\cite{HefnerHLM90} that the claim is false for the binary rank.
Again, the regular matrix $C_n$ is an example for which this claim is false for the Boolean rank.

Inspired by the result of~\cite{BBGJK21}, Haviv an Parnas~\cite{haviv2023binary2}  proved the following regular analogue, thus, showing
that the conjecture of~\cite{HefnerHLM90} is correct. Their proof uses a query-to-communication lifting theorem in non-deterministic communication complexity,
and as achieved by~\cite{BBGJK21} this result  almost matches the upper bound of Yannakakis~\cite{Yannakakis91} stated above. They further
provide regular counterexamples to the Alon-Saks-Seymour conjecture, and show that
for infinitely many integers $d$, there exists a regular graph $G$ with $bp(G) = d$ and $\chi(G) = d^{\widetilde{\Omega}(\log d)}$.

\begin{theorem}[\cite{haviv2023binary2}]
\label{theo-regular-gap}
For infinitely many integers $d$, there exists a regular matrix $M$ with  $\Rbin(M) = d$ and  $\Rbool(\overline{M}) = d^{\widetilde{\Omega}(\log d)}$.
\end{theorem}

As an immediate corollary we get a family of matrices with a quasi-polynomial gap between the real and the binary and Boolean rank.

\begin{corollary}
\label{coro-real-binary-gap}
For infinitely many integers $d$, there exists a regular matrix $A$ for which $\Rreal(A) \leq d$ and $\Rbool(A), \Rbin(A) = d^{\widetilde{\Omega}(\log d)}$.
\end{corollary}
\begin{proof}
Let $M$ be the regular matrix guaranteed by Theorem~\ref{theo-regular-gap}. Therefore,  $\Rbin(M) = d$ and $\Rbool(\overline{M}) = d^{\widetilde{\Omega}(\log d)}$.
By Theorem~\ref{theo-real-acomplement}, $\Rreal(\overline{M}) = \Rreal(M) \leq \Rbin(M) = d$.
Set $A = \overline{M}$. The result follows since $\Rreal(A) \leq d$, whereas, $\Rbool(A), \Rbin(A) = d^{\widetilde{\Omega}(\log d)}$.
\end{proof}

The results of~\cite{Goos15, BBGJK21, haviv2023binary2} all use the query-to-communication lifting technique,
and show again the fruitful connections between communication complexity and the study of the binary and Boolean rank,
which required also to develop new lifting theorems in communication complexity.
Although this technique achieves the maximal gap possible in both the results of~\cite{ BBGJK21, haviv2023binary2},
it is worthwhile to find simpler combinatorial constructions which demonstrate a gap between the binary rank of a matrix and its complement,
as was done by~\cite{HuangS12,Amano14, ShigetaA15}.

\begin{figure}[htb!]
\begin{center}
$
M = \left(
  \begin{array}{cccc|cccc}
    0 & 0 & 1 & 1 & 0 & 0 & 0 & 0 \\
    1 & 0 & 0 & 1 & 0 & 0 & 0 & 0 \\
    1 & 1 & 0 & 0 & 0 & 0 & 0 & 0 \\
    0 & 1 & 1 & 0 & 0 & 0 & 0 & 0 \\\hline
    0 & 0 & 0 & 0 & 0 & 0 & 1 & 1 \\
    0 & 0 & 0 & 0 & 1 & 0 & 0 & 1 \\
    0 & 0 & 0 & 0 & 1 & 1 & 0 & 0 \\
    0 & 0 & 0 & 0 & 0 & 1 & 1 & 0 \\
\end{array}
\right)
\;\;\;\;\;\;\;\;\;\;
\overline{M} = \left(
  \begin{array}{cccc|cccc}
  \textcolor{blue}{\bf 1} & \textcolor{blue}{\bf 1} & 0 & 0 & \textcolor{blue}{\bf 1} & \textcolor{blue}{\bf 1} & \textcolor{orange}{\bf 1} & \textcolor{orange}{\bf 1} \\
    0 & \textcolor{red}{\bf 1} & \textcolor{red}{\bf 1} & 0 & \textcolor{brown}{\bf 1} & \textcolor{darkgray}{\bf 1} & \textcolor{darkgray}{\bf 1} & \textcolor{brown}{\bf 1} \\
    0 & 0 & \textcolor{green}{\bf 1} & \textcolor{green}{\bf 1} & \textcolor{brown}{\bf 1} & \textcolor{darkgray}{\bf 1} & \textcolor{darkgray}{\bf 1} & \textcolor{brown}{\bf 1} \\
    \textcolor{gray}{\bf 1} & 0 & 0 & \textcolor{gray}{\bf 1} & \textcolor{brown}{\bf 1} & \textcolor{darkgray}{\bf 1} & \textcolor{darkgray}{\bf 1} & \textcolor{brown}{\bf 1} \\\hline
    \textcolor{blue}{\bf 1} & \textcolor{blue}{\bf 1} & \textcolor{green}{\bf 1} & \textcolor{green}{\bf 1} & \textcolor{blue}{\bf 1} & \textcolor{blue}{\bf 1} & 0 & 0 \\
    \textcolor{gray}{\bf 1} & \textcolor{red}{\bf 1} & \textcolor{red}{\bf 1} & \textcolor{gray}{\bf 1} & 0 & \textcolor{darkgray}{\bf 1} & \textcolor{darkgray}{\bf 1} & 0 \\
    \textcolor{gray}{\bf 1} & \textcolor{red}{\bf 1} & \textcolor{red}{\bf 1} & \textcolor{gray}{\bf 1} & 0 & 0 & \textcolor{orange}{\bf 1} & \textcolor{orange}{\bf 1} \\
    \textcolor{gray}{\bf 1} & \textcolor{red}{\bf 1} & \textcolor{red}{\bf 1} & \textcolor{gray}{\bf 1} & \textcolor{brown}{\bf 1} & 0 & 0 & \textcolor{brown}{\bf 1} \\
   \end{array}
\right)
$
\end{center}
\caption{An example of a $2$-regular matrix $M$ with binary rank $8$, such that the binary rank of its complement $\overline{M}$ is at most $7$.
It is easy to verify that $\Rbin(M) = 8$, by noticing that each block on the main diagonal is in fact $D_{4,2}$.
As to the binary rank of $\overline{M}$, a partition of the ones into $7$ monochromatic rectangles is presented, where each rectangle is represented by a different color.  }
\label{fig:d-regular}	
\end{figure}

Such a combinatorial construction for regular matrices, although with a modest gap between $\Rbin(M)$ and $\Rbin(\overline{M})$, is given by Haviv and Parnas in~\cite{haviv2023binary}, who considered
the special case of block diagonal $0,1$ matrices, where each block is a circulant regular sub-matrix.
They present a general method for proving upper bounds on the complement of such circulant block matrices, and also prove matching lower bounds for
various families of matrices.

\begin{theorem}[\cite{haviv2023binary}]
For every $k \geq 2$, there exist $k$-regular $0,1$ matrices $M$, such that $\Rbin(M) > \Rbin(\overline{M})$.
\end{theorem}
\begin{proof}
The full proof for general circulant block matrices can be found in~\cite{haviv2023binary2},
 but the basic idea is as follows. Consider the matrix $M$ which has two blocks of $D_{2k,k}$ on the main diagonal, and zeros elsewhere.
Then $\Rbin(M) = 2 \cdot \Rbin(D_{2k,k}) = 4k$. On the other hand, the matrix $\overline{M}$ has two blocks of $D_{2k,k}$ on the main diagonal,
up to a permutation of rows and columns, and ones elsewhere.
It can be verified that $\Rbin(\overline{M}) \leq 4k-1$. See Figure~\ref{fig:d-regular} for an illustration.
\end{proof}

\section{Algorithmic results}
\label{Sec-algorithms}

As computing the Boolean and binary rank is $NP$-hard, it is necessary to find relaxations for these problems.
This section presents algorithmic directions taken for computing or approximating the Boolean or binary rank,
such as parameterized algorithms, approximation algorithms, property testing and the approximate Boolean factorization problem.

\subsection{Parameterized algorithms}
\label{SubSec-parameter}

The field of parameterized complexity focuses on determining the complexity of a problem according to the specified parameters of the inputs to the problem.
In particular, this approach tries to find  algorithms which  solve the problem efficiently in some of the parameters, while having an exponential dependence on other parameters.
This approach is especially useful for $NP$-hard problems that have some fixed parameter which is a part of the problem.
For example, the input of the vertex coloring problem includes a graph $G$ with $n$ vertices and $m$ edges, and some parameter $k$,
and the goal is to decide if the chromatic number of the graph is at most $k$. Thus, the parameters are $n,m,k$.
In the context of this survey, the relevant parameters of the rank problem are the size $n \times m$ of the matrix $M$ and the rank parameter $d$,
and the goal is to decide if the rank of $M$ is at most $d$ for a given rank function.

A parameterized problem is called {\em fixed-parameter tractable} with respect to a parameter $k$ if there exists an algorithm that solves it
in time $n^c \cdot f(k)$ on an input of size $n$, where $c$ is an absolute constant and $f$ is a computable function.
This allows to solve the problem efficiently in polynomial time for a {\em fixed} $k$, even if $f(k)$ is exponential in $k$.
See Downey and Fellows~\cite{downey1995fixed} for a detailed description of parameterized complexity.

As described in Section~\ref{Sec-kernel}, {\em kernelization} is a central technique in parameterized algorithms,
where the goal is to extract a small {\em kernel} of the original problem whose size depends only on a chosen parameter $k$ of the problem,
so that solving the problem on this reduced instance yields a solution to the full problem.
The total complexity of the algorithm depends on the time required to find the small kernel, in addition to the time of solving the problem for the small kernel.
The steps required to find the kernel are called {\em kernelization rules}.

Fleischner,  Mujuni, Paulusma, and Szeider~\cite{fleischner2009covering} suggested the set of kernelization rules described in section~\ref{Sec-kernel},
which construct a kernel suitable for solving the minimum biclique edge partition or cover problem for a bipartite graph.
Nor et al.~\cite{nor2012mod} showed how to use this kernel to get fixed-parameter algorithms for these problems.

For simplicity we consider here the equivalent formulation of finding the binary or Boolean rank of a $0,1$ matrix.
Thus, the kernel of $M$ described by~\cite{fleischner2009covering}  is a sub-matrix of $M$ achieved by removing all-zero rows and columns and removing duplicate rows and columns.
The result of~\cite{nor2012mod} shows how to use this kernel
to get a parameterized algorithm for the binary and Boolean rank with $f(d) = O(2^{2^{2d} \cdot \log d})$.

Basically, the algorithm of~\cite{nor2012mod} goes over all possible partitions or covers of the ones in the kernel into monochromatic rectangles, depending on the rank in question.
By Lemma~\ref{lem:binary_size of matrix}, the resulting kernel has at most  $2^{d}$ rows and at most $2^d$ columns, for a $0,1$ matrix with binary or Boolean rank at most $d$.
Thus, the number of ones in the kernel is bounded by $2^d \cdot 2^d = 4^d$, and this is what determines the upper bound stated on $f(d)$.

As we describe next, it is possible to slightly improve the complexity of this simple algorithm for the binary rank by using Lemma~\ref{lem-product-rows-cols}
which gives an upper bound on the product of the number of rows and  columns of the kernel.

\begin{lemma}
Let $M$ be a $0,1$ matrix of size $n \times m$. It is possible to check if the ones of $M$ can be partitioned into $d$ monochromatic rectangles
in time complexity  $O(2^{2^{d} \cdot d\log d}) + poly(n,m)$.
\end{lemma}
\begin{proof}
Let $M'$ be the kernel of $M$ achieved by removing all zero rows and columns and removing duplicate rows and columns of $M$.
By sorting the rows of $M$ and then the columns and removing duplicate rows and columns, the kernel can be constructed in time $O((n^2 + m^2)\log n)$.

By Lemma~\ref{lem-product-rows-cols}, if $\Rbin(M) \leq d$ then the product of the number of rows and columns of $M'$ is at most $(d+1)2^d$,
and so the number of ones of $M'$ is at most $(d+1)2^d$. If this product is larger, we can already determine that $\Rbin(M) > d$.

Now, using a brute force algorithm, go over all partitions  $S_1,...,S_d$ of the ones of $M'$ into $d$ non-empty disjoint subsets.
For each partition, check if each one of the subsets $S_i$ is a monochromatic rectangle of ones.
If a partition is found for which all $d$ subsets are monochromatic rectangles of ones, then the algorithm can accept, and otherwise it rejects.

As for the complexity of this simple algorithm:
The number of possible partitions of at most $(d+1)2^d$ ones into $d$ non-empty disjoint sets is the Stirling number $S((d+1)2^d,d)$ of the second kind and it holds:
$$S((d+1)2^d,d) = O\left(\frac{d^{(d+1)2^d}}{d!}\right) = O\left(\frac{2^{2^{d} \cdot d\log d}}{d!}\right).$$
Checking for each $d$ subsets $S_1,...,S_d$, if they are a partition of the ones of $M$ into $d$ monochromatic rectangles, takes a total of $O((d+1)2^d)$ for each partition, and thus, the total complexity is:
$$O\left((d+1)\cdot 2^d \cdot \frac{2^{2^{d} \cdot d\log d}}{d!}\right) = O(2^{2^{d} \cdot d\log d})$$
as claimed.
\end{proof}

Nor et al.~\cite{nor2012mod} also describe an algorithm for the Boolean rank for which $f(d) =  O(2^{d2^{d-1} + 3d})$.
Chandran, Issac and  Karrenbauer~\cite{chandran2017parameterized} give an improved algorithm for the binary rank for which $f(d) = O^*(2^{2d^2 + d \log d + d})$.
Moreover,~\cite{chandran2017parameterized} prove that there is no algorithm for the Boolean rank with less than a double exponential complexity in $d$
unless the Exponential Time Hypothesis (ETH) is false, thus, matching up to polynomial factors the algorithm given by~\cite{nor2012mod}.

\subsection{Approximation Algorithms}
\label{SubSec-approx}

Approximation algorithms are a common and useful approach in computer science for finding  an approximate solution to optimization problems which are $NP$-hard.
The quality of an approximation algorithm is determined, among other things, by its {\em approximation ratio} $C$, which guarantees that the size of the solution found by the algorithm should by at most $C$ times larger or smaller than the size of the optimal solution, depending on the problem being a minimization or a maximization problem.

Both the binary and Boolean rank problems are minimization problems, and both are $NP$-hard.
Thus, an approximation algorithm for the binary rank of a matrix $M$ should output some approximation $Approx$ such that
$$
\Rbin(M) \leq Approx \leq C \cdot \Rbin(M)
$$
where $C \geq 1$ is the approximation ratio of the algorithm. A similar definition holds for an Approximation algorithm for the Boolean rank.

Unfortunately, we cannot hope for a small approximation ratio in general, since Chalermsook,  Heydrich,  Holm, and  Karrenbauer~\cite{chalermsook2014nearly} showed that it is hard to approximate the Boolean rank of a $0,1$ matrix $M$ of size $n \times n$  to within a
factor of $n^{1 - \epsilon}$ for any given $\epsilon > 0$.
However, we will show that in some cases a better approximation is possible.

Chandran, Issac and  Karrenbauer~\cite{chandran2017parameterized} give a simple polynomial approximation algorithm
for the minimum biclique edge cover or partition of  a bipartite graph.
For simplicity, we present their algorithm for the equivalent formulation of the Boolean and binary rank,
and show that this algorithm achieves an approximation ratio of $n/\log n$ for a matrix $M$ of size $n \times n$.

The algorithm first finds the kernel of $M$ as described in Section~\ref{Sec-kernel}, that is, removes all-zero and duplicate rows and columns of $M$.
Let $M'$ be the resulting kernel of $M$.
Now, the algorithm covers the ones in each row of the kernel by a separate rectangle
(this corresponds to partitioning the edges into stars originating from all vertices on the left side in the corresponding kernel of the bipartite graph represented by $M'$).

By Lemma~\ref{lem:binary_size of matrix} the number of distinct rows of $M$ is at most $2^{\Rbin(M)}$, and so for any $n \geq 3$,
the size of the resulting partition into monochromatic rectangles is at most:
$$
\min\{n,2^{\Rbin(M)}\} = \frac{\min\{n,2^{\Rbin(M)}\} }{\Rbin(M)} \cdot \Rbin(M) \leq \frac{n }{\log n} \cdot \Rbin(M).
$$
A similar bound holds for the Boolean rank, since any partition of the ones is also a cover of the ones of $M$.

The following algorithm of Haviv~\cite{Haviv25} from 2025 improves this approximation ratio for the binary rank, while
exploiting the connection between the binary rank and the deterministic communication complexity.
Recall that Sudakov and Tomon~\cite{sudakov2023matrix} proved that there exists a constant $c > 1$ such that $D(M) \leq c\sqrt{\Rreal(M)}$,
where $D(M)$ is the deterministic communication complexity of $M$, and it always holds that $\log_2 \Rbin(M) \leq D(M)$.

\bigskip
\begin{center}
\fbox{
\begin{minipage}{4.5in}
\leftline{\bf Algorithm Approximate binary rank ($M$ of size $n \times n$, constant $c$)}
\begin{enumerate}
\item
Compute $\Rreal(M)$.
\item
If $\Rreal(M) \leq log_2^2 n /c^4$, output $Approx = n^{1/c}$.
\item
Else output $Approx = n$.
\end{enumerate}
\end{minipage}
}
\end{center}
\bigskip

\begin{theorem}[\cite{Haviv25}]
Algorithm "Approximate binary rank" is an approximation algorithm for the binary rank with an approximation ratio of $O(n/\log_2^2 n)$, for any $0,1$ matrix $M$ of size $n \times n$.
\end{theorem}
\begin{proof}
Let $c$ be the constant guaranteed by the result of Sudakov and Tomon~\cite{sudakov2023matrix}.  Therefore, by the discussion above:
\begin{enumerate}
\item
If $\Rreal(M) \leq \log_2^2 n /c^4$:
$$\Rbin(M) \leq 2^{D(M)}\leq 2^{c\sqrt{\Rreal(M)}}\leq 2^{c\sqrt{ \log_2^2 n /c^4 }} = n^{1/c}$$
and thus, the approximation  ratio in this case is:
$$
Approx =  n^{1/c}  \leq O\left(\frac{ n}{\log_2^2 n}\right) \leq O\left(\frac{ n}{\log_2^2 n}\right) \cdot \Rbin(M)
 $$
\item
If $\Rreal(M) > \log_2^2 n /c^4 $: then $\Rbin(M) \geq \Rreal(M) > \log_2^2 n /c^4 $,
and therefore:
$$ Approx = n = \frac{c^4 \cdot n}{ \log_2^2 n}\cdot \frac{ \log_2^2 n}{ c^4} \leq  \frac{c^4 \cdot n}{ \log_2^2 n} \cdot \Rbin(M) = O\left(\frac{ n}{\log_2^2 n}\right) \Rbin(M)$$
\end{enumerate}
It is also easy to verify that in both cases $Approx \geq \Rbin(M)$.
Thus, in both cases we get an approximation ration of $O(n/\log_2^2 n)$ as claimed.
\end{proof}
The result does not go through for the Boolean rank since the Boolean rank can be smaller than the real rank.
However, a similar result holds for the non-negative rank as the non-negative rank is bounded between the real and the binary rank.
Note also that the approximation algorithm of~\cite{Haviv25} for the binary rank only finds an approximation for the binary rank with a guaranteed approximation ratio of $O(n/log^2 n)$,
and does not give an actual partition of the ones of the matrix into monochromatic rectangles,
whereas the algorithm of~\cite{chandran2017parameterized} gives an actual approximate cover or partition of the ones, although a very simple one.

Using the reduction described in Section~\ref{Sec-complexity} from the Boolean rank to the vertex coloring problem,
it is possible to use the known approximation algorithms for vertex coloring to get an approximation for the Boolean rank.
This connection was noticed by Chalermsook,  Heydrich1,  Holm and  Karrenbauer~\cite{chalermsook2014nearly} who stated the following result, using a result of
Halld{\'o}rsson~\cite{halldorsson1993still} who described a polynomial time approximation algorithm for the chromatic number with an approximation ratio of $O(n(\log\log n)^2/(\log n)^3)$.

\begin{corollary}
There is an approximation algorithm for the Boolean rank of a matrix $M$ with an approximation ratio of $O(s(\log\log s)^2/(\log s)^3)$, where $s$ is the number of ones in the matrix $M$.
\end{corollary}

This result is better for sparse matrices $M$ than the simple approximation ratio of $n /\log n$ described above.
Note also, that if each column (row) of $M$ has at most $d$ zeros, then by Theorem~\ref{theo:Alon},
the Boolean rank of $M$ is at most $c(d) \log n$ for some constant $c(d)$ which depends on $d$.
Thus, in this case it is possible to simply approximate the Boolean rank by $c(d) \log n$.
Another simple approximation algorithm for the Boolean rank is given in the next theorem proved by Miettinen~\cite{miettinen2010sparse}, using a reduction to the set cover problem.
Here, in order to get polynomial time complexity,  the requirement is that  each column of the matrix has  at most $\log n$ ones.
We give a slightly different proof than that of~\cite{miettinen2010sparse} by using the alternative formulation of covering the ones of the matrix by monochromatic rectangles.

\begin{theorem}[\cite{miettinen2010sparse}]
Let $M$ be a $0,1$ matrix of size $n \times m$ with at most $\log n$ ones in each column.
It is possible to approximate the Boolean rank of $M$ up to a ratio of $O(\log m + \log \log n)$ in polynomial time in $n,m$.
\end{theorem}
\begin{proof}
We reduce the problem of approximating the Boolean rank to a set covering problem as follows. The ones of $M$ are the elements of the universe $U$ that we want to cover.
By our assumption on the number of ones in each column, the total number of ones is at most $|U| \leq m \log n$. The sets used to cover the ones are monochromatic rectangles of ones in $M$.
The classic greedy approximation algorithm of Johnson~\cite{johnson1973approximation} for the set cover problem selects in each step a set (rectangle) which covers the maximal number of uncovered elements (ones) so far.
This greedy algorithm has an approximation ratio of $O(\log |U|) = O(\log (m \log n)) = O(\log m + \log \log n)$.

We next show that the number of rectangles considered by the algorithm in each iteration is at most $n \cdot m$, and thus, the running time is polynomial in $n,m$,
as checking which rectangle covers the most uncovered ones can be done also in polynomial time in $n,m$.
Indeed, for any given subset $R$ of rows of $M$, we can consider the maximal rectangle defined by the rows in $R$ which includes {\em all} columns that have only ones in these rows.
Moreover, it is enough to consider only subsets $R$ of rows, such that there exists a column $j$ which has only ones in the rows of $R$.
Otherwise, if there is no such column $j$ for a specific subset of rows $R$, then the monochromatic rectangle defined by these rows in empty.
Hence, the number of possible maximal rectangles is bounded by the number of possible ways to select such a subset $R$ of rows. But since there are at most $\log n$ ones in each column,
there are at most $2^{\log n}$ ways to select a subset of rows which has only ones in column $j$, and therefore, going over all columns, that are at most $m \cdot 2^{\log n} = m \cdot n$ ways to select subsets of rows.
Thus, the number of monochromatic rectangles the algorithm has to consider in each iteration is at most $n \cdot m$ as claimed.
\end{proof}

It is interesting to find better approximation algorithms for other specific families of matrices, possibly when an upper bound on the rank is known in advance,
as was done by Blum~\cite{blum1994new} for the chromatic number.
For example, a result of Parnas and Shraibman~\cite{parnas2025study}, shows that for a matrix with real rank  $d \leq 4$, the binary and Boolean rank are at most $2d - 2$.
This gives  a trivial  algorithm which can approximate the binary and Boolean rank of such matrices by simply computing the real rank.

\subsection{Property testing}
\label{SubSec-property}

{\em Property testing} is another approach for relaxing the problem of computing the Boolean and binary rank.
The field of property testing originated in a paper of Rubinfeld and Sudan~\cite{rubinfeld1996robust} from 1996, where they
introduced the notion of testing functions, and a paper by Goldreich, Goldwasser, and Ron~\cite{GGR98} from 1998, where they
initiated the study of testing graph  properties.
Since then the field has become an active and fruitful area of research in theoretical computer science,
and many property testing algorithms were devised. For example, properties of graphs such as $k$-colorability, connectivity or bipartiteness,
properties of functions such as linearity, monotonicity and dictatorship, regular languages, clustering, properties of distributions and more
(see a book by Goldreich on property testing and references within~\cite{goldreich2017introduction}).

The basic idea of property testing is to design ultra-fast algorithms that provide a relaxation to classical decision algorithms.
A {\em decision algorithm} should determine if a given object has some pre-specified property, such as a graph being $k$-colorable, a function being linear or a matrix having rank at most $d$.
Traditionally, one would like the decision algorithm to always give a correct answer, that is, accept an object which has the given property and reject it otherwise.
However, this usually requires the algorithm to inspect the entire object, and thus, its running time is at least linear in the size of the input object.
When the object is very large this poses a problem, and in particular, it is unlikely to have a polynomial decision algorithm for $NP$-hard problems.

The idea proposed in~\cite{rubinfeld1996robust,GGR98} was to relax the decision algorithm as follows:
a {\em testing algorithm} should {\em accept} an object that has the given property studied, and {\em reject} with high constant probability an object that is {\em $\epsilon$-far}
from having the property, that is, at least an $\epsilon$-fraction of the object should be modified so that it has the property.
For example, in case of the vertex coloring problem, at least an $\epsilon$-fraction of the edges of the graph should be removed to make it $k$-colorable.
This relaxation allows testing algorithms to be ultra-efficient in many cases in comparison with regular decision algorithms.
Specifically, in many cases their query complexity is sub-linear and even independent of the input size, but depends on $\epsilon$.
In the context of testing the rank of a matrix, the requirement from a testing algorithm is as follows:

\begin{definition}
A {\em testing algorithm} for a given rank function is given query access to an $n \times m$ matrix $M$, an error parameter $\epsilon$ and a parameter $d$.
The algorithm  should accept with probability at least $2/3$ if $M$ has rank at most $d$, and should reject with probability at least $2/3$ if $M$ is $\epsilon$-far
from every matrix of rank at most $d$, that is, more than $(n \cdot m )/\epsilon$ of the entries of $M$ should be modified so that it has rank at most $d$.

The {\em query complexity} of the algorithm is the number of entries of $M$ which it inspects.
The algorithm is {\em non-adaptive} if it determines its queries based only on the input and the random coin tosses, independently of the answers provided to previous queries.
Otherwise, the algorithm is {\em adaptive}.
\end{definition}

Although the real rank can be computed in polynomial time, Krauthgamer and Sasson~\cite{krauthgamer2003property} showed
that there exists a non-adaptive property testing algorithm for the real rank whose query complexity is  $O(d^2 / \epsilon^2)$ and, thus, independent of the size of  $M$.
Li, Wang and Woodruff~\cite{Li} gave an adaptive algorithm for the real rank with a reduced query complexity of $O(d^2 / \epsilon)$.
Balcan, Woodruff and Zhang~\cite{BLWZ} gave a non-adaptive testing algorithm for the real rank with query complexity $\tilde{O}(d^2 / \epsilon)$.

Parnas, Ron and Shraibman~\cite{parnas2021property} provide a non-adaptive testing algorithm for the Boolean rank with query complexity $\tilde{O}(d^4/\epsilon^6)$.
For the binary rank they show a non-adaptive algorithm with query complexity $O(2^{2d} /\epsilon^2)$,  and an adaptive algorithm with query complexity $O(2^{2d} /\epsilon)$.
The query complexity of all three algorithms of~\cite{parnas2021property} is, thus, independent of the size of the matrix.
Bshouty~\cite{bshouty2023property} improved the result of~\cite{parnas2021property} for the binary rank by a factor of $\tilde{\Theta}(2^{d})$,
but the query complexity remains exponential in $d$. He also noted that a small adjustment to the analysis given by~\cite{parnas2021property}
for the Boolean rank results in an improved query complexity of $\tilde{\Theta}\left(d^4/ \epsilon^4\right)$.

The non-adaptive algorithms of~\cite{krauthgamer2003property} and~\cite{parnas2021property}
are a variant of the following simple algorithm, for different settings of the rank function and of the parameter $s$ in the algorithm.

\bigskip
\begin{center}
\fbox{
\begin{minipage}{5in}
\leftline{\bf Algorithm Test rank($M, d, s$)}
\centering
\begin{enumerate}
\item
Select uniformly, independently and at random $s$ entries $(i,j)$ of $M$.
\item
Let $S$ be the $s \times s$ sub-matrix of $M$  induced by these entries.
\item
If the rank of $S$ is at most $ d$ then accept. Otherwise, reject.
\end{enumerate}
\end{minipage}
}
\end{center}
\bigskip

This simple algorithm is similar in flavour to many property testing algorithms.
That is, the algorithm samples the object tested and checks if the sub-object induced by these samples has the desired property and if so, accepts.
For properties which are {\em hereditary} this algorithm always accepts objects having the property,
and the challenge is to prove that the algorithm rejects with probability at least $2/3$ every object which is $\epsilon$-far from having the property.
Usually one shows that if the object is $\epsilon$-far from having the property then it has many imperfections,
and thus, the sub-object defined by the sample will have such an imperfection with a high enough probability and the algorithm will catch this.
To prove this, the contrapositive is shown. That is, if the number of imperfections is bounded by some function which depends on the size of the object and $\epsilon$,
then it is possible to modify the object in at most an $\epsilon$-fraction of positions, and therefore, it is not $\epsilon$-far as claimed.

The analysis in~\cite{krauthgamer2003property} and~\cite{parnas2021property} considers the probabilistic process of choosing
the sub-matrix $S$ in Algorithm "Test rank" in a slightly different way, as done in many property testing algorithms.
The algorithm is viewed as if it chooses the sample iteratively.
Specifically, the algorithm starts with an empty sub-matrix, and in each iteration adds to the current sub-matrix a randomly chosen row and column of $M$,
until finally reaching a sub-matrix $S$ of size $s \times s$.
Then, it is shown that with  probability at least $2/3$, if $M$ is $\epsilon$-far from rank $d$, the final sub-matrix $S$ has rank strictly larger than $d$, and so the algorithm rejects.

Specifically, both results show that if $M$  is $\epsilon$-far from rank $d$ then in each iteration, with a high enough probability,
a row and a column are added to $M$ so that some progress is made.
Krauthgamer and Sasson~\cite{krauthgamer2003property} use the linearity of the real rank, and show that as long as the current sub-matrix has real rank less than $d$,
the algorithm has a high enough probability to add to the sub-matrix a row and/or column which increase
its rank by at least $1$, and thus, s = $O(d/\epsilon)$ iterations suffice for the algorithm to reject.
Parnas, Ron and Shraibman~\cite{parnas2021property} prove that with high enough probability,
 the number of distinct rows and/or columns of the sub-matrix increases by one in each iteration, and thus, after $s = O(2^d/\epsilon)$ iterations the algorithm rejects,
since a matrix with binary rank $d$ has at most $2^d$ distinct rows and at most $2^d$ distinct columns (see Section~\ref{Sec-kernel}).

We next describe how~\cite{krauthgamer2003property} and~\cite{parnas2021property} define
rows, columns and entries which are {\em useful} to the algorithm, in the sense that adding them to the sub-matrix results in the desired progress described above.
See Figure~\ref{fig:propertytesting} for an illustration of the following definitions and proofs.

\begin{figure}[htb!]
\begin{center}
$
A = \left(
  \begin{array}{ccc|c}
    1 & 0 & 0 & 1 \\
    0 & 1 & 0 & 2 \\
    0 & 0 & 1 & 3 \\\hline
    1 & 1 & 1 & {\bf 7} \\
  \end{array}
\right)
$
\;\;\;\;\;\;\;\;\;\;\;\;
$
B = \left(
  \begin{array}{ccc|c}
    1 & 0 & 0 & 1 \\
    0 & 1 & 0 & 0 \\
    0 & 0 & 1 & 0 \\\hline
    1 & 0 & 0 & {\bf 0} \\
  \end{array}
\right)
$
\end{center}
\caption{On the left is a matrix $A$ with real rank $4$. The principal sub-matrix $S_A$ defined by the first $3$ rows and columns of $A$ has real rank $3$.
The last row and last column of $A$ are not useful for $S_A$.
But $A_{4,4}$ is a useful corner entry for $S_A$, since the rank of $A$ increases when augmenting $S_A$ with both the last row and column.
However, if we modify $A_{4,4}$ to a $6$, the modified matrix has real rank $3$.
\newline
On the right is a matrix $B$ with binary rank $4$. All rows and columns are distinct. The principal sub-matrix $S_B$ defined by the first $3$ rows and columns of $B$ has binary rank $3$.
The last row of $B$ is identical to the first row of $S_B$, and the last column of $B$ is identical to the first column of $S_B$.
Entry $B_{4,4}$ is a useful corner entry for $S_B$ since $B_{4,4} \neq B_{1,1}$.
If we modify $B_{4,4}$ to be equal to $B_{1,1} = 1$, then the modified matrix has binary rank $3$.}
\label{fig:propertytesting}	
\end{figure}

\begin{definition}[\cite{krauthgamer2003property}]
A row (column) in $M$  is {\em useful} for $S$ if augmenting $S$ with it increases the real rank of $S$.
An entry $M_{i,j}$ is a {\em corner entry} for $S$ if both row $i$ and column $j$ are not in $S$.
A corner entry $M_{i,j}$  is  {\em useful} for $S$ if row $i$ and column $j$ are not useful for $S$,
but augmenting $S$ with both row $i$ and column $j$ increases its real rank.
\end{definition}

\begin{lemma}[\cite{krauthgamer2003property}]
\label{lem-real-testing}
Let $M$ be a $0,1$ matrix of size $n \times m$ and let $S$ be a sub-matrix of $M$ such that $\Rreal(S) \leq d$.
Suppose that the following three conditions hold:
\begin{enumerate}
\item At most $\epsilon \cdot n/3$ rows in $M$ are useful for $S$.
\item At most $\epsilon \cdot m/3$ columns in $M$ are useful for $S$.
\item At most $\epsilon \cdot n\cdot m/3$ corner entries in $M$ are useful for $S$.
\end{enumerate}
Then $M$ is $\epsilon$-close to having real rank at most $d$.
\end{lemma}
\begin{proof}
We show that if all the three conditions in the lemma hold, it is possible to modify at most $\epsilon \cdot n \cdot m$ entries in $M$
so that the resulting modified matrix $M'$ has real rank at most $d$.

First modify to zero all rows and columns which are useful for $S$.
Next, consider a corner entry $M_{i,j}$ which is useful for $S$.
By definition, column $i$ is not useful, and therefore, any dependence between the rows of $S$ is preserved when $S$ is extended with column $j$  (see Lemma~\ref{sumofrowaugmented}).
Furthermore, row $i$ is also not useful, and thus, it is some linear combination of the rows of $S$.
Therefore, we can change only  entry $M_{i,j}$, so that when $S$ is extended with both $i$ and $j$ its rank does not increase.

Hence, the total number of entries modified is at most $\epsilon \cdot n \cdot m$.
Any row or column which are not useful for $S$, were not modified, and are thus, some linear combination of the rows or columns of $S$,
and all other rows and columns were modified so that they are some linear combination of the rows of $S$.
Since the augmentation property holds for the real rank (see Section~\ref{Sec-base}), then augmenting $S$ with all these rows and columns does not change its rank,
and therefore, $\Rreal(M') = \Rreal(S) \leq d$.
\end{proof}

The proof of Lemma~\ref{lem-real-testing} relies on Lemma~\ref{sumofrowaugmented} and on the augmentation property which holds for the real rank, but does not
hold in general for the Boolean and binary rank (see Section~\ref{Sec-base}).
Therefore, the algorithm of~\cite{parnas2021property} for the binary rank uses the following definition for useful rows, columns and corner entries.

\begin{definition}[\cite{parnas2021property}]
\label{def-prs}
A row (column) of $M$ is {\em useful} for $S$ if it is different from all rows (columns) in $S$.
A corner entry $M_{i,j}$ is  {\em useful} for $S$ if there exists a row $k$ and a column $\ell$  in $S$ which are identical to row $i$ and to column $j$, respectively,
such that $M_{i,j} \neq M_{k,\ell}$.
\end{definition}

\begin{lemma}[\cite{parnas2021property}]
\label{lem-binary-testing}
Let $M$ be a $0,1$ matrix of size $n \times m$ and let $S$ be a sub-matrix of $M$ such that $\Rbin(S) \leq d$.
Suppose that the three conditions mentioned in Lemma~\ref{lem-real-testing} hold,
where useful rows, columns and corner entries are as in Definition~\ref{def-prs}.
Then $M$ is $\epsilon$-close to having binary rank at most $d$.
\end{lemma}
\begin{proof}
We show that if all three conditions hold, then it is possible to modify at most $\epsilon \cdot n \cdot m$ entries in $M$
so that the resulting matrix $M'$ has binary rank at most $d$.

First modify to zero all rows and columns which are useful in $M$ for $S$.
Next, consider a useful corner entry $M_{i,j}$. By definition, there exits a row $k$ and a column $\ell$  in $S$,
such that row $k$ is identical to row $i$, and column $\ell$ is identical to column $j$, and $M_{i,j} \neq M_{k,\ell}$.
Set $M'_{i,j} =  M_{k,\ell}$.

Hence, we modified at most $\epsilon \cdot n \cdot m$ entries, where all rows and columns which were not modified are not useful for $S$,
and, therefore, identical to some row or column in $S$.
Now it is possible to prove that a partition of $S$ into at most $d$ monochromatic rectangles induces a partition of the modified matrix $M'$ into at most $d$ rectangles.
We refer the interested reader to the full proof in~\cite{parnas2021property}.
\end{proof}

Using Lemma~\ref{lem-real-testing} and Lemma~\ref{lem-binary-testing} it is possible to prove that if $M$ is $\epsilon$-far from a matrix with real or binary rank $d$,
then in each iteration the algorithm has a high enough probability of adding a useful row or column or corner entry to the current sub-matrix so that progress is made as described previously.
That is, either the real rank increases or the number of distinct rows and/or columns increases, depending on the rank in question.
The following two theorems follow.

\begin{theorem}[\cite{krauthgamer2003property}]
Algorithm "Test rank" is a non-adaptive testing algorithm for the real rank for $s = O(d/\epsilon)$ and its query complexity is $O(d^2 /\epsilon^2)$.
\end{theorem}

\begin{theorem}[\cite{parnas2021property}]
Algorithm "Test rank" is a non-adaptive testing algorithm for the binary rank for $s = O(2^d/\epsilon)$ and its query complexity is $O(2^{2d} /\epsilon^2)$.
\end{theorem}

Bshouty~\cite{bshouty2023property} gives an adaptive algorithm with a reduced query complexity for the binary rank by using
Lemma~\ref{lem-product-rows-cols}, which provides a tighter bound of $(d+1)2^d$ on the product of the number of distinct rows and columns of a $0,1$ matrix with binary rank $d$.
The basic idea behind his algorithm is to maintain a sub-matrix that is not necessarily square, as is the case in the algorithms of~\cite{krauthgamer2003property}
and~\cite{parnas2021property}, and furthermore, the algorithm makes sure that this sub-matrix contains only distinct rows and columns.

We conclude this subsection by noting that Lemma~\ref{lem-binary-fullrank} in Section~\ref{Subsec-fullrank}
presents an example of a matrix $M$ with binary rank larger than $d$, such that every sub-matrix of $M$ has binary rank at most $d$.
This implies that any algorithm testing for binary rank at most $d$ which simply samples a small sub-matrix of $M$ will accept $M$.
But, this does not imply that there is no efficient algorithm for testing the binary rank,
since this matrix $M$ is  {\em close} to having binary rank at most $d$, that is,
only a small fraction of the entries of $M$ have to be modified so that it has rank at most $d$.
Hence, in this case the testing algorithm is allowed by definition to accept $M$.

The question is if there exist small witnesses which can be sampled efficiently in a matrix that is {\em $\epsilon$-far} from having  binary rank $d$.
Of course, it is possible that there is no efficient algorithm for the binary rank, that is, there is no  algorithm whose query complexity is at most polynomial in $d$ and $\epsilon$.
Thus, it is interesting to try and find lower bounds for the query complexity of testing the binary rank.
See~\cite{goldreich2017introduction} for the approaches used previously in the field of property testing for proving lower bounds.

\subsection{Boolean matrix factorization}
\label{subsec-BMF}

Although computing the Boolean and binary rank is $NP$- hard,
many applications require to find an approximate factorization of a $0,1$ matrix $M$ into a product of low rank $0,1$ matrices.
This problem is known as the {\em Boolean matrix factorization} problem, sometimes denoted by BMF.
This topic is mentioned here only briefly as it is outside the scope of this survey, and we refer the interested reader to a survey by
Miettinen and Neuman~\cite{miettinen2021recent} for a description of the various applications and algorithms known for the BMF problem.

Formally, given a $0,1$ matrix $M$ of size $n \times m$ and an integer parameter $k$,
the goal is to find $0,1$ matrices $A,B$, where $A$ is of size $n \times k$ and $B$ is of size $k \times m$,
such that the product $A \cdot B$ approximates $M$ as well as possible. Note that the parameter $k$ does not have to be the rank of the matrix $M$.
Usually the goal is to minimize the Forbenius norm defined by
 $$
 \| M- A\cdot B\|_{F}^2 = \sum_{i=1}^n \sum_{j = 1}^m | M_{i,j} - (A\cdot B)_{i,j}|^2
 $$
where the summations in the definition of the norm are over the real numbers.
As to the product $A\cdot B$, several variations were considered in the literature, and these variations influence the complexity of the resulting approximation algorithms.
The main variations researched so far are computing $A\cdot B$ using the Boolean operations (i.e. $1+1 = 1$),
or over the integers (i.e. $1 + 1 = 2$), and in some cases over the binary field $GF(2)$ (that is, $1 + 1 = 0$).

Kumar, Panigrahy, Rahimi, and Woodruff~\cite{kumar2019faster} give  a constat factor approximation algorithm for the BMF problem with operations
over the integers and running time  $2^{O(k^2\log k)}poly(m \cdot n)$.
A  result by Velingker,  V\"{o}tsch, Woodruff, and  Zhou~\cite{velingker2023fast} from 2023 gives a $(1 + \epsilon)$-approximation for the problem over the integers
in time $2^{\tilde{O}(k^2/\epsilon^4)}poly(n,m)$, for any $\epsilon > 0$.
In both results of~\cite{kumar2019faster,velingker2023fast}  the Forbenius norm is used, and both papers give also algorithms for the case of operations over $GF(2)$.
These algorithms are of course not practical unless the parameter $k$ is very small.
On the other hand, there are other more practical algorithms, but these algorithms do not have a theoretical running time guarantee.
Such algorithms are based on combinatorial optimization and use various heuristics and iterative update approaches to overcome the difficulty of the problem
(see~\cite{miettinen2021recent}).

\section{Conclusion and Further Directions of Research}
\label{sec-discussion}

In this survey we presented a comprehensive overview of the known results concerning the binary and Boolean rank, from both a mathematical and a computational perspective,
with particular emphasis on their relationship to the real rank.
These rank functions were studied over the years by both the mathematical community, in fields such as algebra, combinatorics and graph theory,
and by the computer science community, mainly in the context of communication complexity, but also from an algorithmic perspective.
This wide applicability gave rise to the development of a variety of tools, techniques and algorithms,
and  significant progress was made in the understanding of these rank functions, although many open problems remain.

Together, the mathematical, computational and algorithmic results presented in the survey outline the current theoretical knowledge in this area and suggest directions for further research.
We summarize the main open problems presented throughout the survey, as well as suggest further directions of research.
The first two open problems are related to the mathematical properties of these rank functions as discussed in Section~\ref{Sec-properties},
and specifically relate to the binary rank.

\begin{open}
\label{SmallWitness}
Find an example of a $0,1$ matrix $M$ of size $n \times n$ with binary rank $d < n$, such that every $d \times d$ submatrix has a significantly smaller binary rank.
\end{open}

\begin{open}
Are there $0,1$ matrices $A,B$ for which $\Rbin(A \otimes B) < \Rbin(A)\cdot \Rbin(B)$?
\end{open}

Another issue discussed is the maximal gap possible between the real, binary and Boolean rank for various families of matrices.
We saw that the Boolean rank can be exponentially smaller than the real and the binary rank and that this is the maximal gap possible.
On the other hand, we saw an example of a family of matrices with real rank at most $d$ and Boolean and binary rank at least $d^{\widetilde{\Omega}(\log d)}$.
Is this the optimal gap, or is there a family of matrices with a larger gap between the real and the binary rank?

\begin{open}
Is there a family of $0,1$ matrices with a larger gap between the real and the binary rank?
\end{open}

Some of the results presented in this survey employed the powerful lifting technique to compare the rank of a matrix and its complement, and using
this technique an almost optimal gap was proved.
Are there simpler combinatorial or algebraic proofs for these results that will shed light on the structure of the matrices which achieve this gap?
\begin{open}
Give alternative combinatorial or algebraic proofs to Theorems~\ref{theo-complement} and Theorem~\ref{theo-real-acomplement}.
\end{open}

Finally, as the Boolean and binary rank have also many applications, we suggest the following  algorithmic open problems.

\begin{open}
Find better approximation algorithms for the binary and Boolean rank, perhaps for specific families of matrices, as discussed in Section~\ref{SubSec-approx}.
\end{open}

\begin{open}
Find a property testing algorithm for the binary rank with polynomial query complexity in the rank $d$ and the error parameter $\epsilon$, or prove that no such algorithm exists.
\end{open}

We conclude this survey by mentioning again that although it focuses on the real, binary and Boolean rank,
it is interesting to study the relationship of other rank functions in comparison to the functions surveyed here,
as well as use the results and techniques presented here to study other rank functions.
For example, the  non-negative rank, $\Rnon(M)$, mentioned in the introduction is also an important rank function which was studied extensively.

As the Boolean and binary rank, the non-negative rank  is also $NP$-hard (see Vavasis~\cite{vavasis2010complexity} and Shitov~\cite{shitov2017nonnegative}).
By definition, its relation with the binary and real rank is:
$$\Rreal(M) \leq \Rnon(M) \leq \Rbin(M).$$
Watson~\cite{watson2016nonnegative} showed that the non-negative rank and the binary rank are equal when they are at most $3$,
and gave an example of a matrix $M$ of size $5 \times 6$ for which $\Rnon(M) = 4$ and  $\Rbin(M) = 5$.
For real rank at most $2$, the non-negative rank is always equal to the real rank.
Shitov~\cite{shitov2014upper} proved that if $M$ is a non-negative matrix of size $n \times n$, $n > 6$,
and $\Rreal(M) = 3$, then $\Rnon(M) \leq 6n/7$, and showed that
there exist non-negative matrices of size $n \times n$ with real rank $3$ and non-negative rank $n$, for $n = 3,4,5,6$.

As the binary and Boolean rank, the non-negative rank does not have the full-rank sub-matrix property (see Section~\ref{Subsec-fullrank}, Theorem~\ref{theo-moitra}),
and it was also studied extensively in the context of communication complexity (see for example,~\cite{Yannakakis91,kol2014approximate, faenza2015extended}).

From an algorithmic point of view, the non-negative matrix factorization problem (NMF) is an
important practical problem, as the Boolean matrix factorization problem (BMF) mentioned in Section~\ref{subsec-BMF},
and it was studied extensively in this context.
See a survey by Wang and Zhang~\cite{wang2012nonnegative} and a book by Gillis~\cite{gillis2020nonnegative}.

Together, these results and the results presented in this survey, illustrate the depth of the connections among different rank functions,
the various techniques used from both mathematics and computer science, and of course, the many open questions which remain unsolved.
Further exploration of these connections and the relationship among rank functions remains a rich and challenging direction for future research, and
promises to yield new insights across areas of mathematics and theoretical computer science, as we encountered in this survey.

\section*{Acknowledgements}
I would like to thank Ishay Haviv for telling me about the approximation algorithm for the binary rank described in Section~\ref{SubSec-approx},
and thank Dana Ron for helpful comments on Section~\ref{SubSec-property} which describes property testing.

The latex code of some of the figures was written with the assistance of ChatGpt, and then the code was corrected by me when needed.




\bibliographystyle{plain}
\bibliography{isf2024}

\end{document}